\documentclass[journal,letterpaper]{IEEEtran}

\usepackage{amsmath,amssymb,amsthm}
\usepackage{xifthen}
\usepackage{mathtools}
\usepackage{enumerate}
\usepackage{microtype}
\usepackage{xspace}
\usepackage{bm}
\usepackage[T1]{fontenc}
\usepackage{fancyhdr}
\usepackage{lastpage}
\usepackage{bbm}

\newcommand{\mbs}[1]{\pmb{#1}}
\newcommand{\vect}[1]{{\lowercase{\mbs{#1}}}}
\newcommand{\mat}[1]{{\uppercase{\mbs{#1}}}}

\newcommand{\T}{{\scriptscriptstyle\mathsf{T}}}
\renewcommand{\H}{{\scriptscriptstyle\mathsf{H}}}

\newcommand{\cond}{\,\vert\,}
\renewcommand{\Re}[1][]{\ifthenelse{\isempty{#1}}{\operatorname{Re}}{\operatorname{Re}\left(#1\right)}}
\renewcommand{\Im}[1][]{\ifthenelse{\isempty{#1}}{\operatorname{Im}}{\operatorname{Im}\left(#1\right)}}

\newcommand{\SNR}{\mathsf{snr}}

\newcommand{\av}{\vect{a}}

\newcommand{\cv}{\vect{c}}

\newcommand{\ev}{\vect{e}}

\newcommand{\gv}{\vect{g}}
\newcommand{\hv}{\vect{h}}

\newcommand{\qv}{\vect{q}}
\newcommand{\rv}{\vect{r}}

\newcommand{\tv}{\vect{t}}
\newcommand{\uv}{\vect{u}}
\newcommand{\vv}{\vect{v}}

\newcommand{\xv}{\vect{x}}
\newcommand{\yv}{\vect{y}}
\newcommand{\zv}{\vect{z}}

\newcommand{\alphav}{\vect{\alpha}}
\newcommand{\betav}{\vect{\beta}}
\newcommand{\zetav}{\vect{\zeta}}

\newcommand{\Sigmam}{\pmb{\Sigma}}

\newcommand{\Em}{\mat{e}}

\newcommand{\Mm}{\mat{M}}

\newcommand{\Qm}{\mat{q}}

\newcommand{\Xm}{\mat{x}}
\newcommand{\Ym}{\mat{y}}

\newcommand{\Bc}{{\mathcal B}}
\newcommand{\Cc}{{\mathcal C}}
\newcommand{\Dc}{{\mathcal D}}

\newcommand{\Mc}{{\mathcal M}}
\newcommand{\Nc}{{\mathcal N}}

\newcommand{\Qc}{{\mathcal Q}}
\newcommand{\Rc}{{\mathcal R}}
\newcommand{\Sc}{{\mathcal S}}

\newcommand{\CC}{\mathbb{C}}

\newcommand{\RR}{\mathbb{R}}

\newcommand{\Id}{\mat{\mathrm{I}}}

\newcommand{\CN}[1][]{\ifthenelse{\isempty{#1}}{\mathcal{N}_{\mathbb{C}}}{\mathcal{N}_{\mathbb{C}}\left(#1\right)}}

\renewcommand{\P}[1][]{\ifthenelse{\isempty{#1}}{\mathbb{P}}{\mathbb{P}\left(#1\right)}}
\newcommand{\E}[1][]{\ifthenelse{\isempty{#1}}{\mathbb{E}}{\mathbb{E}\left[#1\right]}}
\newcommand{\I}[1][]{\ifthenelse{\isempty{#1}}{\mathbb{I}}{\mathbb{I}\left\{#1\right\}}}
\renewcommand{\det}[1][]{\ifthenelse{\isempty{#1}}{\mathrm{det}}{\text{det}\left(#1\right)}}
\newcommand{\trace}[1][]{\ifthenelse{\isempty{#1}}{\mathrm{tr}}{\text{tr}\left(#1\right)}}
\newcommand{\rank}[1][]{\ifthenelse{\isempty{#1}}{\mathrm{rank}}{\text{rank}\left(#1\right)}}
\newcommand{\diag}[1][]{\ifthenelse{\isempty{#1}}{\mathrm{diag}}{\text{diag}\left(#1\right)}}
\newcommand{\Cov}[1][]{\ifthenelse{\isempty{#1}}{\mathsf{Cov}}{\mathsf{Cov}\left(#1\right)}}

\DeclareMathOperator{\arccot}{arccot}

\newcommand{\defeq}{\triangleq}
\newcommand{\eqdef}{\triangleq}

\newtheorem{proposition}{Proposition}
\newtheorem{property}{Property}
\newtheorem{definition}{Definition}
\newtheorem{corollary}{Corollary}
\newtheorem{lemma}{Lemma}

\newcounter{enumi_saved}
\usepackage{answers}
\Newassociation{solution}{Solution}{solutionfile}

\AtBeginDocument{\Opensolutionfile{solutionfile}[\jobname]}
\AtEndDocument{\Closesolutionfile{solutionfile}\clearpage
}

\IfFileExists{MinionPro.sty}{
}{
}

\usepackage{subfigure,graphics}
\usepackage{multirow}
\usepackage{array}
\usepackage{amsmath,bbm}
\usepackage{color}
\usepackage{enumitem}
\usepackage{stfloats}
\usepackage[flushleft]{threeparttable}
\usepackage{upgreek}
\usepackage{pgfplots}
\pgfplotsset{minor grid style={dotted,gray!25}}
\pgfplotsset{major grid style={dashed,gray!25}}

\usepackage{grffile}
\pgfplotsset{compat=newest}
\usetikzlibrary{plotmarks}
\usetikzlibrary{arrows.meta}
\usepgfplotslibrary{patchplots}

\mathtoolsset{showonlyrefs}

\renewcommand{\rv}[1]{{\mathsf{#1}}}
\newcommand{\rvVec}[1]{\pmb{\mathsf{#1}}}
\newcommand{\rvMat}[1]{\pmb{\mathsf{#1}}}

\newcommand{\Span}[1]{{\text{Span}\left(#1\right)}}

\newcommand*\dif{\mathop{}\mathrm{d}}

\renewcommand{\defeq}{:=}
\renewcommand{\eqdef}{=:}
\renewcommand{\SNR}{{\rm SNR}}

\newcommand{\im}{{\jmath}}

\newcommand{\LLR}{{\rm LLR}}
\newcommand{\LLRp}{{\rm LLR}^{\rm pilot}}

\usepackage{tikz}
\usetikzlibrary{calc,fadings,decorations.pathreplacing}

\tikzset{%
	>=latex, 
	inner sep=0pt,%
	outer sep=2pt,%
	mark coordinate/.style={inner sep=0pt,outer sep=0pt,minimum size=3pt,
		fill=black,circle}%
}

\title{Cube-Split: A Structured Grassmannian Constellation for Non-Coherent SIMO Communications} 
  
\author{	
Khac-Hoang Ngo,~\IEEEmembership{Student Member,~IEEE,} Alexis Decurninge,~\IEEEmembership{Member,~IEEE,} Maxime Guillaud,~\IEEEmembership{Senior Member,~IEEE,} Sheng Yang,~\IEEEmembership{Member,~IEEE}%
\thanks{Khac-Hoang Ngo is with Mathematical and Algorithmic Sciences Laboratory, Paris Research Center, Huawei Technologies, 92100 Boulogne-Billancourt, France, and Laboratory of Signals and Systems, CentraleSup\'elec, University of Paris-Saclay, 91190 Gif-sur-Yvette, France.~(e-mail: \texttt{ngo.khac.hoang@huawei.com}) 		}%
\thanks{Alexis Decurninge and Maxime Guillaud are with Mathematical and Algorithmic Sciences Laboratory, Paris Research Center, Huawei Technologies, 92100 Boulogne-Billancourt, France~(e-mail: \texttt{\{alexis.decurninge, maxime.guillaud\}@huawei.com})}%
\thanks{Sheng Yang is with Laboratory of Signals and Systems, CentraleSup\'elec, University of Paris-Saclay, 91190 Gif-sur-Yvette, France.~(e-mail: \texttt{sheng.yang@centralesupelec.fr})}%
\thanks{The material in this paper was partially presented at the 51st Asilomar Conference on Signals,
	Systems, and Computers, CA, USA, 2017~\cite{hoangAsilomar2017cubesplit}.}%
}

\begin{document}

\maketitle
\begin{abstract} 
	In this paper, we propose a practical structured constellation for non-coherent communication with a single transmit antenna over Rayleigh flat and block fading channel without instantaneous channel state information. The constellation symbols belong to the Grassmannian of lines and are defined up to a complex scaling. The constellation is generated by partitioning the Grassmannian of lines into a collection of bent hypercubes and defining a mapping onto each of these bent hypercubes such that the resulting symbols are approximately uniformly distributed on the Grassmannian. With a reasonable choice of parameters, this so-called cube-split constellation has higher packing efficiency, represented by the minimum distance, than the existing structured constellations. Furthermore, exploiting the constellation structure, we propose low-complexity greedy symbol decoder and log-likelihood ratio computation, {as well as an efficient way to associate it to a multilevel code with multistage decoding}. Numerical results show that the performance of the cube-split constellation is close to that of a numerically optimized constellation and better than other structured constellations. It also outperforms a coherent pilot-based scheme in terms of error probability and achievable data rate {in the regime of short coherence time and large constellation size}.
\end{abstract}
\section{Introduction} \label{sec:intro}
In communication over fading channels, the knowledge of instantaneous channel state information (CSI) enables to adapt the transmission and reception to current channel conditions. The capacity of communication with {\em a priori} CSI at the receiver, so-called {\em coherent} communication, is well known to increase linearly with the minimum number of transmit and receive antennas~\cite{Telatar1999capacityMIMO,Foschini}. In practice, however, the channel coefficients are not granted for free prior to communication. They need to be estimated, typically by periodic transmission of reference (or pilot) symbols known to the receiver~\cite{Hassibi2003howmuchtraining}. 
If the channel state
is not stable (in, e.g., time or frequency domain), accurate channel estimation requires regular pilot transmissions which can occupy a disproportionate fraction of communication resources. In this case, the cost of channel estimation is
significant, and it might be beneficial to use a communication scheme that
does not rely on the knowledge of instantaneous CSI. Communication without {\em a priori} CSI is said to be {\em non-coherent}~\cite{ZhengTse2002Grassman}.

We consider non-coherent single-input multiple-output~(SIMO) communication in which a single-antenna transmitter transmits to an $N$-antenna receiver. We assume flat and block fading channel, i.e., the channel vector remains constant within each coherence block of $T$ channel uses with $T\ge2$ and changes independently between blocks. 
The non-coherent capacity of this channel with independent and identically distributed~(i.i.d.) Rayleigh fading was calculated for the high signal-to-noise ratio (SNR) regime as
\begin{equation} \label{eq:noncoherentCapacity}
	C(\SNR,N,T) = \Big(1-\frac{1}{T}\Big) \log_2\SNR + c(N,T) + o(1) 
\end{equation}
bits/channel use as $\SNR \to \infty$, where $c(N,T)$ (given in \eqref{eq:c(N,T)}) is a constant independent of the SNR~\cite{ZhengTse2002Grassman,Hochwald2000unitaryspacetime,Yang2013CapacityLargeMIMO}. The pre-log factor $1 - \frac{1}{T}$ can be achieved by a pilot-based scheme: the transmitter sends a pilot symbol in one of the channel uses and data symbols in the remaining $T-1$ channel uses of a coherence block; the receiver estimates the channel based on the known pilot symbol and performs coherent detection on the received data symbols based on the channel estimate~\cite{Hassibi2003howmuchtraining}. This approach, however, can only achieve a rate at a constant gap below the full capacity $C(\SNR,N,T)$
~\cite{ZhengTse2002Grassman}.

In~\cite{ZhengTse2002Grassman}, it was shown that the optimal strategy achieving the high-SNR capacity $C(\SNR,N,T)$ is to transmit isotropically distributed vectors on $\CC^T$ belonging to the Grassmannian of lines, which is the space of one-dimensional subspaces in $\CC^T$~\cite{boothby1986introduction}, and use the {span} of these vectors to carry information.\footnote{When $T< N+1$, a further condition for achieving the capacity is that the input norm square is beta distributed; the rate achieved with {\em constant-norm} isotropically distributed input approaches the capacity within a constant factor of $O(\frac{\log N}{T})$~\cite{Yang2013CapacityLargeMIMO}.} The intuition behind that result is that the random channel coefficients only scale the transmitted signal vector without changing its {span}. In other words, the transmitted vector $\rvVec{x}$ and the noise-free observation $\rv{h}_j\rvVec{x}$ at receive antenna $j \in [N]$ represent {\em identical} point on the Grassmannian. Thus, the constellation design for non-coherent communication can be formulated as sphere packing on the Grassmann manifold.\footnote{A Grassmannian constellation can also be used as a precoding codebook in limited-feedback communication systems~\cite{Love2008limited_feedback}. The set of the intersections of the Grassmannian constellation symbols and the unit sphere is also called an antipodal spherical code, which has applications in designing measurement matrix for compressive sensing~\cite{Conde2017fast_antipodal_spherical_codes}.} The ultimate packing criteria is to minimize the detection error under noisy observation. This typically amounts to maximizing the distance between the constellation points, for which the packing efficiency limits are derived in, e.g.,~\cite{Conway1996packing,Barg2002BoundsOP,Dai2008quantizationBounds}. A number of Grassmannian constellations have been proposed 
with different criteria, constellation generation, and decoding methods. They follow two main approaches. 

The first approach uses numerical optimization tools to solve the sphere-packing problem on the Grassmannian by maximizing the minimum symbol pairwise distance~\cite{Agrawal2001MIMOconstellations,Gohary2009GrassmanianConstellations,BekoTSP2007noncoherentColoredNoise,Tahir2019constructing_Grassmannian_frames} or directly minimizing the error probability upper bound~\cite{McCloudIT2002signalDesignAndConvCode,Wu2008USTM_based_on_Chernoff_bound}. This results in constellations with a good distance spectrum. However, due to the lack of structure, this kind of constellation needs to be stored at both the transmitter and receiver, and decoded with the high-complexity maximum-likelihood (ML) decoder, which limits practical use to only small constellations. 

The second approach imposes particular structure on the constellation based on, e.g., algebraic construction~\cite{Hochwald2000systematicDesignUSTM,TarokhIT2002existence,ZhaoTIT2004orthogonalDesign}, parameterized mappings of unitary matrices~\cite{Kammoun2007noncoherentCodes,JingTSP2003Cayley}, concatenation of phase shift keying~(PSK) and coherent space-time codes~\cite{Zhang2011full_diversity_blind_ST_block_codes}, or geometric motion on the Grassmannian~\cite{AttiahISIT2016systematicDesign}. The pilot-data structured input of a pilot-based scheme can also be seen as a non-coherent code~\cite{DayalIT2004leveraging}. The constellation structure facilitates low complexity constellation mapping and, probably, demapping.

In this work, we focus on the second approach due to its low-complexity advantage. 
{Let us briefly review the aforementioned structured Grassmannian constellations. The Fourier constellation~\cite{Hochwald2000systematicDesignUSTM} contains the rows of the discrete Fourier transform matrix. It coincides with the algebraically constructed constellation in~\cite[Sec.III-A]{TarokhIT2002existence} when $T =2$. Unfortunately, this design still requires numerical optimization of the Fourier frequencies and needs ML decoding. The exp-map constellation~\cite{Kammoun2007noncoherentCodes} is obtained by mapping each vector $\qv$ containing $T-1$ quadrature amplitude modulation~(QAM) symbols into a non-coherent symbol $\cv$ via the exponential map $\cv = \big[\cos(\gamma\|\qv\|) \ -\frac{\sin(\gamma\|\qv\|)}{\|\qv\|}\qv^\T\big]^\T$ with the homothetic factor $\gamma$ given in \cite[Eq.(19)]{Kammoun2007noncoherentCodes}. 
The coprime-PSK constellation~\cite{Zhang2011full_diversity_blind_ST_block_codes} has symbols of the form $\cv = [x\ y \ \zv^\T]^\T$ where $x$ and $y$ are respectively $Q_x$-PSK and $Q_y$-PSK symbols such that~(s.t.) $Q_x$ and $Q_y$  are coprime, and $\zv \in \CC^{(T-2)\times 1}$ belongs to a sub-constellation $\{\zv_1,\dots,\zv_{Q_z}\}$ s.t. $\zv_j$ and $\zv_l$ are linearly independent for any $j\ne l$. The drawback of this design is that a good choice for $\{\zv_1,\dots,\zv_{Q_z}\}$ is not specified and one might need to numerically generate it. 
The constellation in~\cite{AttiahISIT2016systematicDesign} has a multi-layer construction: starting from an initial constellation, the layer-$j$ symbols are generated by moving each previous-layer symbols along a set of $K$ geodesics. Specifically, given $\cv$, the new symbol $\bar{\cv}$ is generated as $\bar{\cv} = \cv v_k \cos(\phi_j) + \cv^\perp \sin(\phi_j)$ where $v_k$, $|v_k| = 1$, determines the $k$-th geodesic, $\sin(\phi_j)$ is the moving distance, and $\cv^\perp$ is orthogonal to $\cv$. According to~\cite[Thm.2]{AttiahISIT2016systematicDesign}, $\{v_k\}_{k=1}^K$ should be evenly spread on the unit circle. {However, a good choice for $\sin(\phi_j)$ is not known. Both the coprime-PSK constellation~\cite{Zhang2011full_diversity_blind_ST_block_codes} and the multi-layer constellation~\cite{AttiahISIT2016systematicDesign} require ML decoding.} The constellations in~\cite[Sec.III-B]{TarokhIT2002existence} and \cite{ZhaoTIT2004orthogonalDesign} are designed for the multi-transmit-antenna case only, while the constellation based on the Cayley transform~\cite{JingTSP2003Cayley} relies partly on numerical optimization of $T\times T$ matrices. The construction in \cite{Kim2010USTM_QOS} is possible only for constellations of size at most $T^2$.

In this paper, we introduce a novel fully structured Grassmannian constellation.} 
This constellation is structurally generated by partitioning the Grassmannian of lines with a collection of bent cubes and mapping the symbol's coordinates in the Euclidean space onto one of these bent cubes.\footnote{Our constellation was used in \cite{AlexisWCNCcubesplit} as a quantization codebook on the Grassmannian of lines. Although the constellation structure is similar, the labeling and log-likelihood ratio~(LLR) computation presented here do not appear in the quantization problem.} The main advantages of our so-called {\em cube-split constellation} are as follows:
\begin{itemize} 
	\item It has a good packing efficiency: its minimum distance is larger than that of existing structured constellation and compares well with the fundamental limits.
	\item It allows for a systematic decoder which has low complexity, hence can be easily implemented in practice, yet achieves near-ML performance.
	\item It admits a very simple yet effective binary labeling which leads to a low bit error rate.
	\item It allows for an accurate LLR approximation which can be efficiently computed, {and can be efficiently associated to a multilevel coding-multistage decoding~(MLC-MSD)~\cite{Wachsmann1999multilevelCoding} scheme.}
\end{itemize}
We verify by simulation that under i.i.d. Rayleigh block fading channel, in terms of error probability (with or without channel codes) and achievable data rate, our cube-split constellation achieves performance close to the numerically optimized constellation, and outperforms existing structured Grassmannian constellations and a (coherent) pilot-based scheme {in the regime of short coherence time and large constellation size}.

The remainder of this paper is organized as follows.  
The system model is presented and Grassmannian constellations are overviewed in Section~\ref{sec:model}. We describe the construction and labeling of our cube-split constellation in Section~\ref{sec:encoder}. We next propose low-complexity decoder and LLR computation, and a MLC-MSD scheme in Section~\ref{sec:decoder}. Numerical results on the error rates and achievable data rate are provided in Section~\ref{sec:performance}. Section~\ref{sec:conclusion} concludes the paper. We discuss the extension to the MIMO case and present the preliminaries and proofs in the appendices.

{\em Notation:} 
Random quantities are denoted with non-italic sans-serif letters, e.g., a scalar $\rv{x}$, a vector $\rvVec{v}$, and a matrix $\rvMat{M}$. 
Deterministic quantities are denoted 
with italic letters, e.g., a scalar $x$, a vector $\pmb{v}$, and a
matrix $\pmb{M}$. 
The Euclidean norm is denoted by $\|\vv\|$ and the Frobenius norm $\|\Mm\|_F$. 
The trace, conjugate, 
transpose, and conjugated transpose of $\Mm$ are denoted $\trace\{\Mm\}$, $\Mm^*$, 
$\Mm^\T$ and $\Mm^\H$, respectively. 
{$\log(\cdot)$ and $\log_2(\cdot)$ are respectively the natural and binary logarithms.} 
$\ev_j$ is the $T\times 1$ canonical basis vector with $1$ at position $j$ and $0$ elsewhere. 
{The inverse of a function $f(\cdot)$ is denoted $f^{-1}(\cdot)$. $[n] \defeq \{1,2,\dots,n\}$.}
The Grassmann manifold $G(\mathbb{K}^T,M)$ is the space of $M$-dimensional subspaces in $\mathbb{K}^T$ with $\mathbb{K} = \CC$ or $\mathbb{K} = \RR$~\cite{boothby1986introduction}. In particular, $G(\mathbb{K}^T,1)$ is the Grassmannian of lines. We use a unit-norm vector $\xv \in \CC^T$ ($\|\xv\| = 1$) to represent the set $\{\lambda \xv, \lambda \in \CC\}$, which is a point in $G(\CC^T,1)$. 
The chordal distance between two lines represented by $\xv_1$ and $\xv_2$ is $d(\xv_1,\xv_2) = \sqrt{1-|\xv_1^\H\xv_2|^2}$. {Given two functions $f(x)$ and $g(x)$, we write: $f(x) = O(g(x))$ if there exists constant $c>0$ and some $x_0$ s.t. $|f(x)| \le c |g(x)|, \forall x \ge x_0$;
$f(x) = o(g(x))$ if $\lim\limits_{x\to\infty}\frac{f(x)}{g(x)} = 0$.}

\section{System Model and Grassmannian Constellations} \label{sec:model}
\subsection{System Model}
We consider a SIMO non-coherent channel in which a single-antenna transmitter transmits to a receiver equipped with $N$ antennas. The channel between the transmitter and the receiver is assumed to be flat and block fading with coherence time $T$ symbol periods ($T\ge 2$). That is, the channel vector $\rvVec{h} \in \CC^{N}$ remains constant during each coherence block of $T$ symbols, and changes to an independent realization in the next block, and so on. The inter-block independence is relevant, e.g., in the context of sporadic transmission in which the interval between successive transmissions is indefinite. We assume that the distribution
of $\rvVec{h}$ is known, but its realizations are \textit{unknown} to both ends of the link. {We assume i.i.d. Rayleigh fading, i.e., $\rvVec{h} \sim \Cc\Nc(\mathbf{0},\Id_N)$.} 
Within a coherence block, the transmitter sends a signal $\rvVec{x} \in \CC^T$, and the receiver receives 
\begin{equation} \label{eq:channelModel}
\rvMat{Y} = \sqrt{\rho T} \rvVec{x} \rvVec{h}^\T + \rvMat{Z},
\end{equation}
where $\rvMat{Z} \in \CC^{T\times N}$ is the additive noise with i.i.d. $\Cc\Nc(0,1)$ entries independent of $\rvVec{h}$, and the block index is omitted for simplicity. We consider the power constraint 
$
\E[\|\rvVec{x}\|^2] = 1,
$
so that the transmit power $\rho$ is identified with the SNR at each receive antenna. From~\cite[Eq.(9)]{Yang2013CapacityLargeMIMO}, the high-SNR capacity $C(\rho,N,T)$ of this channel is given in~\eqref{eq:noncoherentCapacity} with $\SNR = \rho$ and
\begin{align} \label{eq:c(N,T)}
c(N,T) &= \frac{1}{T} \log_2 \frac{(\underline{L}-1)!}{(N-1)!(T-1)!} + \Big(1-\frac{1}{T}\Big)\log_2 T \notag\\
&\quad + \frac{\underline{L}}{T} \log_2\frac{N}{\underline{L}} + \frac{\overline{L}}{T}(\psi(N)-1){\log_2 e},
\end{align}
where $\underline{L} \defeq \min\{N,T-1\}$, $\overline{L} \defeq \max\{N,T-1\}$, and $\psi(\cdot)$ is Euler's digamma function~{\cite[Eq.(6.3.1)]{abramovitz}}. {Note that this generalizes and coincides with \cite[Eq.(24)]{ZhengTse2002Grassman} when $T\ge N+1$.}

We assume that the input $\rvVec{x}$ is taken from a finite constellation $\Cc$ of size $|\Cc|$. 
Given an observation $\rvMat{Y} = \Ym$, the ML decoder is
$
\hat{\xv}^{\rm ML} = \arg\max_{\xv \in \Cc} p_{\rvMat{Y}|\rvVec{x}}(\Ym|\xv).
$
Conditioned on $\rvVec{x} = \xv$, $\rvMat{Y}$  is a Gaussian matrix with independent columns having the same covariance matrix $\Id_T + \rho T \xv\xv^\H$, hence 
\begin{align}
p_{\rvMat{Y}|\rvVec{x}}(\Ym|\xv) &= \frac{\exp\left(-\trace\{\Ym^\H(\Id_T + \rho T \xv\xv^\H)^{-1}\Ym\}\right)}{\pi^T \det(\Id_T + \rho T \xv\xv^\H)} \\
&= \frac{\exp\big(-\|\Ym\|_F^2 + \frac{\rho T}{1+\rho T \|\xv\|^2}\|\Ym^\H \xv\|^2\big)}{\pi^T (1 + \rho T \|\xv\|^2)}. \label{eq:pdf}
\end{align}
Thus, for unit-norm input $\|\rvVec{x}\| = 1$, the ML decoder is simply
\begin{equation} \label{eq:MLoneUser}
\hat{\xv}^{\rm ML} = \arg\max_{\xv \in \Cc} \|\Ym^\H \xv\|^2.
\end{equation}

Assuming that all constellation symbols are equally likely to be transmitted, i.e., the input law is $p_{\rvVec{x}}(\xv) = \frac{1}{|\Cc|} \mathbbm{1}\{\xv\in \Cc\}$ where $\mathbbm{1}\{.\}$ is the indicator function, the achievable rate is 
\begin{multline}
R = \frac{1}{T}I(\rvVec{x};\rvMat{Y}) 
= \frac{1}{T} \E\bigg[\log_2 \frac{p_{\rvMat{Y}|\rvVec{x}}(\rvMat{Y}|\rvVec{x})}{\frac{1}{|\Cc|}\sum_{\cv \in \Cc} p_{\rvMat{Y}|\rvVec{x}}(\rvMat{Y} | \cv)}\bigg] \\
= \frac{\log_2|\Cc|}{T}-\frac{1}{T}\E\bigg[\log_2\frac{\sum_{\cv \in \Cc} p_{\rvMat{Y}|\rvVec{x}}(\rvMat{Y}| \cv)}{p_{\rvMat{Y}|\rvVec{x}}(\rvMat{Y}|\rvVec{x})}\bigg] 
\label{eq:achievableRate}
\end{multline}
bits/channel use. Here, $\frac{\log_2|\Cc|}{T}$ is the rate achievable in the noiseless case, and $\frac{1}{T}\E\Big[\log_2\frac{\sum_{\cv \in \Cc} p_{\rvMat{Y}|\rvVec{x}}(\rvMat{Y}| \cv)}{p_{\rvMat{Y}|\rvVec{x}}(\rvMat{Y}|\rvVec{x})}\Big]$ is the rate loss due to noise. {In the large constellation regime, the achievable rate converges to the channel capacity as the considered constellation gets close to the optimal constellation. Thus it can be used as a performance metric as done in \cite[Sec.V]{Hochwald2000systematicDesignUSTM}.} The expectation in \eqref{eq:achievableRate} does not have a closed form in general, and we resort to the Monte Carlo method to compute $R$.

\subsection{Grassmannian Constellations} \label{sec:grassmann_const}
It was shown that the high-SNR capacity~\eqref{eq:noncoherentCapacity} is achieved with isotropically distributed input $\rvVec{x}$ s.t. its distribution is invariant under rotation, i.e., $p_{\rvVec{x}}(\xv) = p_{\rvVec{x}}(\pmb{Q}\xv)$ for any $T\times T$ deterministic unitary matrix $\pmb{Q}$~\cite{ZhengTse2002Grassman,Yang2013CapacityLargeMIMO}. The span of such $\rvVec{x}$ is uniformly distributed on the Grassmannian of lines $G(\CC^T,1)$~\cite{boothby1986introduction}. 
Motivated by this, the constellation $\Cc$ can be designed by choosing $|\Cc|$ elements of $G(\CC^T,1)$, represented by $|\Cc|$ unit-norm vectors $\{\cv_1,\dots,\cv_{|\Cc|}\}$. {Both $\rvVec{x}$ and the noise-free observation $\rv{h}_j\rvVec{x}$ at receive antenna $j\in [N]$ belong to the set $\{\lambda\rvVec{x}, \lambda \in \CC \}$. Thus, by definition, $\rvVec{x}$ and $\rv{h}_j\rvVec{x}$ represent the same symbol in $G(\CC^T,1)$.} Therefore, Grassmannian signaling guarantees error-free detection {\em without} CSI in the noiseless case if the symbols are not colinear. 
When the noise $\rvMat{Z}$ is present, since its columns are almost surely not aligned with the signal $\rvVec{x}$, the {one-dimensional span} of the received signal at each receive antenna deviates from that of $\rvVec{x}$ with respect to (w.r.t.) a distance measure,\footnote{There are several choices for the distance measure between subspaces, such as chordal distance, spectral distance, Fubini-Study distance, geodesic distance (see, for example, 
\cite[Sec.~I]{Barg2002BoundsOP}). For the Grassmannian of lines, these distances are equivalent up to a monotonically increasing transformation. In this paper, we adopt the commonly used chordal distance.} leading to a detection error if $\rvMat{Y}$ is outside the decision region of the transmitted symbol. With the chordal distance $d(\xv,\yv) \defeq \sqrt{1-|\xv^\H \yv|^2}$,  the decision regions of the optimal ML decoder~\eqref{eq:MLoneUser} in the case $N=1$ are the Voronoi regions defined for symbol $\cv_j$, $j\in[|\Cc|]$, as 
\begin{equation} \label{eq:VoronoiRegions}
V_j = \{\xv \in G(\CC^T,1): d(\xv,\cv_j) \le d(\xv,\cv_l), \forall l \ne j\}.
\end{equation}
The constellation $\Cc$ should be designed so as to minimize the probability of decoding error. 


Following the footsteps of \cite{Hochwald2000unitaryspacetime}, we can derive the pairwise error probability~(PEP) of mistaking a symbol $\cv_j$ for another symbol $\cv_l$ of the ML decoder as
\begin{align}
P_{j,l}^{\rm ML} &= \Pr\left\{ \|\rvMat{Y}^\H \cv_l\|^2 > \|\rvMat{Y}^\H \cv_j\|^2 \big| \ \rvVec{x} = \cv_j\right\} \\
&= \frac{1}{2} \Bigg[1-\bigg(1+\frac{4(1+\rho T)}{(d(\cv_j,\cv_l)\rho T)^2}\bigg)^{-1/2}\Bigg].
\end{align}
We can verify that the PEP is decreasing with the chordal distance.
 The error probability $P_e^{\rm ML}$ of ML decoder can be upper bounded in terms of the PEP using the union bound as
\begin{multline} 
P_e^{\rm ML} = \frac{1}{|\Cc|} \sum_{j=1}^{|\Cc|} \Pr\left\{\hat{\rvVec{x}} \ne \rvVec{x}| \rvVec{x} = \cv_j\right\} 
\le \frac{1}{|\Cc|} \sum_{j=1}^{|\Cc|}\sum_{l\ne j} P_{j,l}^{\rm ML} \\
\le \frac{|\Cc|-1}{2} \Bigg[1-\bigg(1+\frac{4(1+\rho T)}{(d_{\rm min}\rho T)^2}\bigg)^{-1/2}\Bigg], \label{eq:unionbound}
\end{multline}
where $d_{\rm min} \defeq \displaystyle\min_{1\le j<l \le |\Cc|} d(\cv_j,\cv_l)$ is the minimum pairwise chordal distance of the constellation. Therefore, maximizing the minimum pairwise distance minimizes the union bound. This leads to a commonly used constellation design criteria
\begin{align} \label{eq:designCriteria}
\max_{\Cc = \{\cv_1,\dots,\cv_{|\Cc|}\}} \, \min_{1\le j<l \le |\Cc|} d(\cv_j,\cv_l).
\end{align}
This optimization problem can be solved numerically. The resulting constellation, however, is hard to exploit in practice due to its lack of structure. In particular, the unstructured constellations are normally used with the high-complexity ML decoder, do not admit a straightforward binary labeling, and need to be stored at both ends of the channel.  
In our design, we would rather relax (slightly) the optimality requirement~\eqref{eq:designCriteria} to have a structured constellation while preserving good packing properties. We describe our proposed Grassmannian constellation in the next section. 
 
\section{Cube-Split Constellation} \label{sec:encoder} 
\newcommand{\cell}{S}
{
\subsection{Design Approach}
\subsubsection{Partitioning of the Grassmannian}
We consider a set of $V$ Grassmannian points $\{\zetav_1,\dots, \zetav_V\}$ and partition the Grassmannian into $V$ cells whereby cell $i$ is defined as 
$$
\cell_i \defeq \big\{\xv\in G(\CC^T,1): d(\xv, \zetav_i)<d(\xv, \zetav_j), \forall j \in [V] \setminus \{i\}\big\}.
$$
We ignore the cell boundaries for which $d(\xv,\zetav_i) = d(\xv, \zetav_j) \le d(\xv, \zetav_k)$ for some $i \ne j$ and any $k \notin \{i,j\}$ since this is a set of measure zero. In this way, a symbol $\xv$ belongs to cell $\cell_{i}$ if 
$
i = \arg\min_{j\in[T]} d(\xv,\zetav_j), 
$
that is, $\zetav_i$ is the closest point to $\xv$ among $\{\zetav_1,\dots, \zetav_V\}$. Thus, these cells correspond to the Voronoi regions associated to the initial set of points $\{\zetav_1,\dots,\zetav_V\}$.
 
\subsubsection{Mapping from the Euclidean Space onto a Cell}
Since the symbols are defined up to a complex scaling factor, $G(\CC^T,1)$ has $T-1$ complex dimensions, i.e. $2(T-1)$ real dimensions, and so do the cells. Therefore, any point on a cell can be parameterized by $2(T-1)$ real coefficients. We choose to define these coefficients in the Euclidean space of $2(T-1)$ real dimensions, and let them determine the symbol through a bijective mapping $\gv_{i}(.)$ from this Euclidean space onto the cell $\cell_{i}$. 
On the other hand, {to define a grid in the Euclidean space}, let $A_j$ denote finite sets of $2^{B_j}$ points regularly spread on the interval $(0,1)$ in order to maximize the minimum distance within the set ($B_j$ then denotes the number of bits necessary to characterize a point in $A_j$) so that
\begin{equation} \label{eq:coor_set}
A_j = \bigg\{\frac{1}{2^{B_j+1}},\frac{3}{2^{B_j+1}},\dots,\frac{2^{B_j+1}\!-\!1}{2^{B_j+1}}\bigg\}, ~ j\in [2(T\!-\!1)].
\end{equation}
The Cartesian product of these sets $\bigotimes_{j=1}^{2(T-1)} A_j$ is then a grid in the Cartesian product of $2(T-1)$ intervals $(0,1)$ that we denote by $(0,1)^{2(T-1)}$.

The constellation can be formally described as the collection of the mapping of this grid onto each cell $\cell_1,\dots,\cell_V$, i.e.
\begin{align} \label{eq:const_construct}
\Cc = \bigg\{\cv = \gv_{i}(\av): i \in[V],\ \av \in \bigotimes_{j=1}^{2(T-1)} A_j \bigg\}.
\end{align}
Hence, the constellation can be seen as the collection of $V$ deformed lattice constellations containing in total $V\prod_{i=j}^{2(T-1)} 2^{B_j}$ symbols.
Note that for asymptotically large $B_j$, choosing the points in $A_j$ with equal probability defines a distribution that converges weakly to the uniform distribution on $(0,1)$. {Therefore, in this regime, the distribution of the points on the Euclidean grid $\bigotimes_{j=1}^{2(T-1)} A_j$ 
converges to a continuous uniform distribution in $(0,1)^{2(T-1)}$. On the other hand, as mentioned in Section~\ref{sec:grassmann_const}, an optimal constellation has symbols uniformly distributed on the Grassmannian. Therefore, focusing on this asymptotic regime (to support large constellations)}, we design $\{\zetav_i\}_{i=1}^V$ and $\{\gv_{i}\}_{i=1}^V$ s.t. for any $i$, the image of the uniform distribution in $(0,1)^{2(T-1)}$ through $\gv_{i}$ is uniformly distributed in $\cell_{i}$, i.e. we seek mappings satisfying the following property.
\begin{property} \label{prop:uniform}
	Let $\rvVec{a}$ be a random vector uniformly distributed on $(0,1)^{2(T-1)}$, then for any $i\in[V]$,  $\gv_{i}(\rvVec{a})$ is uniformly distributed on the cell $\cell_{i}$ of $G(\CC^T,1)$.
\end{property}

\subsection{Constellation Specifications}
We first choose the cell centers to be the canonical basis vectors $\{\ev_1,\dots,\ev_T\}$, so $V=T$. 
With this choice, the cells $\cell_1,\dots,\cell_T$ are easily characterized since $\cell_i$ can be defined in terms of the symbol's coordinates as
$
\cell_{i} \defeq \big\{\xv\in G(\CC^T,1): \ |x_{i}|>|x_j|, \forall j \in [T] \setminus \{i\}\big\}.
$
{However, the direct manipulation of the Grassmannian cell $\cell_{i}$ is not straightforward. Hence, we equivalently define the mapping $\gv_{i}$ through a mapping $\boldsymbol{\xi}_{T-1} : \av \mapsto \tv=[t_1 \dots t_{T-1}]^\T$ as}
\begin{equation} 
\gv_{i}(\av) 
= \frac{1}{\sqrt{1+\sum_{j=1}^{T-1}|t_j|^2}}[t_1 \ \dots \ t_{i-1} \ 1 \ t_{i} \ \dots \ t_{T-1}]^\T. 
\label{eq:encodeMapping}
\end{equation}
The rationale of this definition is that the target space of  $\boldsymbol{\xi}_{T-1}$ is then an Euclidean space contrary to $\gv_{i}$. 
Moreover, looking at the definition of $\cell_{i}$, the fact that $\gv_{i}$ takes values in $\cell_{i}$ is equivalent to the fact that $\boldsymbol{\xi}_{T-1}$ takes values in $\Dc(0,1)^{T-1}$ with $\Dc(0,1) \defeq\{z\in \CC:|z| \le 1\}$.

Let us first define $\boldsymbol{\xi}_1$ in the case $T = 2$.
In this case, $\rvVec{t}$ is a complex scalar. According to Lemma~\ref{lem:Cauchy-Grassmann}, Property~\ref{prop:uniform} is equivalent to $\boldsymbol{\xi}_1(\av)$ being distributed as a ${\rm Cauchy}(0,1)$ truncated on the unit disc $\Dc(0,1)$.
In that case, we show in Lemma~\ref{lem:Cauchy-Unif} that Property~\ref{prop:uniform} is satisfied by considering
\begin{align} \label{eq:unif2cauchy}
\boldsymbol{\xi}_1(\av) \defeq \sqrt{\frac{1-\exp(-\frac{|w|^2}{2})}{1+\exp(-\frac{|w|^2}{2})}} \frac{w}{|w|} 
\end{align} 
with ${w} = \Nc^{-1}({a}_1)+\im\Nc^{-1}({a}_2)$.

When $T > 2$, for $\rvVec{x}$ to be uniformly distributed on $\cell_{i}$, the distribution underlying $\rvVec{t}$ has dependencies which are hard to characterized. We choose to ignore these dependencies and design $\boldsymbol{\xi}_{T-1}$ in a way similar to the case $T=2$. To this end, we apply $\boldsymbol{\xi}_{1}$ to $\av$ in a component-wise manner
\begin{equation} \label{eq:encodeMappingT}
\boldsymbol{\xi}_{T-1}(\av) 
= \big[ \boldsymbol{\xi}_{1}([a_{1} \ a_{2}]^\T) \ \dots \  \boldsymbol{\xi}_{1}([a_{2T-3} \ a_{2T-2}]^\T)\big]^\T.
\end{equation}
However, note that Property~\ref{prop:uniform} is then not satisfied in that case.
Combining \eqref{eq:encodeMapping}, \eqref{eq:unif2cauchy} and \eqref{eq:encodeMappingT} completely defines $\gv_{i}$. This mapping is bijective and its inverse $\av = \gv_{i}^{-1}(\xv): \cell_{i} \to  (0,1)^{2(T-1)}$ is given by 
$
a_{2j-1} = \Nc(\Re(w_j))$ and $a_{2j} = \Nc(\Im(w_j)), j \in  [T-1]$, where 
\begin{align} 
w_j &= \sqrt{2\log\frac{1+|t_j|^2}{1-|t_j|^2}} \frac{t_j}{|t_j|}, \label{eq:w} \\
\tv &= \left[\frac{x_1}{x_{i}}, \dots, \frac{x_{i-1}}{x_{{i}}}, \frac{x_{{i}+1}}{x_{{i}}}, \dots, \frac{x_T}{x_{{i}}}\right]^\T. \label{eq:t}
\end{align}

The constellation is then constructed as in \eqref{eq:const_construct} where the coordinate sets $\{A_j\}$ is given in \eqref{eq:coor_set}. 
The grid of symbols in each cell is analogous to a bent hypercube, hence the name {\em cube-split constellation}. The constellation contains $T \prod_{j=1}^{2(T-1)} 2^{B_j}$ symbols. An example of the grid of points in $(0,1)^{2(T-1)}$ is shown in Fig.~\ref{fig:cubesplit}{(a)} and the resulting cube-split constellation is illustrated in Fig.~\ref{fig:cubesplit}{(b)}. For the sake of representation, we plot the constellation built on the {\em real} Grassmannian $G(\RR^T,1)$ following the same principle. 
\begin{figure}[!h] 
	\centering
	\subfigure[The grid of points on $(0,1)^2$]{
		\includegraphics[width=.25\textwidth,trim=.5cm .5cm .5cm .5cm,clip=true]{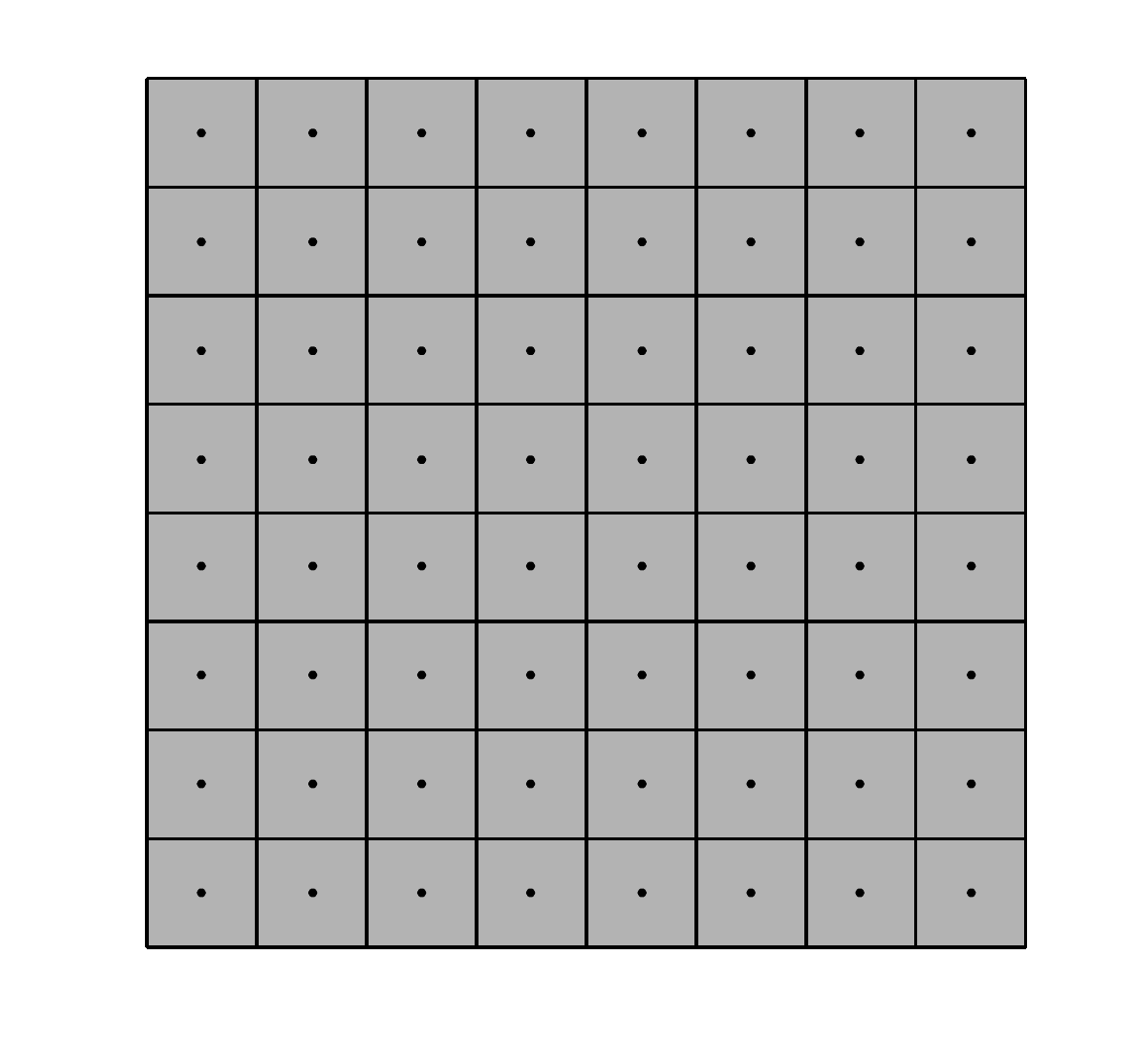}}
	\subfigure[The cube-split constellation on $G(\RR^3,1)$]{
		\includegraphics[width=.42\textwidth,trim=0cm .5cm .5cm 0cm,clip=true]{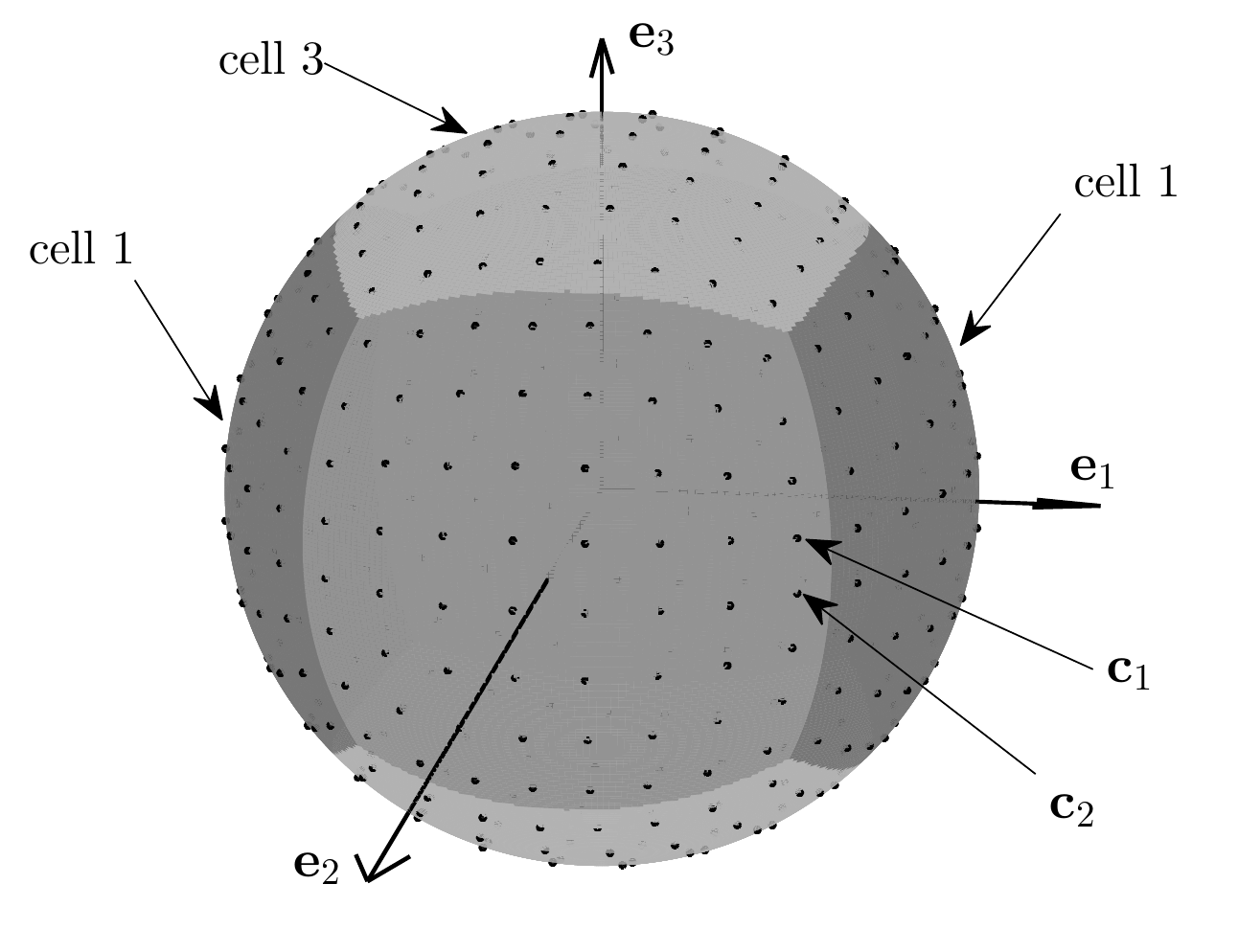}}
	\caption{Illustration of the cube-split constellation on $G(\RR^3,1)$ for $B_1 = B_2 = 3$ bits. {The symbols, represented by the dots, are mapped from a grid on $(0,1)^2$ (above) to one of the cells depicted by different gray levels (below).} Note that each symbol in $G(\RR^3,1)$ is depicted twice due to the sign indeterminacy associated to the Grassmannian. The constellation defines $T \times 2^{B_1+B_2} = 192$ lines in $\RR^3$. 
	Two symbols $\cv_1$ and $\cv_2$ with minimum distance are in the middle of an edge of a cell. 
	}
	\label{fig:cubesplit}
\end{figure}

\subsection{Minimum Distance}
The following lemma provides theoretical benchmarks for the minimum distance of an optimal constellation with given size.\footnote{In our setting with $B_j \ge 1$, $j\in [2(T-1)]$, the upper bound in \eqref{eq:minDistBounds} is tighter than the Rankin bounds  $\delta \le \sqrt{\big(1-\frac{1}{T}\big)\frac{|\Cc|}{|\Cc|-1}}$ if $|\Cc| \le \frac{T(T+1)}{2}$ and $\delta \le \sqrt{1-\frac{1}{T}}$ if $|\Cc|> \frac{T(T+1)}{2}$~\cite{Conway1996packing}.} 
\begin{lemma} \label{lemma:mindistBounds}
	The minimum distance $\delta$ of an optimal constellation $\Cc_{\rm opt}$ of cardinality $|\Cc|$ on the complex Grassmannian of lines $G(\CC^T,1)$ is bounded by
	\begin{align} \label{eq:minDistBounds}
	{\min\big\{1,2|\Cc|^{-\frac{1}{2(T-1)}}\big\} \ge \delta \ge |\Cc|^{-\frac{1}{2(T-1)}}.}
	\end{align}
\end{lemma}
\begin{proof}
	According to~\cite{Mukkavilli2003rateFeedback} (also stated in~\cite[Corollary~1]{Dai2008quantizationBounds}), the volume of a metric ball $\Bc(\delta)$ of radius $\delta$ on $G(\CC^T,1)$ with normalized invariant measure $\mu(.)$ is given by
$
	\mu(\Bc(\delta)) = \delta^{2(T-1)}$.
	Let $\Cc_{\rm opt}$ denotes the optimal constellation on $G(\CC^T,1)$ with minimum distance $\delta$, we have the Gilbert-Varshamov lower bound~\cite[Eq.(2)]{Dai2008quantizationBounds} and Hamming upper bound~\cite[Eq.(3)]{Dai2008quantizationBounds} on the size of the code as 
$
	\frac{1}{\mu(\Bc(\delta/2))} \ge |\Cc| \ge \frac{1}{\mu(\Bc(\delta))}$.
	Next, by substituting the volumes $\mu(\Bc(\delta/2))$ and $\mu(\Bc(\delta))$ into this, \eqref{eq:minDistBounds} follows readily.
\end{proof}
	
In Fig.~\ref{fig:mindist}{(a)}, we compare the minimum distance of the cube-split constellation for $T=2$ and $T=4$ with these fundamental limits. We also plot the minimum distance of the numerically optimized constellation generated by approximating the optimization~\eqref{eq:designCriteria} by 
	$
	\min_{\Cc} \, \log \sum_{1\le j<l \le |\Cc|} \exp\big(\frac{|\cv_j^\H\cv_l|}{\epsilon}\big)
	$
	with a small {``diffusion constant''} $\epsilon$ for smoothness, then solving it by conjugate gradient descent on the Grassmann manifold using the Manopt toolbox~\cite{manopt}. We also compare with other structured constellations: the Fourier constellation~\cite{Hochwald2000systematicDesignUSTM}, which coincides with the constellation in~\cite[Sec.~III-A]{TarokhIT2002existence} when $T =2$; the exp-map constellation~\cite{Kammoun2007noncoherentCodes}; 
	the coprime-PSK constellation~\cite{Zhang2011full_diversity_blind_ST_block_codes} where $\zv$ is taken from a numerically optimized constellation in $G(\CC^{T-2},1)$; 
	and the multi-layer constellation~\cite{AttiahISIT2016systematicDesign} in which we use the canonical basis as the initial constellation and adopt the parameters in \cite[Sec.IV]{AttiahISIT2016systematicDesign}: $L = 3$ layers, the moving distances $\sin(\phi_2) = 0.6$ and $\sin(\phi_3) = 0.35$.
\begin{figure*}[!h] 
	\centering
	\hspace{-.5cm}
	\subfigure[Minimum distance vs. constellation size $|\Cc|$ for $T \in \{2,4\}$]{
%
%
\definecolor{mycolor1}{rgb}{1.00000,0.00000,1.00000}%
\begin{tikzpicture}[scale=0.66,style={mark size=3pt,line width=3pt}]
\begin{axis}[%
width=5.5in,
height=3.6in,
at={(0in,-0.5in)},
scale only axis,
xmin=2,
xmax=20,
xlabel style={font=\color{white!15!black},at={(axis description cs:0.5,-0.05)}},
xlabel={$\log_2(|\Cc|)$},
ymode=log,
ymin=3.5e-05,
ymax=1,
yminorticks=true,
ylabel style={font=\color{white!15!black}},
ylabel={Minimum distance},
axis background/.style={fill=white},
xmajorgrids,
ymajorgrids,
yminorgrids,
label style={font=\large},
legend style={at={(0.01,0.006)}, anchor=south west, legend cell align=left, align=left, draw=white!15!black,nodes={scale=0.85}}
]
\addplot [color=blue, mark size=2pt, mark=*, mark options={solid, blue}]
table[row sep=crcr]{%
	3	0.54654578501864\\
	5	0.232630836788292\\
	7	0.0971628697802939\\
	9	0.0414025256220549\\
	11	0.0180976818375968\\
	13	0.00808885346018279\\
	15	0.00367940611257187\\
	17	0.0016963956932459\\
	19	0.000790308275004024\\
};
\addlegendentry{Cube-split constellation}

\addplot [color=black, dashdotted, mark=asterisk, mark options={solid, thick, black}]
table[row sep=crcr]{%
	3	0.607522787499177\\
	5	0.320324474603031\\
	7	0.157278215541468\\
	9	0.0746538110291644\\
	11	0.0357013698028152\\
	13	0.0177335087366133\\
};
\addlegendentry{Numerically optimized constellation}

\addplot [color=mycolor1, dotted, mark=o, line width=1.2pt, mark options={solid, thick, mycolor1}]
table[row sep=crcr]{%
	3	0.38268343236509\\
	5	0.0980171403295624\\
	7	0.0245412285229348\\
	9	0.00613588464923292\\
	11	0.00153398018667962\\
	13	0.000383495189158785\\
	15	9.5873803622534e-05\\
	17	2.39684656701872e-05\\
};
\addlegendentry{Fourier constellation~\cite{Hochwald2000systematicDesignUSTM}}

\addplot [color=mycolor1, dotted, mark=+, line width=1.2pt, mark options={solid, mycolor1,thick}]
table[row sep=crcr]{%
	3	0.216656009406344\\
	5	0.107652703711034\\
	7	0.0481417266935054\\
	9	0.0226953155532321\\
	11	0.0110111367326326\\
	13	0.0040575594205403\\
	15	0.00162747988655728\\
};
\addlegendentry{Exp-map constellation~\cite{Kammoun2007noncoherentCodes}}

\addplot [color=green, dashdotted, line width=1.2pt, mark=x, mark options={solid, green}]
table[row sep=crcr]{%
	2.58496250072116	0.5\\
	4.90689059560852	0.104528463267653\\
	7.04439411935845	0.0237976975461042\\
	8.98299357469431	0.0062086412452293\\
	11.0154150523867	0.00151767702771016\\
	12.9996477365284	0.00038358883656127\\
	15.0076405272819	9.53673908952537e-05\\
};
\addlegendentry{Coprime-PSK constellation~\cite{Zhang2011full_diversity_blind_ST_block_codes}}

\addplot [color=green, dashdotted, line width=1.2pt, mark=triangle, mark options={solid, green,rotate=-90}]
table[row sep=crcr]{%
	3	0.0715798318654663\\
	5	0.0715798318654632\\
	7	0.00379451041404461\\
	9	0.00156779726960514\\
	11	0.0035204095976859\\
	13	8.78718140079727e-05\\
};
\addlegendentry{Multi-layer constellation~\cite{AttiahISIT2016systematicDesign}}

\addplot [color=red, dashdotted, line width=1pt]
table[row sep=crcr]{%
	3	0.353553390593274\\
	5	0.176776695296637\\
	7	0.0883883476483184\\
	9	0.0441941738241592\\
	11	0.0220970869120796\\
	13	0.0110485434560398\\
	15	0.0055242717280199\\
	17	0.00276213586400995\\
	19	0.00138106793200498\\
	21	0.000690533966002488\\
};
\addlegendentry{Lower bound of optimal constellation~\eqref{eq:minDistBounds}}

\addplot [color=red, dashed, line width=1pt]
table[row sep=crcr]{%
	3	0.7937005259841\\
	5	0.39685026299205\\
	7	0.198425131496025\\
	9	0.0992125657480125\\
	11	0.0496062828740062\\
	13	0.0248031414370031\\
	15	0.0124015707185016\\
	17	0.00620078535925078\\
	19	0.00310039267962539\\
	21	0.00155019633981269\\
};
\addlegendentry{Upper bound of optimal constellation~\eqref{eq:minDistBounds}}

\addplot [color=blue, mark size=2pt, mark=*, mark options={solid, blue}, forget plot]
table[row sep=crcr]{%
	5	0.556357841428337\\
	6	0.510745001645169\\
	8	0.481507178770003\\
	10	0.22867743566646\\
	12	0.194257030948684\\
	14	0.183408120938834\\
	16	0.0891581317317616\\
	18	0.0781513772775711\\
	20	0.0761550325122619\\
};
\addplot [color=black, dashdotted, mark=asterisk, mark options={solid,thick, black}, forget plot]
table[row sep=crcr]{%
	5	0.816432439746329\\
	6	0.723386833170898\\
	8	0.587732860154555\\
	10	0.471700416303291\\
	12	0.363971564187616\\
};
\addplot [color=mycolor1, dotted, mark=o, line width=1.2pt, mark options={solid, thick,mycolor1}, forget plot]
table[row sep=crcr]{%
	5	0.707106781186548\\
	8	0.408793812018675\\
	10	0.267241904713808\\
	12	0.169470093133017\\
	14	0.105709669277041\\
	16	0.0664589887902453\\
	18	0.0411840308950683\\
	20	0.0259342925914511\\
	26	0.00417769112508464\\
};
\addplot [color=mycolor1, dotted, mark=+, line width=1.2pt, mark options={solid, mycolor1,thick}, forget plot]
table[row sep=crcr]{%
	6	0.612372435695794\\
	8	0.230942834495574\\
	9	0.147433999135508\\
	10	0.213189195430122\\
	11	0.253594828293568\\
	12	0.252532650614033\\
	13	0.182773497263885\\
	14	0.118570915845471\\
	15	0.0981875846606188\\
	16	0.10492791807493\\
	18	0.1046\\
	20	0.0763133976310256\\
};
\addplot [color=green, dashdotted, line width=1.2pt, mark=x, mark options={solid, green}, forget plot]
table[row sep=crcr]{%
	4.90689059560852	0.430837618314107\\
	5.90689059560852	0.399358120899028\\
	7.97727992349992	0.198623282219553\\
	9.99435343685886	0.221338979747704\\
	11.9943534368589	0.110132003128971\\
	13.9985904297453	0.112984536976245\\
};
\addplot [color=green, dashdotted, line width=1.2pt, mark=triangle, mark options={solid, green,rotate=-90}, forget plot]
table[row sep=crcr]{%
	4	0.0715798318654663\\
	6	0.0715798318654632\\
	8	0.00379451041404461\\
	10	0.00156779726960514\\
	12	0.0035204095976859\\
	14	8.78718140079727e-05\\
};
\addplot [color=red, dashdotted, line width=1pt, forget plot]
table[row sep=crcr]{%
	5	0.5\\
	8	0.353553390593274\\
	14	0.176776695296637\\
	20	0.0883883476483184\\
	26	0.0441941738241592\\
	32	0.0220970869120796\\
	38	0.0110485434560398\\
	44	0.0055242717280199\\
	50	0.00276213586400995\\
	56	0.00138106793200498\\
	62	0.000690533966002488\\
};
\addplot [color=red, dashed, line width=1pt, forget plot]
table[row sep=crcr]{%
	5	1.12246204830937\\
	8	0.7937005259841\\
	14	0.39685026299205\\
	20	0.198425131496025\\
	26	0.0992125657480125\\
	32	0.0496062828740062\\
	38	0.0248031414370031\\
	44	0.0124015707185016\\
	50	0.00620078535925078\\
	56	0.00310039267962539\\
	62	0.00155019633981269\\
};
\end{axis}

\begin{axis}[%
width=5.5in,
height=3.6in,
at={(0in,-0.5in)},
scale only axis,
xmin=0,
xmax=1,
ymin=0,
ymax=1,
axis line style={draw=none},
ticks=none,
axis x line*=bottom,
axis y line*=left
]
\draw [black] (axis cs:0.11,.88) ellipse [x radius=0.012, y radius=0.085];
\draw [black] (axis cs:0.555,0.86) ellipse [x radius=0.015, y radius=0.1];
\draw [-latex,thick] (.07,.64) node [below] {\large $T = 2$} -- (.105,.8);
\draw [-latex,thick] (.12,.61) -- (.27,.467);
\draw [-latex,thick] (.7,.65) node [right] {\large $T = 4$} -- (.566,.77);
\draw [-latex,thick] (.7,.65) -- (.595,.34);
\end{axis}
\end{tikzpicture}%
}
	\subfigure[Minimum distance vs. symbol length $T$ for $|\Cc| = T2^{2(T-1)}$ ($B_0 = 1$)]{
		\includegraphics[width=.43\textwidth]{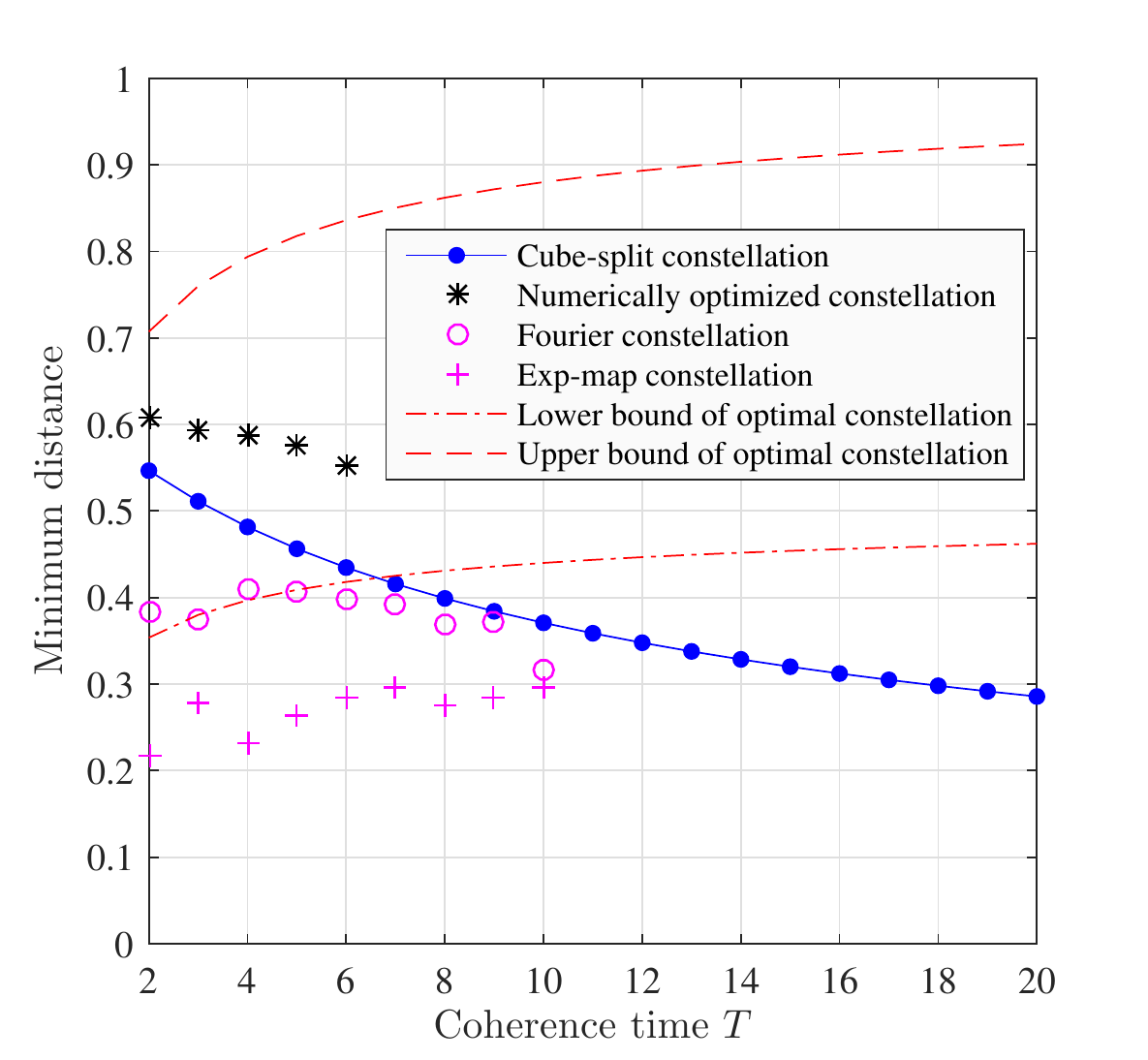}
	}
	\caption{The minimum distance of the cube-split constellation in comparison with other constellations and the fundamental limits of an optimal constellation given in~\eqref{eq:minDistBounds}.} 
	\label{fig:mindist}
\end{figure*}
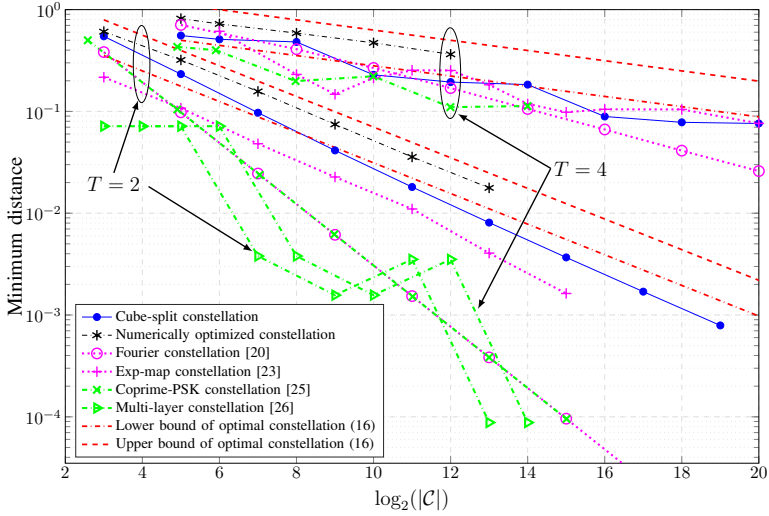
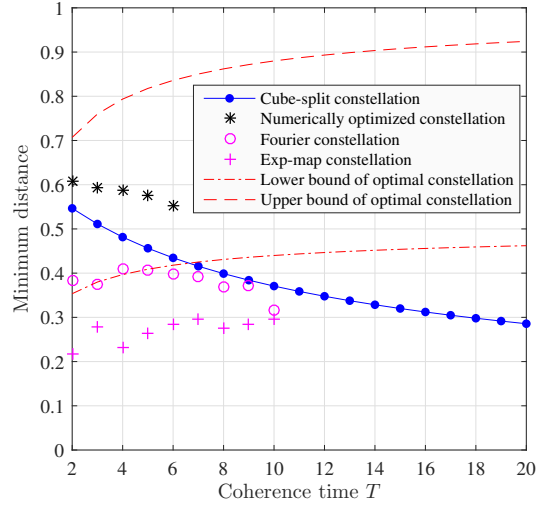 

We observe that the cube-split constellation has the largest advantage over other structured constellations when $B_0 \defeq B_1 = \dots = B_{2(T-1)}$. (For $T = 4$, this corresponds to $\log_2(|\Cc|) = 8, 14, 20$ bits/symbol.) In this case, all the real dimensions of a cell accommodate the same number of symbols, thus the symbols are more evenly spread. Hereafter, we denote this symmetric cube-split constellation by $CS(T,B_0)$. Let us consider a pair of symbols $\cv_1$ and $\cv_2$ in $CS(T,B_0)$ that are in the same cell of $G(\CC^T,1)$ with respective local coordinates $\av^{(1)}$ and $\av^{(2)}$ differing in only one component s.t.
\begin{equation}
a^{(1)}_{j_0} \ne a^{(2)}_{j_0}, \left\{{a}^{(1)}_{j_0}, {a}^{(2)}_{j_0}\right\} = \Big\{\frac{1}{2}-\frac{1}{2^{B_{0}+1}}, \frac{1}{2}+\frac{1}{2^{B_{0}+1}} \Big\}, \label{eq:mindist_coor_1}
\end{equation}
\begin{equation}
{a}^{(1)}_j = {a}^{(2)}_j \in \Big\{\frac{1}{2^{B_0+1}}, 1-\frac{1}{2^{B_0+1}} \Big\}, \forall j\ne j_0, \label{eq:mindist_coor_2}
\end{equation} 
for some $j_0 \in [2(T-1)]$. One such pair of symbols is illustrated in Fig.~\ref{fig:cubesplit}. The two symbols are in the middle of an edge of a cell. We conjecture, and have verified with all the cases depicted in Fig.~\ref{fig:mindist}, that this pair of symbols achieves the minimum distance of $CS(T,B_0)$.
While a proof of this conjecture for general $B_0$ remains elusive, {it can be proved for the particular case $B_0 = 1$.

\begin{lemma} \label{lem:mindist_CS(T,1)}
For the $CS(T,1)$ constellation, two symbols with minimum distance are in the same cell and have the respective local coordinates given by \eqref{eq:mindist_coor_1} and \eqref{eq:mindist_coor_2} with $B_0 = 1$. The minimum distance is given by $d_{\min}(T,1) = \sqrt{1-\big|1-\frac{1+\im}{c^{-1} + T-1}\big|^2}$, where $c \defeq \frac{1-e^{-m^2}}{1+e^{-m^2}}$ with $m \defeq \Nc^{-1}(3/4)$.
\end{lemma}
\begin{proof}
	The proof is given in Appendix~\ref{proof:mindist_CS(T,1)}.
\end{proof}

We plot the minimum distance of the $CS(T,1)$ constellation as a function of $T$ in Fig.~\ref{fig:mindist}{(b)}. For the cases where we can compute the minimum distance of the Fourier constellation and the exp-map constellation of the same size, we notice that these constellations have smaller minimum distance than that of our constellation, especially for small $T$.}

In general, the conjectured minimum distance over the constellation, i.e., the distance between $\cv_1$ and $\cv_2$ defined by \eqref{eq:mindist_coor_1} and \eqref{eq:mindist_coor_2}, is given by
%
	\begin{align} \label{eq:mindistance}
	{\tilde{d}(T,B_0)} \defeq \sqrt{1-\Big|1+\frac{\alpha}{\alpha+\beta}(e^{\im 2 \varphi}-1)\Big|^2},
	\end{align}
	where $\alpha \defeq \frac{1-\exp\big(-\frac{m_0^2+m_1^2}{2}\big)}{1+\exp\big(-\frac{m_0^2+m_1^2}{2}\big)}$, $\beta \defeq 1+(T - 2)\frac{1-e^{-m_0^2}}{1+e^{-m_0^2}}$, and $\varphi \defeq \arctan\left(\frac{m_1}{m_0}\right)$,
	with $m_0 \defeq \Nc^{-1}\left(2^{-B_0-1}\right)$ and $m_1 \defeq \Nc^{-1}\left(\frac{1}{2}+2^{-B_0-1}\right)$.
\begin{lemma}
When $|\Cc|$ is large, it holds that
\begin{multline}
\log_2\big({\tilde{d}(T,B_0)}\big) \\= -\frac{1}{2(T-1)}\log_2(|\Cc|)-\frac12\log_2\log_2(|\Cc|) + O(1).
\end{multline}
\end{lemma}
\begin{proof}
First, it is straightforward to see that, when $|\Cc|$ is large, i.e. $B_0$ is large, $m_0$ goes to $-\infty$ and $m_1$ goes to $0$ and it follows that $\alpha$ goes to $1$, $\beta$ goes to $T-1$, and $\varphi = \frac{m_1}{m_0} + o\left(\frac{m_1}{m_0}\right)$.
Then
\begin{align}
\lefteqn{\tilde{d}(T,B_0)
}\notag \\ 
 &= \sqrt{\frac{2\alpha}{\alpha+\beta}(1-\cos(2\varphi)-\sin(2\varphi)) - \Big(\frac{\alpha}{\alpha+\beta}\Big)^2{|e^{\im 2\varphi}-1|}^2}\\
&= \sqrt{\frac{4}{T}\varphi^2 - \frac{4}{T^2}\varphi^2+o(\varphi^2)}
= 2\frac{\sqrt{T-1}}{T}|\varphi|+o(|\varphi|)\label{eq:2}.
\end{align}
On the other hand, using \cite[Sec.26.2]{abramovitz}, it follows that $m_1 = \sqrt{2\pi}2^{-B_0-1}+o( 2^{-B_0})$ and $|m_0| = \sqrt{2\log(2^{B_0+1})}+o\big( \sqrt{2\log(2^{B_0+1})}\big)$.
Inserting this and $\varphi = \frac{m_1}{m_0} + o\big(\frac{m_1}{m_0}\big)$ into \eqref{eq:2} gives
\begin{align}
\log_2(\tilde{d}(T,B_0)) &= \log_2(m_1) - \log_2(|m_0|) + O(1) \\
&= -B_0 - \frac12 \log_2(B_0) + O(1),
\end{align}
which yields the result.
\end{proof}
Therefore, under the conjecture that $\tilde{d}(T,B_0)$ is the minimum distance of the $CS(T,B_0)$ constellation, the constellation is asymptotically optimal w.r.t. the bounds in Lemma~\ref{lemma:mindistBounds} up to a log factor and a constant. In Fig.~\ref{fig:mindist_stem}, we plot the spectrum of the symbol-wise minimum distance, i.e., the distance from each symbol to its nearest neighbor, of $CS(2,4)$ and $CS(4,1)$. The symbol-wise minimum distances are concentrated and compare well to the bounds. 
Every symbol in $CS(4,1)$ has the same distance to its nearest neighbor, which is ${d}_{\min}(4,1)$. {This property holds for any $CS(T,1)$ constellation because for any $\xv_1 \in CS(T,1)$, there exists other symbol $\xv_2$ in the same cell with coordinates satisfying \eqref{eq:mindist_coor_1} and \eqref{eq:mindist_coor_2} for $B_0 = 1$, thus $d(\xv_1,\xv_2) = d_{\min}(T,1)$.}
\begin{figure}[!h] 
	\centering
	\subfigure[the $CS(2,4)$ constellation ($512$ symbols)]{
		\includegraphics[width=.23\textwidth]{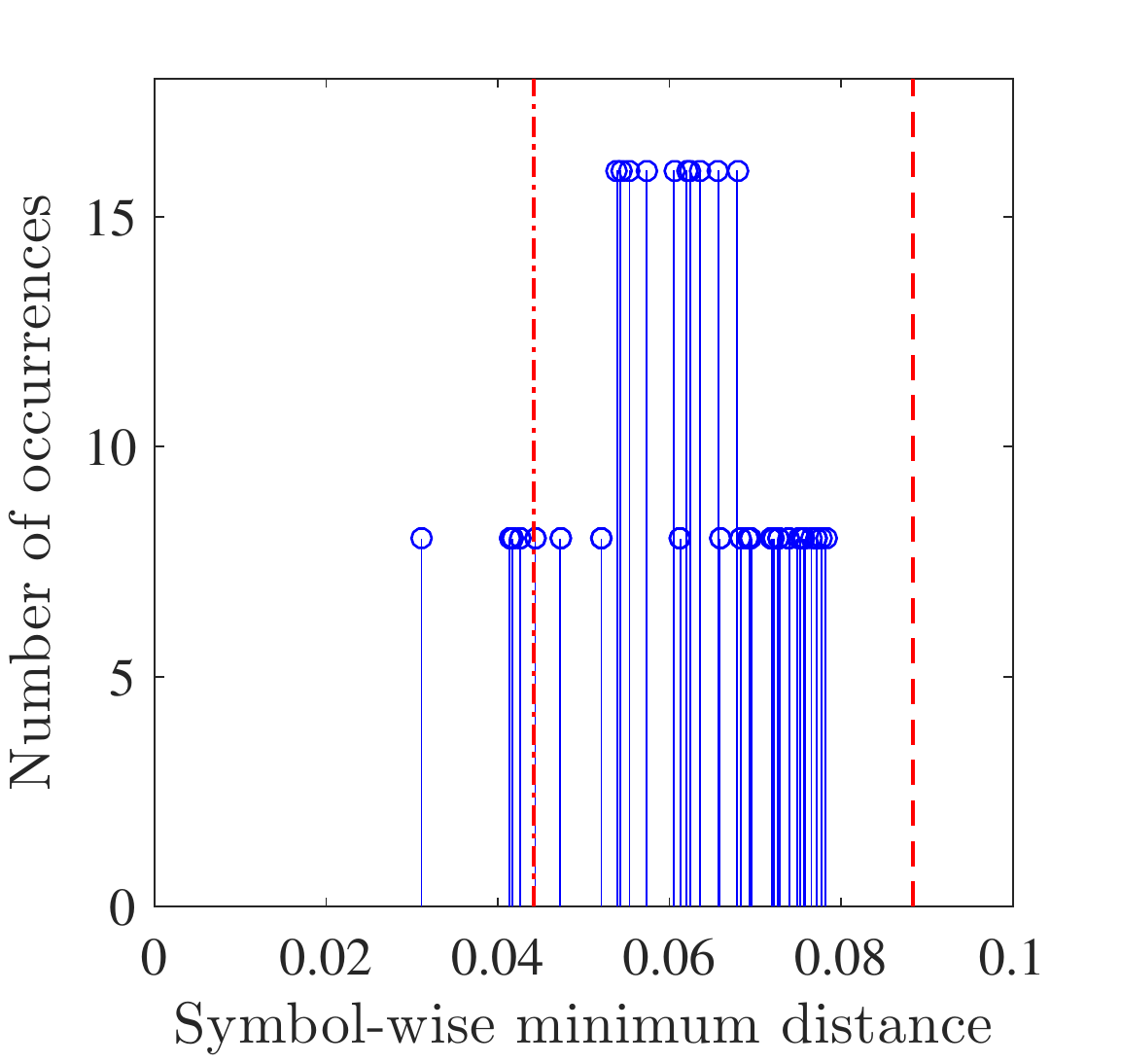}}
	\subfigure[the $CS(4,1)$ constellation ($256$ symbols)]{
		\includegraphics[width=.23\textwidth]{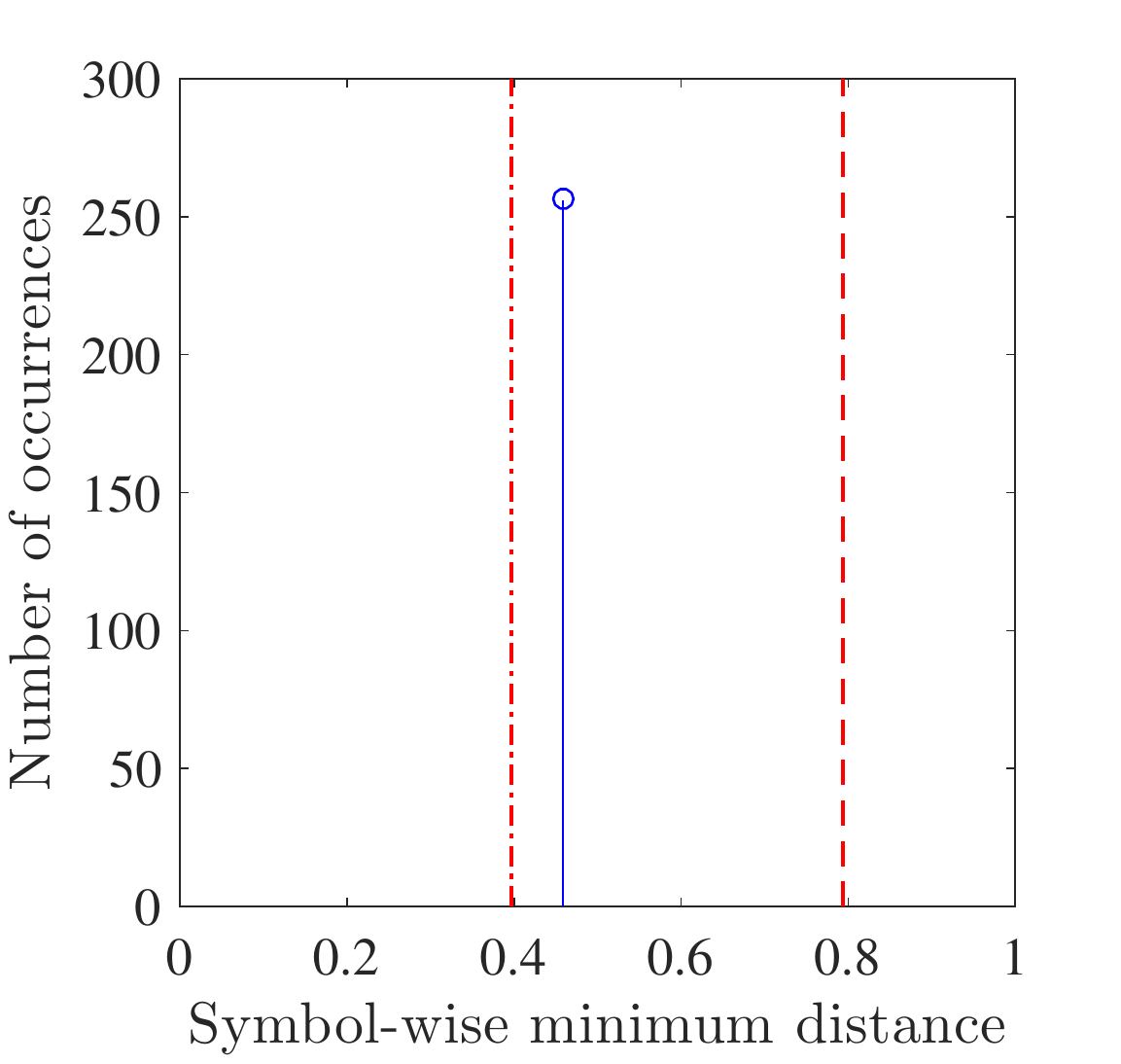}}
	\caption{The symbol-wise minimum distance spectrum of the $CS(2,4)$ and $CS(4,1)$ constellations. The dashed and dash-dotted lines are respectively the upper and lower bounds \eqref{eq:minDistBounds} of the minimum distance of an optimal constellation of the same size.}
	\label{fig:mindist_stem}
\end{figure} 

\subsection{Binary Labeling} \label{sec:labeling}
Another important aspect of designing a constellation is to label each symbol with a binary vectors. {If $T$ is a power of $2$, the number of bits $B$ required to represent a symbol is
\begin{equation}
B = \log_2(T)+ \sum_{j=1}^{2(T-1)}B_j.
\end{equation} 
This can be understood as a hierarchical labeling, where $T$ bits indicate the cell index, and the remaining bits indicate the local coordinates of a symbol.
If $B - \log_2 T$ is not divisible by $2(T-1)$, the sets $A_j, j \in [2(T-1)]$ cannot be chosen to have equal size. In that case, 
we can initially let $B_1 = \dots = B_{2(T-1)} = \big\lfloor\frac{B- \log_2(T) }{2(T-1)}\big\rfloor$, then allocate one more bit to each of $B - \log_2(T) - 2(T-1)\big\lfloor\frac{B - \log_2(T) }{2(T-1)}\big\rfloor$ randomly selected dimensions.} 
If $T$ is not a power of $2$, the constellation size is also not a power of $2$, which does not support a convenient bit mapping. Although we can manipulate (e.g., augment or truncate) the constellation s.t. the size becomes a power of $2$, this alters the constellation structure. In the remainder of the paper, we focus on the case $T$ being a power of $2$ whenever the bit mapping is concerned.}

The binary labels should be assigned s.t. a symbol error does not cause many bit errors. This requires that symbols which are likely to be mistaken for each other should differ by a minimal number of bits in their labels. In other words, symbols with small {(chordal or Euclidean)} distance are given labels with small Hamming distance. This is the principle of Gray labeling which was shown to be optimal in terms of average bit error probability for structured scalar constellations such as PSK, pulse amplitude modulation~(PAM), and QAM~\cite{Agrell2004ITGraycode}. 
Ideally, a Gray labeling scheme gives the neighboring symbols labels that differ by exactly {\em one} bit. As shown in \cite[Thm.~1]{Nghi2007USTM_iterativeDecoding}, this is possible for a special case of the Grassmannian constellation in~\cite{ZhaoTIT2004orthogonalDesign}. Nevertheless, this is rarely the case in general due to the irregular neighboring properties. When a true Gray labeling is not possible, finding a quasi-Gray one requires an exhaustive search over $|\Cc|!$ candidate labelings. Therefore, one often resorts to sub-optimal labeling schemes.

An iterative labeling scheme consisting in propagating the labels along the edges of the neighboring graph was proposed in~\cite{Baluja2013neighborhood}. Unfortunately, building and storing such a graph is possible only for constellations of low dimension and small size. 
In~\cite{Colman2011quasi-gray}, two matching methods to label a Grassmannian constellation, say $\Cc$, are proposed. The first, so-called match-and-label algorithm, matches $\Cc$ to an auxiliary constellation which can be Gray labeled. The second, so-called successive matching algorithm, matches the chordal distance spectrum of $\Cc$ with the Hamming distance spectrum of an auxiliary Gray label. 
However, these three schemes still have complexity at least cubic in the constellation size and
do not offer any optimality guarantee. 

For our cube-split constellation, we introduce a simple yet effective and efficient Gray-like labeling scheme by exploiting the constellation structure. Recall that the number of bits per symbol is $B = \log_2(T)+ \sum_{j=1}^{2(T-1)}B_j$, and a symbol is entirely determined by the cell index $i$ and the set of local coordinates $\{a_1,\dots,a_{2(T-1)}\}$. Our labeling scheme works as follows.
\begin{itemize} [leftmargin=*]
	\item We let the first $\log_2(T)$ bits represent the cell index $i$ and denote them by {\em cell bits}. These bits are defined simply as the binary representation of $i-1$. Note that no optimization of the labels of the cell index is possible since each cell have common boundaries with all other cells, as can be seen in Fig.~\ref{fig:cubesplit}. 
	\item We let each of the next groups of $B_j$ bits represent the local coordinate $a_j \!\in\! A_j$ and denote them by {\em coordinate bits}. {Since $A_j$ is a set of points in one dimension, there always exist Gray labels associated to its elements, in the same manner as a Gray labeling for a PAM constellation. 
	For example, if $B_j = 3$ (as in Fig.~\ref{fig:cubesplit}), $A_j = \big\{\frac{1}{16},\frac{3}{16},\frac{5}{16},\dots,\frac{15}{16}\big\}$ and one option to Gray label it is $\{000,001,011,010,110,111,101,100\}$, as illustrated in Fig.~\ref{fig:Ai_visualization}.}
	\begin{figure} [h]
		\centering
		\includegraphics[width=.42\textwidth]{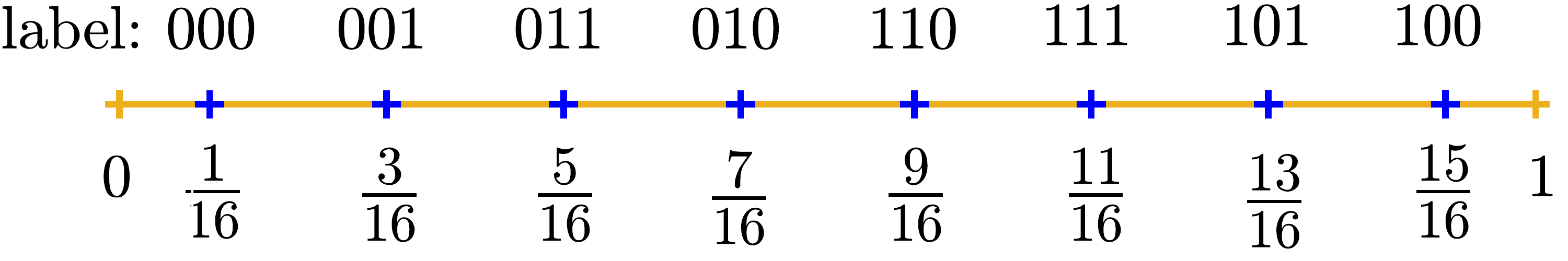}
		\caption{A visualization of a Gray labeling for the coordinate set $A_j$ with $B_j = 3$.}
		\label{fig:Ai_visualization}
	\end{figure}
\end{itemize}
Thanks to the structured constellation design and labeling scheme outlined above, the complexity of generating a constellation symbol $\xv$ from its binary representation is essentially linear in $T$.
It can be done on-the-fly and requires no storage.

In Fig.~\ref{fig:BER_labels}, we compare the performance in terms of bit error rate of this Gray-like labeling scheme with random labeling, graph propagation labeling~\cite{Baluja2013neighborhood}, match-and-label labeling and successive matching labeling~\cite{Colman2011quasi-gray}. For the match-and-label scheme, we use the exp-map constellation~\cite{Kammoun2007noncoherentCodes} as the auxiliary constellation. This constellation is mapped from coherent symbols $\qv \in \CC^{T-1}$ containing $T-1$ QAM symbols, and thus can be quasi-Gray labeled by taking the Gray label of $\qv$. 
It can be seen that for the considered $CS(2,4)$ and $CS(4,1)$ constellations, our Gray-like labeling scheme, albeit being simpler, outperforms the other considered schemes. 
\begin{figure}[!h] 
	\centering
	\subfigure[$T = 2, B_0 = 4$ ($512$ symbols)]{
		\includegraphics[width=.48\textwidth]{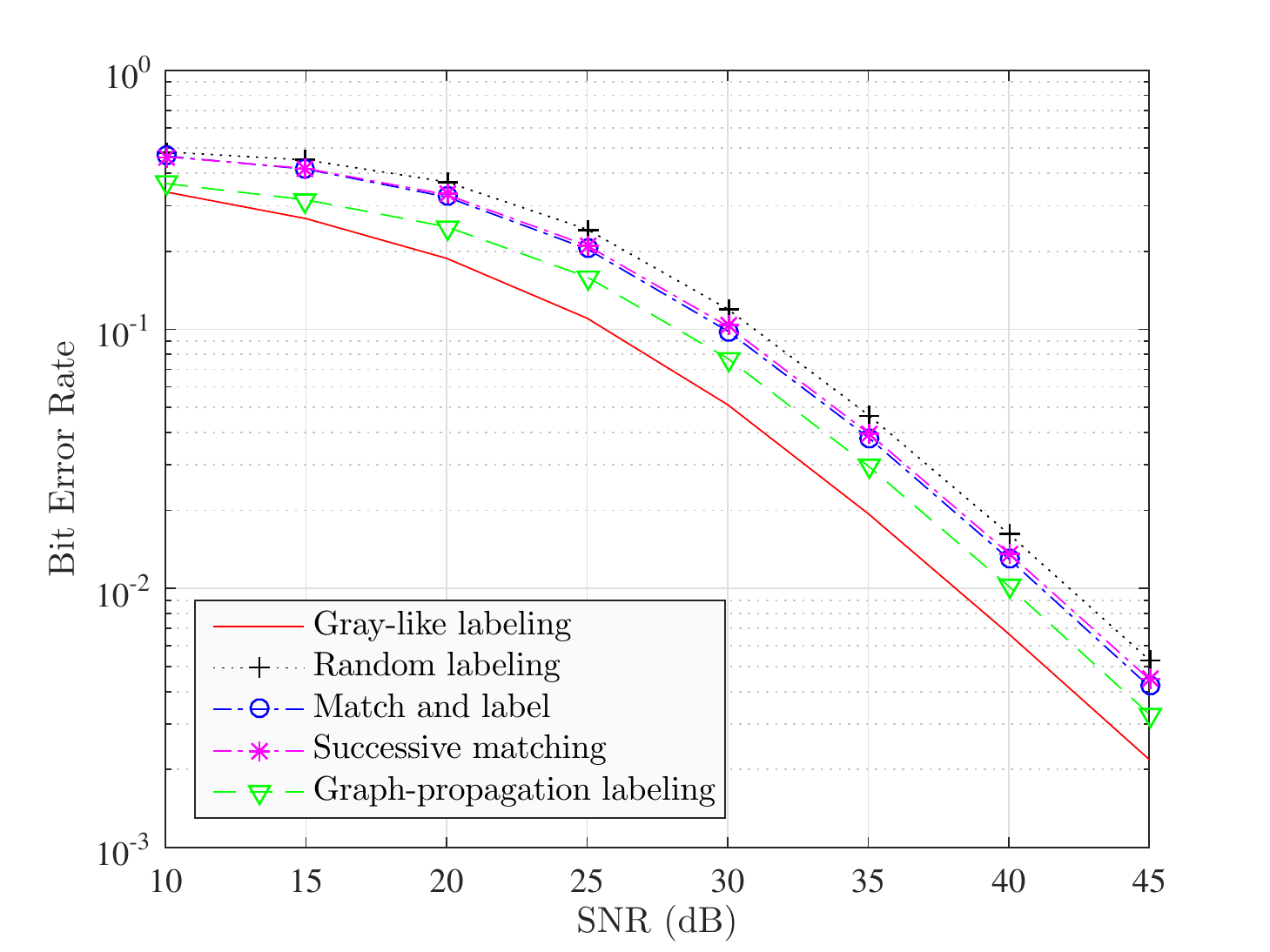}}
	\subfigure[$T = 4, B_0 = 1$ ($256$ symbols)]{
		\includegraphics[width=.48\textwidth]{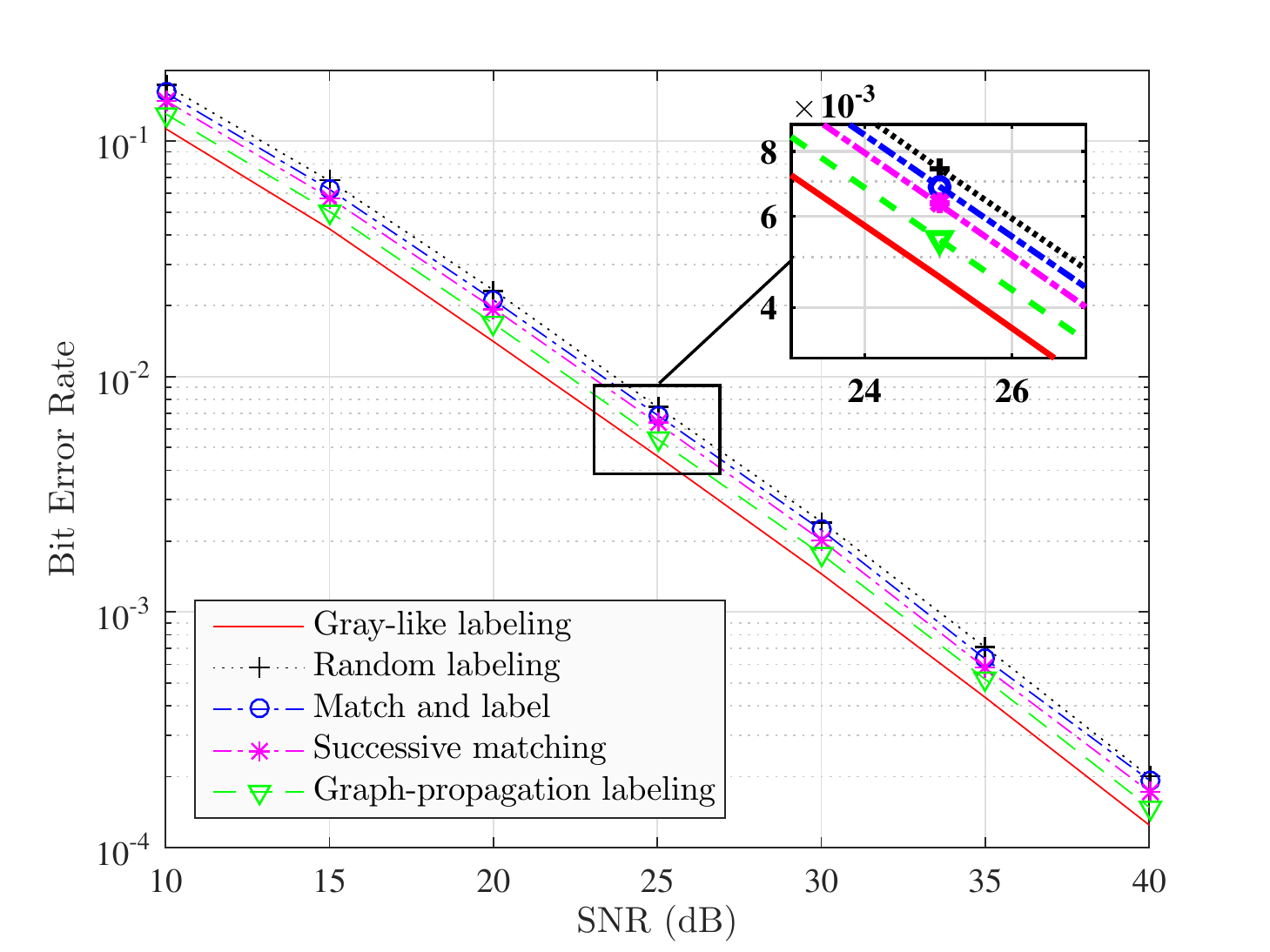}}
	\caption{The bit error rate of the cube-split constellation with ML decoder and different labeling schemes in a single-receive-antenna system. The proposed Gray-like labeling outperforms the other schemes.}
	\label{fig:BER_labels}
\end{figure}

\section{Low-Complexity Receiver Design} \label{sec:decoder}
In this section, leveraging the constellation structure, we design efficient symbol decoder and LLR computation from the observation $\rvMat{Y}$.
\subsection{Low-Complexity Greedy Decoder}
In order to avoid the high-complexity ML decoder, we propose to decode the symbol in a greedy manner by estimating sequentially the cell index $i$ and the local coordinates $\rvVec{a}$. 

\subsubsection{Step 1 - Denoising} We first use the fact that, by construction, the signal of interest is supported by a rank-$1$ component of $\rvMat{Y}$ (see \eqref{eq:channelModel}). We compute the left singular vector $\rvVec{u} = [\rv{u}_1 \ \rv{u}_2 \dots \rv{u}_T]^\T$ corresponding to the largest singular value of $\rvMat{Y}$, which is also the solution of 
\begin{equation} \label{eq:relaxed_ML}
\arg\max_{\uv \in \CC^T: \ \|\uv\|^2 = 1} \|\rvMat{Y}^\H \uv\|^2.
\end{equation}
Observe that this is a relaxed version of the ML decoder~\eqref{eq:MLoneUser} if we disregard the discrete nature of the constellation. Thus, $\rvVec{u}$ is a rough estimate of the transmitted symbol $\rvVec{x}$ on the unit sphere. 

\subsubsection{Step 2 - Estimating the Cell Index and the Local Coordinates} We then find the closest symbol to $\rvVec{u}$ by localizing $\rvVec{u}$ on the system of bent grids defined for the constellation. 
To do so, we estimate the cell index and the local coordinates. The cell index estimate is obtained as
\begin{equation}
\hat{i} = \arg\min_{j\in[T]} d(\rvVec{u},\ev_j) = \arg\max_{j\in[T]}|\rv{u}_j|.
\end{equation}
In the noise-free case, $\rvMat{Y} = \sqrt{\rho T} \rvVec{x} \rvVec{h}^\T$ has rank one and $\rvVec{u} = e^{\im \theta}\rvVec{x}$ for some $\theta \in [0,2\pi]$, thus $\hat{i} = i$ since $\rv{x}_{i}$ is the strongest component in $\rvVec{x}$ by construction (see~\eqref{eq:encodeMapping}). {The local coordinates are estimated by first applying the inverse mapping  $\gv_{\hat{i}}^{-1}$ (see~\eqref{eq:w} and \eqref{eq:t}) to obtain $\tilde{\rvVec{a}} = [\tilde{\rv{a}}_1 \ \dots \ \tilde{\rv{a}}_{2T-2}]^\T = \gv_{\hat{i}}^{-1}(\rvVec{u})$ and then find the closest point to $\tilde{\rvVec{a}}$ in $\bigotimes_{j=1}^{2(T-1)} A_j$. This can be done component-wise as}
\begin{equation}
 \hat{\rv{a}}_{j} = \arg\min_{{a} \in A_{j}} |\tilde{\rv{a}}_j-{a}|, ~~  j \in [2(T-1)]. 
\end{equation}
Again, in the absence of noise, $\rvVec{u} = e^{\im \theta} \rvVec{x}$, $\hat{i} = i$, and thus $\hat{\rvVec{a}} = \rvVec{a} = \gv_{i}^{-1}(\rvVec{x})$. 

The decoded symbol $\hat{\rvVec{x}}$ is then identified from the estimated parameters $\{\hat{i},\hat{\rvVec{a}}\}$ as $\hat{\rvVec{x}} = \gv_{\hat{i}}(\hat{\rvVec{a}})$. The cell bits are decoded by taking the binary representation of $\hat{i}-1$. The coordinate bits are demapped from $\hat{\rv{a}}_i$ using the Gray code defined for $A_j$, independently for each real component $j \in [2(T-1)]$.  
We observe that the decision regions of the proposed greedy decoder are close to the Voronoi regions~\eqref{eq:VoronoiRegions}, which are the optimal decision regions, as depicted in Fig.~\ref{fig:decisionRegions} for the constellation shown in Fig.~\ref{fig:cubesplit}{(b)} and $N=1$.
\begin{figure}[!h] 
	\centering
	\includegraphics[width=.48\textwidth,trim=1cm 3.7cm 1cm 1cm,clip=true]{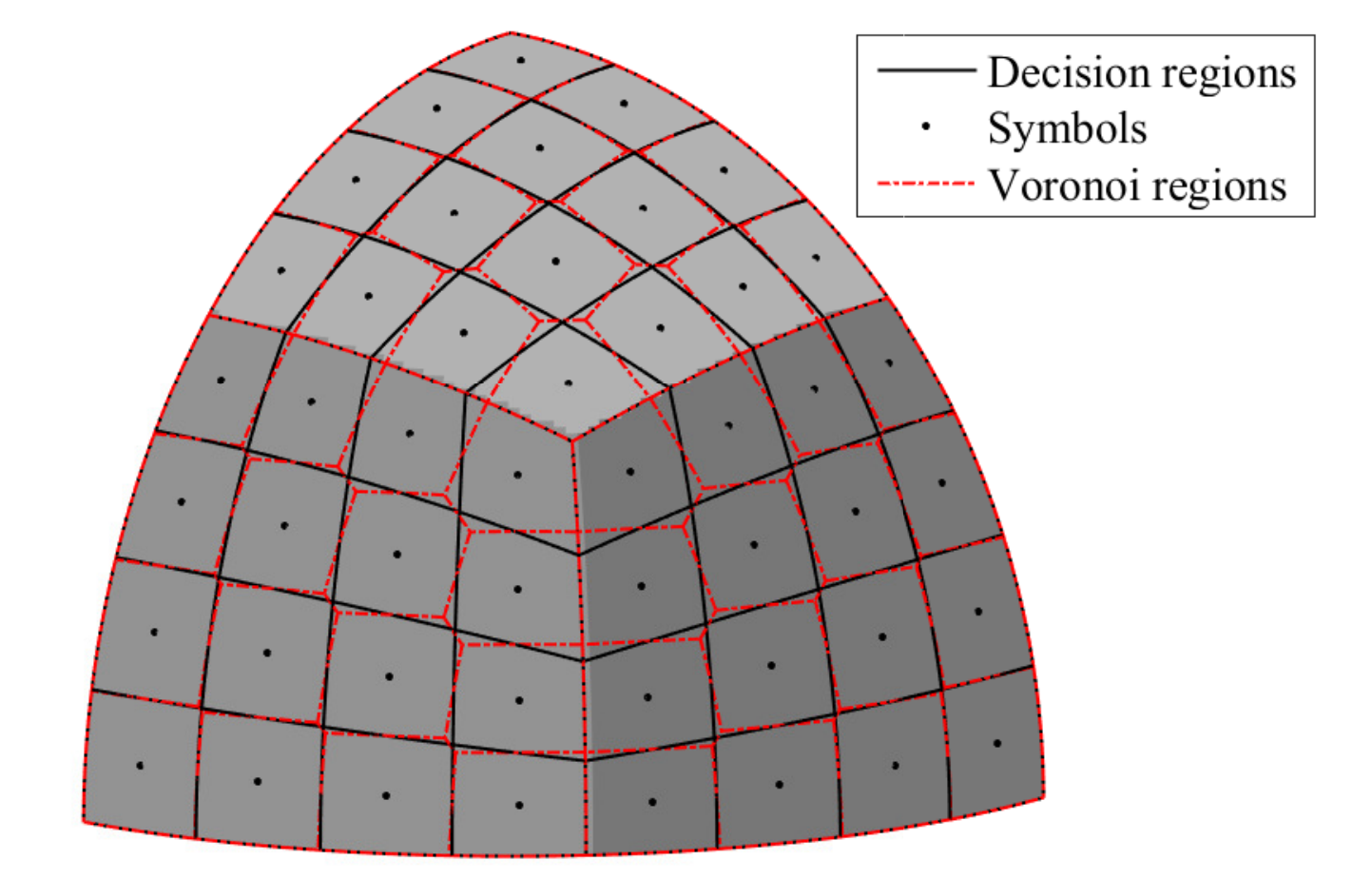}
	\caption{Illustration of the decision regions of the greedy decoder for a section (around the cell boundaries) of the cube-split constellation on $G(\RR^3,1)$ with $B_1 = B_2 = 3$ bits. These decision regions well match the Voronoi regions~\eqref{eq:VoronoiRegions}, which are the optimal decision regions.
	}
	\label{fig:decisionRegions}
\end{figure}

{The complexity of this greedy decoder is $O(NT\min\{N,T\})$ (dominated by the singular value decomposition of $\rvMat{Y}$ to find $\rvVec{u}$), independent of the constellation size. 
This is significantly reduced w.r.t the ML decoder whose complexity order is $O(NT|\Cc|)$. Among the other aforementioned structured constellations, only the exp-map constellation admits a simple sub-optimal decoder~\cite[Sec.~IV-B]{Kammoun2007noncoherentCodes}  
with equivalent complexity order $O(NT\min\{N,T\})$.}

\subsection{Demapping Error Analysis}
With the above greedy decoder, two types of error can occur. First, a {\em cell error} can occur due to false detection of the cell index $i$. The probability of cell error is
$
	\Pr\{\hat{i} \ne i\} = \Pr\big\{\arg\max_j |\rv{u}_j| \ne \arg\max_j |\rv{x}_j|\big\}.
$
Second, a {\em coordinate error} occurs when {at least} one of the local coordinates in $\rvVec{a}$ is wrongly detected. 
The probability of a coordinate error given correct cell detection is $\Pr\{\hat{\rvVec{a}} \ne \rvVec{a} \big| \hat{i} = i\}$.  
Then, the symbol error probability of the greedy decoder is
\begin{align}
P_e &= \Pr\{\hat{i} \ne i\} + (1-\Pr\{\hat{i} \ne i\})  \Pr\{\hat{\rvVec{a}} \ne \rvVec{a} \big| \hat{i} = i\}. 
\end{align} 
The error rate can be computed analytically for the $CS(T,1)$ constellation and $N=1$ as follows. 
\begin{proposition} \label{prop:Pe,b0=1}
	When $N=1$ and $B_1 = \dots = B_{2(T-1)} = 1$, the cell error probability is
	\begin{align} 
	\Pr\{\hat{i} \ne i\} = 1-&\int_{0}^{\infty} \int_{0}^{\infty} \left(1-Q_1(\sqrt{2c\rho_0 x},\sqrt{2y})\right)^{T-1} \\
	&\quad\cdot I_0(2\sqrt{\rho_0 x y}) e^{-y-(\rho_0+1)x} \dif y \dif x, \label{eq:SER_cell}
	\end{align}
	where $m \defeq \Nc^{-1}(3/4)$, $c \defeq \frac{1-e^{-m^2}}{1+e^{-m^2}}$, $\rho_0 \defeq \frac{\rho T}{1+(T-1)c}$, $I_0(x) \defeq \frac{1}{\pi} \int_{0}^{\pi} \exp\left(x \cos(\theta)\right)\dif \theta$ is the modified Bessel function of the first kind at order $0$~{\cite[Eq.(9.6.16)]{abramovitz}}, 
	and $Q_1(a,b) \defeq \int_{b}^{\infty} x \exp\left(-\frac{x^2+a^2}{2}\right) I_0(ax) \dif x$ is the Marcum Q-function~{\cite[Eq.(16)]{Marcum1960statistical}} with parameter $1$. 
	Given correct cell detection, the error probability of one pair of local coordinates is given in \eqref{eq:SER_coor_cond} for all $j \in [T-1]$.
	\begin{figure*}[!t]
		\begin{align} 
		&\Pr\big\{\{\hat{\rv{a}}_{2j-1},\hat{\rv{a}}_{2j}\} \ne \{\rv{a}_{2j-1},\rv{a}_{2j}\} \big|\ \hat{i} = i \big\}
		= \Pr\big\{\hat{\rv{q}}_j \ne \rv{q}_j \big|\hat{i} = i \big\}
	    \notag \\ 
		&= 1-\bigg(1+\frac{(1-c)\rho_0}{\sqrt{(2+(1+c)\rho_0)^2-4c\rho_0^2}}\bigg)^{-1} \Bigg(\frac{1}{4}+ \frac{\sqrt{2c}\rho_0\arccot\frac{1+(c-\sqrt{\frac{c}{2}})\rho_0}{\sqrt{1+(c+1)\rho_0+\frac{c}{2}\rho_0^2}}}{\pi\sqrt{1+(1+c)\rho_0+\frac{c}{2}\rho_0^2}} + \frac{(1-c)\rho_0\arccot\frac{2+(1-2\sqrt{2c}+c)\rho_0}{\sqrt{(2+(1+c)\rho_0)^2-4c\rho_0^2}}}{\pi\sqrt{(2+(1+c)\rho_0)^2-4c\rho_0^2}}\Bigg),\label{eq:SER_coor_cond}
		\end{align}
		\setlength{\arraycolsep}{1pt}
		\hrulefill \setlength{\arraycolsep}{0.0em}
		\vspace*{1pt}
	\end{figure*}
	The symbol error probability is then bounded by the union bound as
	\begin{equation} \label{eq:SER_coor}
	P_e \le \Pr\{\hat{i} \ne i\} + (T-1)\big(1-\Pr\{\hat{i} \ne i\}\big)\Pr\{\hat{\rv{q}}_j \ne \rv{q}_j \big|\hat{i} = i \}.
	\end{equation}
\end{proposition}
\begin{proof}
	The proof is given in~Appendix~\ref{app:PeBo1}.
\end{proof}
In particular, if $T=2$, the symbol error probability can be computed in closed form as follows.

\begin{corollary} \label{coro:Pe}
	When $N\!=\!1$, $T\!=\!2$ and $B_1 = \dots = B_{2(T-1)} = 1$, the cell error probability can be written explicitly as
	\begin{equation}
	\Pr\{\hat{i} \ne i\} = \frac{1}{2}\bigg(1-\frac{(1-c)\rho_0}{\sqrt{(2+(1+c)\rho_0)^2-4c\rho_0^2}}\bigg), \label{eq:PcellB01T2}
	\end{equation}
	Whereas the conditional coordinate error probability $\Pr\{\hat{\rvVec{a}} \ne \rvVec{a} | \hat{i} = i\}$ is exactly the right-hand side~(RHS) of~\eqref{eq:SER_coor_cond}. Accordingly, the symbol error probability $P_e$ is 
	\begin{multline}
	P_e = \frac{7}{8} - \frac{\sqrt{c}\rho_0\arccot \frac{1+\big(c-\sqrt{\frac{c}{2}}\big)\rho_0}{\sqrt{1+(c+1)\rho_0+\frac{c}{2}\rho_0^2}}}{\pi\sqrt{2+2(1+c)\rho_0+c\rho_0^2}} \\- \frac{(1-c)\rho_0\arccot\frac{2+\big(1-2\sqrt{2c}+c\big)\rho_0}{\sqrt{(2+(1+c)\rho_0)^2-4c\rho_0^2}}}{2\pi\sqrt{(2+(1+c)\rho_0)^2-4c\rho_0^2}}.
	\end{multline} 
\end{corollary}
\begin{proof}
	The proof is given in Appendix~\ref{app:PeT2Bo1}.
\end{proof}
\subsection{Log-Likelihood Ratio Computation {and Code Design}} \label{sec:LLR} 
When a channel code is employed, most channel decoders require the LLR of the coded bits as an input. LLR computation is performed independently from the code structure, assuming uniform input probabilities, i.e., all the symbols are equally likely to be transmitted, and so are the bits. Denote the binary label of symbol $\rvVec{x}$ as $\rvVec{b}(\rvVec{x}) = [\rv{b}_1(\rvVec{x}) \ \rv{b}_2(\rvVec{x}) \dots \rv{b}_B(\rvVec{x})]$. The log-likelihood ratio (LLR) of bit $\rv{b}_i$ given the observation $\rvMat{Y}=\Ym$ can be computed using~\eqref{eq:pdf} as
\begin{multline}
\LLR_j(\Ym) = \log \frac{p_{\rvMat{Y}|\rv{b}_j}(\Ym|1)}{p_{\rvMat{Y}|\rv{b}_j}(\Ym|0)} 
= \log \frac{\sum_{\alphav \in \Cc_ji^{(1)}}p_{\rvMat{Y}|\rvVec{x}}(\Ym|\alphav)}{\sum_{\betav \in \Cc_j^{(0)}}p_{\rvMat{Y}|\rvVec{x}}(\Ym|\betav)} \\
= \log \frac{\sum_{\alphav \in \Cc_j^{(1)}} \exp\big(\frac{\rho T}{1+\rho T}\|\Ym^\H \alphav\|^2\big)}{\sum_{\betav \in \Cc_j^{(0)}} \exp\big(\frac{\rho T}{1+\rho T}\|\Ym^\H \betav\|^2\big)}, \label{eq:LLRcubesplit}
\end{multline}
where $\Cc_j^{(b)}$ denotes the set of all possible symbols in $\Cc$ with bit $j$ being equal to $b$, i.e., $\Cc_i^{(b)} \defeq \{\cv \in \Cc : \ \rv{b}_i(\cv) = b \}$, for $i \in [B]$ and $b \in \{0,1\}$.

In general, the LLR distribution differs between the bit positions, thus have different error protection properties. This can be seen in Fig.~\ref{fig:LLRhist22}, where we depict the LLR histogram of the cell bit and the first two coordinate bits (the remaining two coordinate bits have the same LLR distribution as the first two due to symmetry), given that $0$ was sent in that bit, for the $CS(2,2)$ constellation and single receive antenna. It can also be  observed that the LLR distribution truncated on $[0,\infty)$ (or $(-\infty,0]$) closely fits the exponential distribution (or the flipped exponential distribution, respectively). The LLR distribution fitting can be useful, e.g., to calculate the mutual information $I(\rv{b}_i(\rvMat{x});\LLR_i(\rvMat{Y}))$, which reveals how much information is carried in different bit positions, as well as in the framework of an extrinsic information transfer~(EXIT) chart~\cite{ten2000designing} analysis. 
\begin{figure}[!h] 
	\centering
	\hspace{-.3cm}
	\subfigure[$1^{\text{st}}$ bit (cell bit)]{
		\includegraphics[width=.155\textwidth]{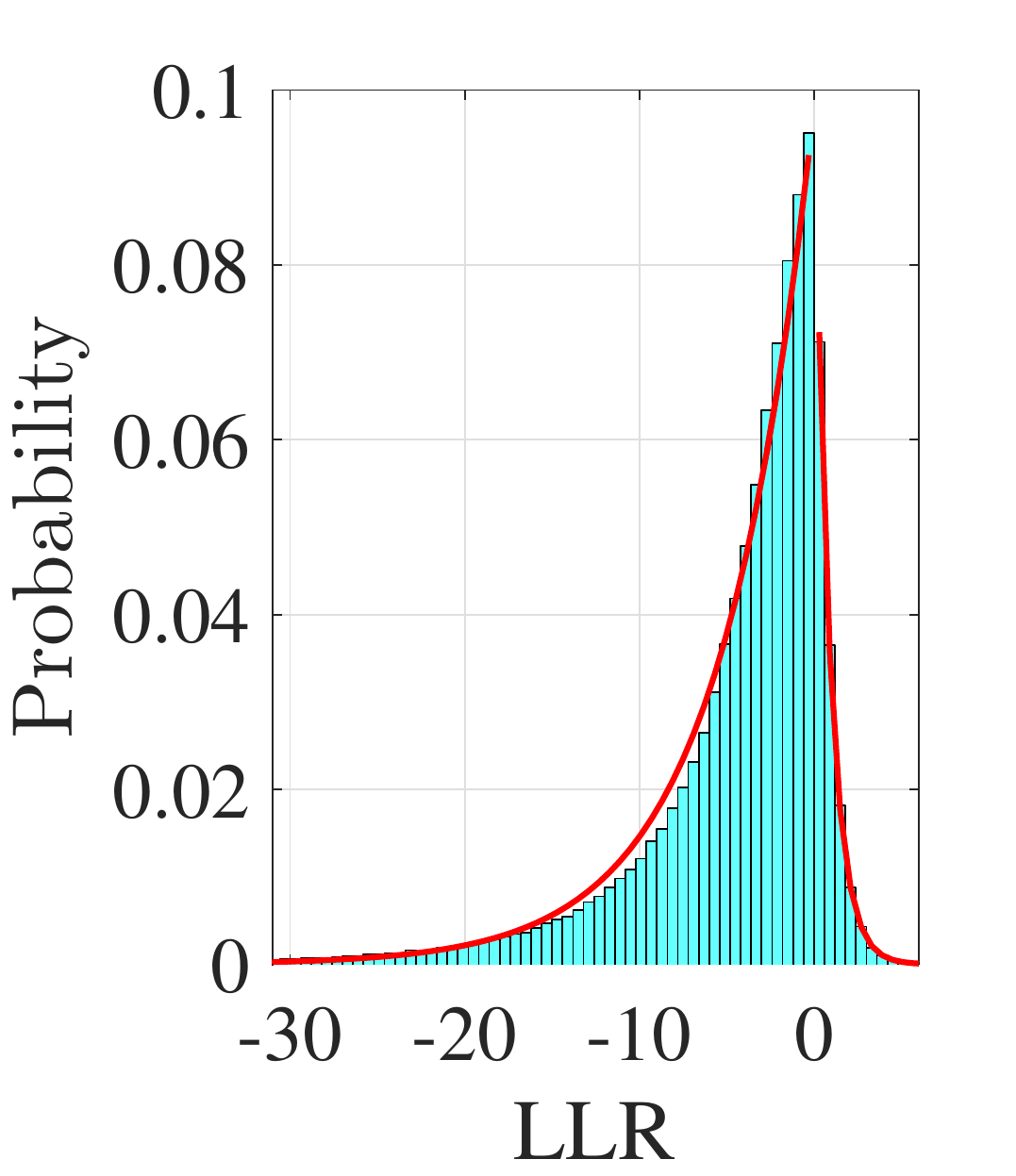}}
	\subfigure[$2^{\text{nd}}$ bit (coordinate bit)]{
		\includegraphics[width=.155\textwidth]{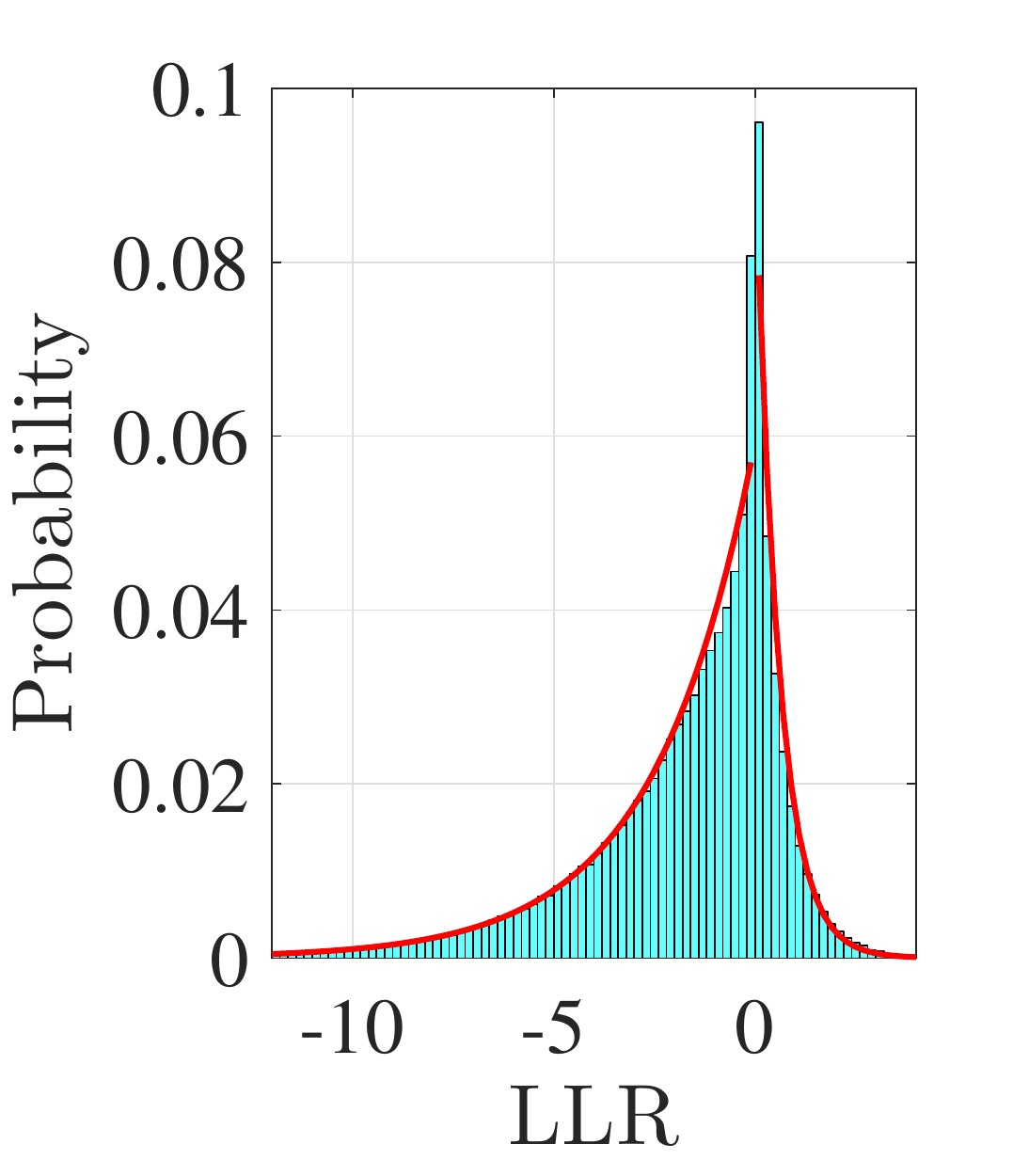}} 
	\subfigure[$3^{\text{rd}}$ bit (coordinate bit)]{
		\includegraphics[width=.155\textwidth]{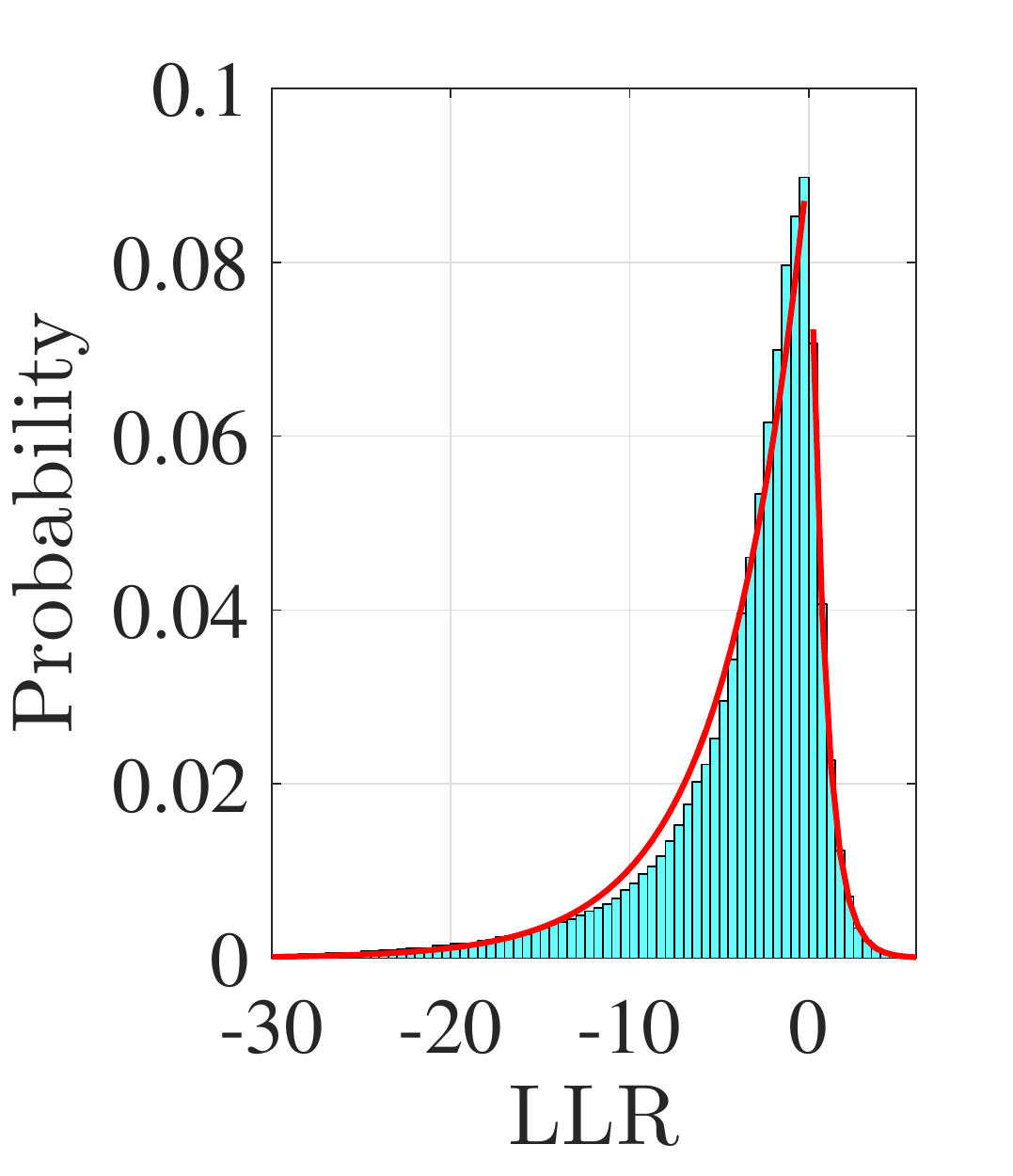}}
	\caption{Histograms of the LLR of the first 3 bits, given that 0 was sent, of $CS(2,2)$ constellation (5 bits/symbol), single receive antenna, and $\SNR = 10$~dB. The red solid lines show the fitted exponential distribution of the LLR truncated on $[0,\infty)$ and the fitted flipped exponential distribution of the LLR truncated on $(-\infty,0]$ obtained by matching the first moment (mean).} 
	\label{fig:LLRhist22}
\end{figure}

\subsubsection{Low-Complexity LLR Computation}
Computing the LLR according to~\eqref{eq:LLRcubesplit} requires enumerating the whole constellation. To avoid this, we propose a low-complexity LLR computation as follows. First, for any real-valued array $x_1,\dots,x_n$, denote by $\Mc_{m}$ the set of $m$ largest values, i.e., $x_j \le x_l$ for all $x_l\in \Mc_{m}$ and $x_j\notin \Mc_{m}$. We have that
\begin{align}
&\log \sum_{j = 1}^n e^{x_j} = \log \sum_{x \in \Mc_{m}} e^{x} + \log\bigg(1+  \frac{\sum_{x_j\notin \Mc_{m}} e^{x_j}}{\sum_{x \in \Mc_{m}} e^{x}}\bigg) \\
&\le \log  \sum_{x \in \Mc_{m}}  e^{x} \notag \\
&\quad+ \log\bigg(1+ (n-m)\exp\Big(\max_{x\notin \Mc_{m}} x - \max_{j\in [n]} x_j\Big)\bigg).
\end{align}
For a sufficiently large value of $\max_{j\in [n]} x_j - \max_{x\notin \Mc_{m}} x$, the second term in the RHS vanishes and $\log \sum_{j = 1}^n e^{x_j}$ can be well approximated\footnote{When $m = 1$, this approximation coincides with the well-known one $\log \sum_j e^{x_j} \approx \max_j x_j$.} by $\log \sum_{x \in \Mc_{m}} e^{x}$. Applying this to~\eqref{eq:LLRcubesplit} yields
\begin{align} \label{eq:LLRapprox}
\LLR_j(\Ym) &\approx \log \sum_{\alphav \in \Mc_{\eta,j}^{(1)}} \exp\left(\frac{\rho T}{1+\rho T}\|\Ym^\H \alphav\|^2\right) \notag \\
&\quad- \log \sum_{\betav \in \Mc_{\eta,j}^{(0)}} \exp\left(\frac{\rho T}{1+\rho T}\|\Ym^\H \betav\|^2\right),
\end{align}
where $\Mc_{\eta,j}^{(b)}$ stores the $\eta$ symbols corresponding to the $\eta$ largest terms in $\big\{\|\Ym^\H \xv\|\big\}_{\xv\in \Cc_j^{(b)}}$ for $b \in \{0,1\}$. 
Observe that one symbol in either $\Mc_{\eta,j}^{(1)}$ or $\Mc_{\eta,j}^{(0)}$ is exactly the output $\hat{\xv}^{\text{ML}}$ of the ML decoder~\eqref{eq:MLoneUser}. The remaining symbols in $\Mc_{\eta,j}^{(1)}$ and $\Mc_{\eta,j}^{(0)}$ are expected to be close to $\hat{\xv}^{\text{ML}}$ since they also lead to high likelihood of $\Ym$. Furthermore, $\hat{\xv}^{\text{ML}}$ can be approximated by the output $\hat{\xv}$ of the greedy decoder. 
Thus, the symbols in  $\Mc_{\eta,j}^{(1)}$ and $\Mc_{\eta,j}^{(0)}$ are expected to be in the neighborhood of the greedy decoder's output $\hat{\xv}$. Therefore, the LLR can be further approximated by replacing $\Mc_{\eta,j}^{(1)}$ and $\Mc_{\eta,j}^{(0)}$ in \eqref{eq:LLRapprox} by the sets of the greedy decoder's output $\hat{\xv}$ and its neighbors as
\begin{align} \label{eq:LLRapprox2}
\LLR_j(\Ym) &\approx \log \sum_{\alphav \in \Bc_j(\hat{\xv},1)} \exp\left(\frac{\rho T}{1+\rho T}\|\Ym^\H \alphav\|^2\right) \notag\\
&\quad- \log \sum_{\betav \in \Bc_j(\hat{\xv},0)} \exp\left(\frac{\rho T}{1+\rho T}\|\Ym^\H \betav\|^2\right),
\end{align}
where {the set $\Bc_j(\cv,b)$ contains $\eta$ nearest symbols to $\cv$ (one of them being $\cv$) with the $j$-th bit in their labels being equal to $b$, i.e,
\begin{align}
&\Bc_j(\cv,b) \notag \\
&\defeq \left\{\hat{\cv}_1,\dots,\hat{\cv}_{\eta} \in \Cc : \Bigg\{
\begin{matrix*}[l]
d(\hat{\cv}_l,\cv) \le  d(\xv,\cv), \\
\quad \forall \xv \in \Cc \setminus \{\hat{\cv}_1,\dots,\hat{\cv}_{\eta}\}, l \in [\eta] \\
\rv{b}_j(\hat{\cv}_1) = \dots = \rv{b}_j(\hat{\cv}_{\eta}) = b 
\end{matrix*}
\right\}
\end{align}
for $\cv \in \Cc$, $b \in \{0,1\}$, and $j \in [B]$.}

The sets $\Bc_j(\cv,b)$ can be precomputed for each $\cv \in \Cc$ prior to communication (with negligible complexity) and stored at the receiver. In this way, the complexity of computing the RHS of~\eqref{eq:LLRapprox2} is only $O(NT\min\{N,T\} + NT\eta)$ ($O(NT\min\{N,T\})$ for the hard detection to find $\hat{\xv}$ and $O(NT\eta)$ for the computation of the RHS of~\eqref{eq:LLRapprox2}). Alternatively, one can look for an approximation of $\Bc_j(\hat{\xv},b)$ (possibly constructed on-the-fly upon detecting $\hat{\xv}$) when the constellation size is too large. The latter option does not require storage but increases the complexity. 

\subsubsection{Multilevel Coding~(MLC) and Multistage Decoding~(MSD)} \label{sec:MLC_MSD}
We propose a MLC-MSD scheme~\cite{Wachsmann1999multilevelCoding} customized for the cube-split constellation as follows. First, each input bit stream is divided into two substreams and each substream is fed into an individual channel encoder. Note that the two encoders can have different code rates. Then, the coded bits are mapped into cube-split symbols by taking the cell bits from the output of the first encoder and the coordinate bits from the output of the second encoder. 
At the receiver, the cell bits are decoded first (using the exact LLR~\eqref{eq:LLRcubesplit} or approximate LLR~\eqref{eq:LLRapprox2}), thus we obtain an estimate of the cell index of each transmitted symbol. Then the LLRs of the coordinate bits are computed based on the received signal and the estimated cell index in a manner similar to~\eqref{eq:LLRapprox2} except that $\Bc_j(\hat{\xv},b)$ is replaced by $\Sc_j({\hat{i}},b)$ where $\hat{i}$ is the estimated cell index of the corresponding transmitted symbol and $\Sc_j({\hat{i}},b)$ is the set of constellation symbols in cell $\cell_{\hat{i}}$ with bit $j$ being equal to $b$, i.e.,
$
\Sc_j(\hat{i},b) \defeq \{\xv\in \Cc \cap \cell_{\hat{i}}: b_j(\xv) = b \}.
$
This decoding structure is also similar to a turbo equalization receiver~\cite{Tuchler2002turboEqualization}.

\section{Performance Evaluation} \label{sec:performance}
We evaluate numerically the performance of our cube-split constellation in comparison with other constellations and a baseline (coherent) pilot-based scheme.

\subsection{A Baseline Pilot-Based Scheme} \label{sec:pilotscheme}
We consider a baseline scheme based on channel training~\cite{Hassibi2003howmuchtraining}. 
The transmitted signal is
\begin{align}
\rvVec{x} = (\rho T)^{-1/2}\left[\sqrt{\rho_\tau} \ \ \sqrt{\rho_d} \rvVec{x}_d^\T\right]^\T,
\end{align}
{i.e., the first symbol is constant and known to the receiver,} the data symbol vector $\rvVec{x}_d = [\rv{x}_2 \dots \rv{x}_T]^\T$ is normalized s.t. $\E[\rvVec{x}_d\rvVec{x}_d^\H] = \Id_{T-1}$. The power factors $\rho_\tau$ and $\rho_d$ satisfy $\rho_\tau + (T-1)\rho_d = \rho T$ and can be optimized. The received signal can be written as $\rvMat{Y} = [\rvVec{y}_\tau \ \rvMat{Y}_d^\T]^\T$ where $\rvVec{y}_\tau = \sqrt{\rho_\tau} \rvVec{h} + \rvVec{z}_\tau$ and $\rvMat{Y}_d = \sqrt{\rho_d} \rvVec{x}_d \rvVec{h}^\T  + \rvMat{Z}_d$ are the received signals in the training phase and data transmission phase, respectively. The receiver uses minimum-mean-square-error (MMSE) channel estimation $\hat{\rvVec{h}} = \frac{\sqrt{\rho_\tau}}{1+\rho_\tau} \rvVec{y}_\tau \sim\Cc\Nc\big(\mathbf{0},\frac{\rho_\tau}{1+\rho_\tau}\Id_N\big)$. {Let $\bar{\rvVec{h}} \defeq \hat{\rvVec{h}}\big(\frac{1}{N}\E\big[\|\hat{\rvVec{h}}\|^2\big]\big)^{-1/2} \sim \Cc\Nc(\mathbf{0},\Id_N)$ be the normalized estimate. 
From~\cite[Thm.3]{Hassibi2003howmuchtraining}, a lower bound on the achievable rate of this pilot-based scheme with i.i.d. Gaussian input $\rvVec{x}_d \sim \Cc\Nc(\mathbf{0},\frac{1}{T-1}\Id_{T-1})$ 
is given as
\begin{align}
&R_{\rm pilot} (\rho,N,T) 
\defeq  \Big(1-\frac{1}{T}\Big) \E[\log_2\left(1 + \rho_{\rm eff} \|\bar{\rvVec{h}}\|^2 \right)]
\\&= \frac{T-1}{T} \log_2(e) \sum_{n=1}^{N} \frac{(N-1)!}{(N-n)!} \left(-\frac{1}{\rho_{\rm eff}}\right)^{N-n} \notag\\&\quad\times
\bigg[e^{1/\rho_{\rm eff}} E_1\left(\frac{1}{\rho_{\rm eff}}\right) + \sum_{m=1}^{N-n}(m-1)!\left(-\rho_{\rm eff}\right)^{m}\bigg], \label{eq:pilot_capacity}
\end{align}
where $\rho_{\rm eff} = \frac{\rho_\tau \rho_d}{1+\rho_\tau + \rho_d}$, $E_1(x) \defeq \int_{x}^\infty \frac{e^{-t}}{t}\dif t$ is the exponential integral function~{\cite[Eq.(5.1.1)]{abramovitz}}, and \eqref{eq:pilot_capacity} is derived using~\cite[4.337.5]{gradshteyn2007table}.
The optimal power allocation is given by $\rho_\tau = \rho$ if $T \!=\! 2$ and $\rho_\tau \!=\! \frac{\sqrt{T-1+\rho T}\left(\sqrt{(T-1)(1+\rho T)}-\sqrt{T-1+\rho T}\right)}{T-2}$ if $T\!>\!2$.} 

Let $\tilde{\rvVec{h}} = \rvVec{h} - \hat{\rvVec{h}}$ be the channel estimation error, then $\tilde{\rvVec{h}} \sim \Cc\Nc\big(\mathbf{0},\frac{1}{1+\rho_\tau}\Id_N\big)$ and $\tilde{\rvVec{h}}$ and $\hat{\rvVec{h}}$ are uncorrelated. The output can be written as $\rvMat{Y}_d = \sqrt{\rho_d} \rvVec{x}_d \hat{\rvVec{h}}^\T + \hat{\rvMat{Z}}_d$, where $\hat{\rvMat{Z}}_d = \sqrt{\rho_d} \rvVec{x}_d \tilde{\rvVec{h}}^\T + \rvMat{Z}_d$. Given the input, the rows of $\hat{\rvMat{Z}}_d$ are independent and the $j$-th row follows $\Cc\Nc\left(\mathbf{0},\big(1+\frac{\rho_d|\rv{x}_j|^2}{1+\rho_\tau}\big)\Id_N\right)$, $j\in\{2,\dots,T\}$. Thus, the likelihood function of the output at slot $j\in \{2,\dots,T\}$ is
\begin{multline} \label{eq:likelihood_pilot}
p_{\rvMat{Y}_{[j]} | \rv{x}_j, \hat{\rvVec{h}}}(\Ym_{[j]} | {x}_j, \hat{\hv}) \\= {\pi^{-N} \left(1+\frac{\rho_d|{x}_j|^2}{1+\rho_\tau}\right)^{-N}} \exp\bigg(-\frac{\big\|\Ym_{[j]} - \sqrt{\rho_d} x_j \hat{\hv}^\T\big\|^2}{1+\frac{\rho_d|{x}_j|^2}{1+\rho_\tau}} \bigg),
\end{multline}
where $\rvMat{Y}_{[j]}$ denotes the $j$-th row of $\rvMat{Y}$.

In practice, the data symbols in $\rvVec{x}_d$ are normally taken from finite scalar constellations such as QAM or PSK in order to reduce the complexity of the ML decoder based on~\eqref{eq:likelihood_pilot}. A sub-optimal method consists in linear equalization followed by component-wise scalar demapping. With zero-forcing (ZF) or MMSE equalizer, the equalized symbols are respectively
\begin{align} \label{eq:equalization}
\hat{\rvVec{x}}_d^{\rm zf} = \frac{\rvMat{Y}_d}{\sqrt{\rho_d}} \frac{\hat{\rvVec{h}}^*}{\|\hat{\rvVec{h}}\|^2}, ~~ \text{or} ~~
\hat{\rvVec{x}}_d^{\rm mmse} = \frac{\rvMat{Y}_d}{\sqrt{\rho_d}} \frac{\hat{\rvVec{h}}^*}{\|\hat{\rvVec{h}}\|^2+1/\rho_d}.
\end{align}

{The LLR of bit $\rv{b}_i$ given $\rvMat{Y} = \Ym$ and the channel estimate $\hat{\rvVec{h}} = \hat{\hv}$ is calculated as
\begin{align}
\LLRp_j(\Ym,\hat{\hv}) = \log \frac{\sum_{\alpha \in \Qc_j^{(1)}}p_{\rvMat{Y}|\rv{x}_{\{j\}},\hat{\rvVec{h}}}(\Ym|\alpha,\hat{\hv})}{\sum_{\beta \in \Qc_j^{(0)}}p_{\rvMat{Y}|\rv{x}_{\{j\}},\hat{\rvVec{h}}}(\Ym|\beta,\hat{\hv})}, \label{eq:LLRpilot}
\end{align} 
where $\Qc_j^{(b)}$, $b \in \{0,1\}$, denotes a subset of the chosen scalar constellation (e.g., QAM) with bit $j$ being $b$, $\rv{x}_{\{j\}}$ denotes the symbol accounting for bit $\rv{b}_j$, and $p_{\rvMat{Y}|\rv{x}_{\{j\}},\hat{\rvVec{h}}}$ is given in~\eqref{eq:likelihood_pilot}.}\footnote{One can also compute the LLR based on the equalized symbols $\hat{\rv{x}}_j$. When $N=1$, the likelihood function $p_{\hat{\rv{x}}_j | \rv{x}_j}$ for ZF equalized symbols \eqref{eq:equalization} can be derived explicitly using Lemma~\ref{lem:Cauchy-Gaussian} as
	$
	p_{\hat{\rv{x}}^{\rm zf}_j | \rv{x}_j}(\hat{x}^{\rm zf}_j | x_j) = \frac{\frac{1+\rho_\tau}{\rho_\tau\rho_d(T-1)}+\frac{|\hat{{x}}^{\rm zf}_j|^2}{\rho_\tau}}{\pi\big(\frac{1+\rho_\tau}{\rho_\tau\rho_d(T-1)}+\frac{|{x}_j|^2}{\rho_\tau} + \left|\hat{{x}}^{\rm zf}_j - {x}_j\right|^2\big)^2}, j\in \{2,\dots,T\}.
	$}

In the remainder of this section, we 
compare different schemes with the same transmission rate of $B$ bits/symbol. 
{Having observed in Fig.~\ref{fig:mindist} that the Fourier constellation~\cite{Hochwald2000systematicDesignUSTM} and the exp-map constellation~\cite{Kammoun2007noncoherentCodes} have similar or higher packing efficiency than the coprime-PSK constellation~\cite{Zhang2011full_diversity_blind_ST_block_codes} and the multi-layer constellation~\cite{AttiahISIT2016systematicDesign}, and to keep the comparison clear, we hereafter consider only the two former constellations.}

\subsection{Achievable Data Rate}
In {Fig.~\ref{fig:rate}{(a)}}, we compare the achievable rate (computed as in~\eqref{eq:achievableRate}) of cube-split constellation with the rate of the numerically optimized constellation and the high-SNR capacity $C(\rho,N,T)$ given in~\eqref{eq:noncoherentCapacity} for $T=2$ and single receive antenna. We also include the rate lower bound $R_{\rm pilot} (\rho,N,T)$ of a pilot-based scheme with Gaussian input given in~\eqref{eq:pilot_capacity}, and the achievable rate of the pilot-based scheme with QAM input. The cube-split constellation can achieve almost the same rate as the numerically optimized constellation and a higher rate than the pilot-based scheme with QAM input at a given SNR. For example, at $25$~dB, the cube-split constellation can achieve about $0.3$ bits/channel use higher than the rate achieved with the pilot-based scheme. Furthermore, the achievable rate of a large cube-split constellation approaches the high-SNR capacity $C(\rho,N,T)$.

Next, in {Fig.~\ref{fig:rate}{(b)}}, we plot the achievable rate of the cube-split constellation, the numerically optimized constellation, the Fourier constellation~\cite{Hochwald2000systematicDesignUSTM}, and the exp-map constellation~\cite{Kammoun2007noncoherentCodes}, and the pilot-based scheme with QAM input for $T = 4$ and $N = 2$. Again, the rate achieved with cube-split constellation is close to the rate achieved with the numerically optimized constellation and higher than that of other structured constellations and the pilot-based scheme. 
\begin{figure}[!h]
	\centering
	\subfigure[$T = 2$, $N =1$, $B \in \{3, 5, 7, 9 , 11, 13\}$]{\includegraphics[width=.48\textwidth]{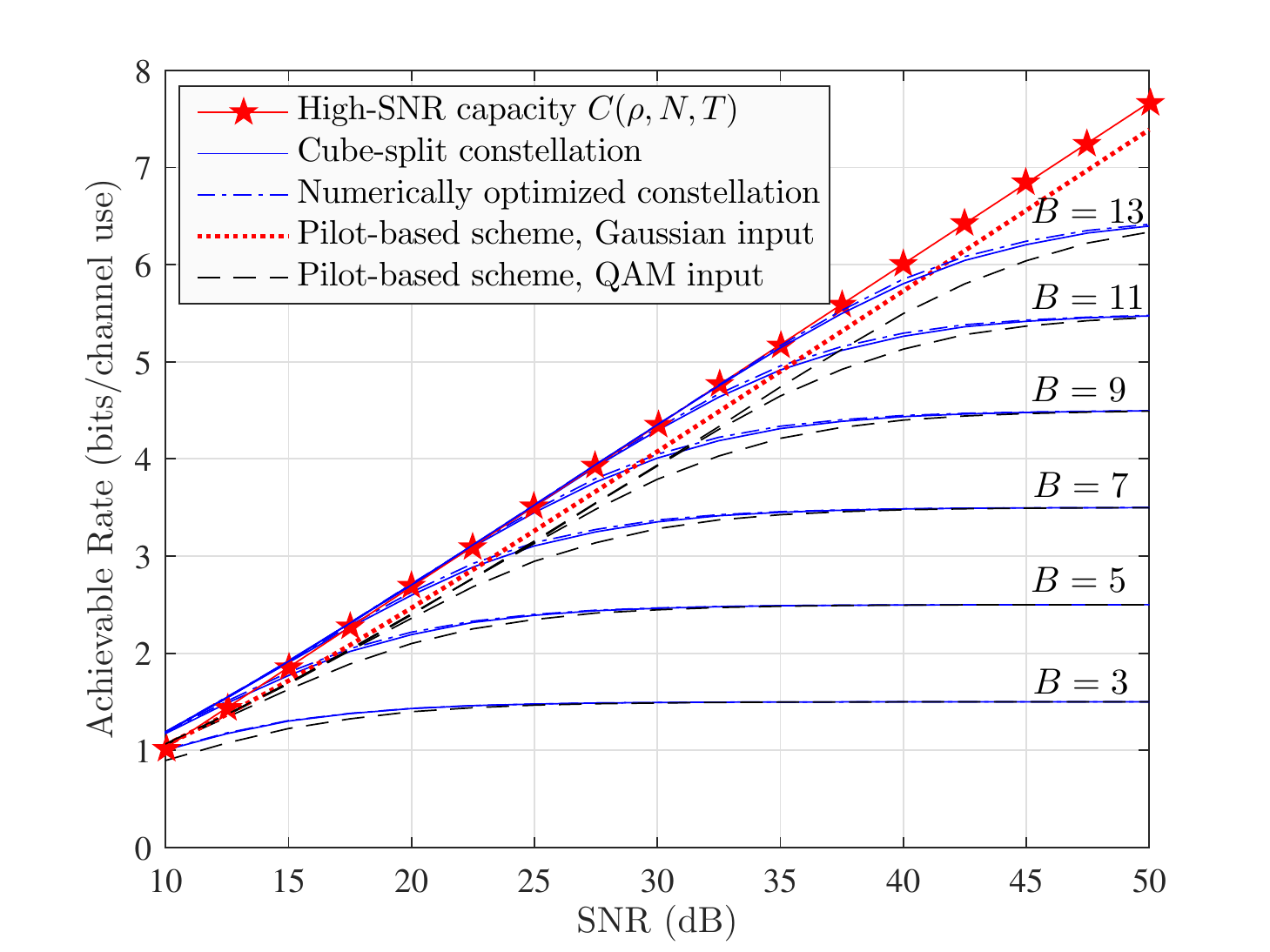}}
	\subfigure[$T = 4$, $N = 2$, $B \in \{8,14\}$]{\includegraphics[width=.48\textwidth]{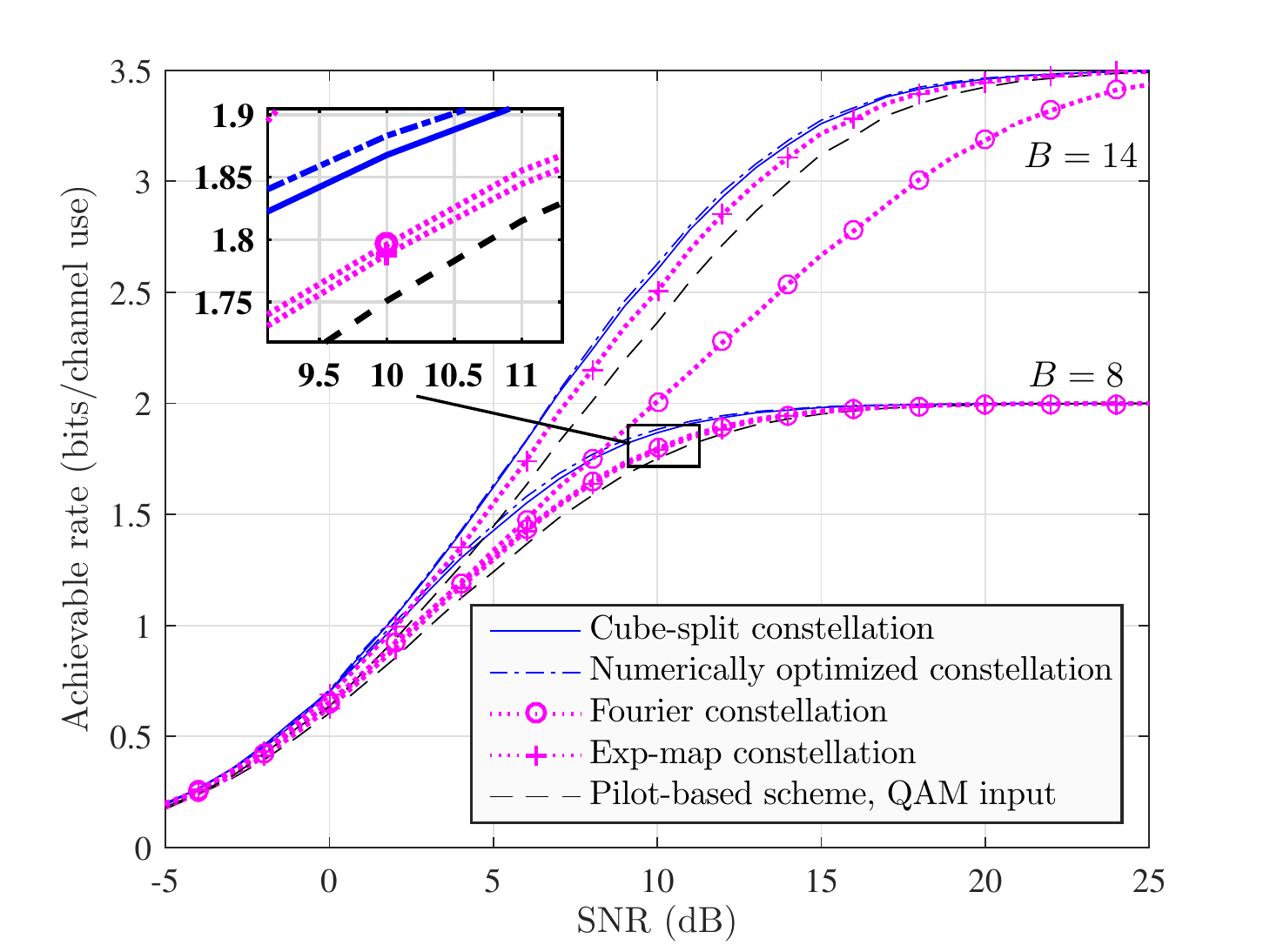}}
	\caption{The achievable rate~\eqref{eq:achievableRate} of the cube-split constellation in comparison with the channel capacity given in~\eqref{eq:noncoherentCapacity}, and the rate achieved with the numerically optimized constellation, other structured constellations, and the pilot-based scheme with Gaussian input~\eqref{eq:pilot_capacity} or	QAM input for $T\in \{2,4\}$, $N = T/2$, and different transmission rate $B$ bits/symbol.}
	\label{fig:rate}
\end{figure}

\subsection{Error Rates of Uncoded Constellations}

In {Fig.~\ref{fig:SER_BER_T2T4}{(a)} and Fig.~\ref{fig:SER_BER_T2T4}{(b)}}, we plot the symbol error rate~(SER) of the cube-split constellation (with ML or greedy decoder), the pilot-based scheme with QAM input, the numerically optimized constellation (with ML decoder), the Fourier constellation~\cite{Hochwald2000systematicDesignUSTM} (with ML decoder), and the exp-map constellation~\cite{Kammoun2007noncoherentCodes} (with ML or {greedy decoder}). 
{In Fig.~\ref{fig:SER_BER_T2T4}{(c)}, we show the corresponding bit error rate (BER) but omit the Fourier and the numerically optimized constellations for their lack of an effective binary labeling scheme. The cube-split constellation uses the Gray-like labeling in Sec.~\ref{sec:labeling}, the exp-map constellation takes the Gray label of the QAM vector $\qv$ for the mapped symbol $\cv$, and the pilot-based scheme uses Gray labels of the QAM symbols.} We observe that the greedy decoder for the cube-split constellation achieves near-ML performance. 
The cube-split constellation achieves performance close to the numerically optimized constellation, and outperforms other structured constellations and the pilot-based scheme. 
\begin{figure}[!h] 
	\centering
	\subfigure[Symbol error rate, $T=2,N=1$]{\includegraphics[width=.47\textwidth]{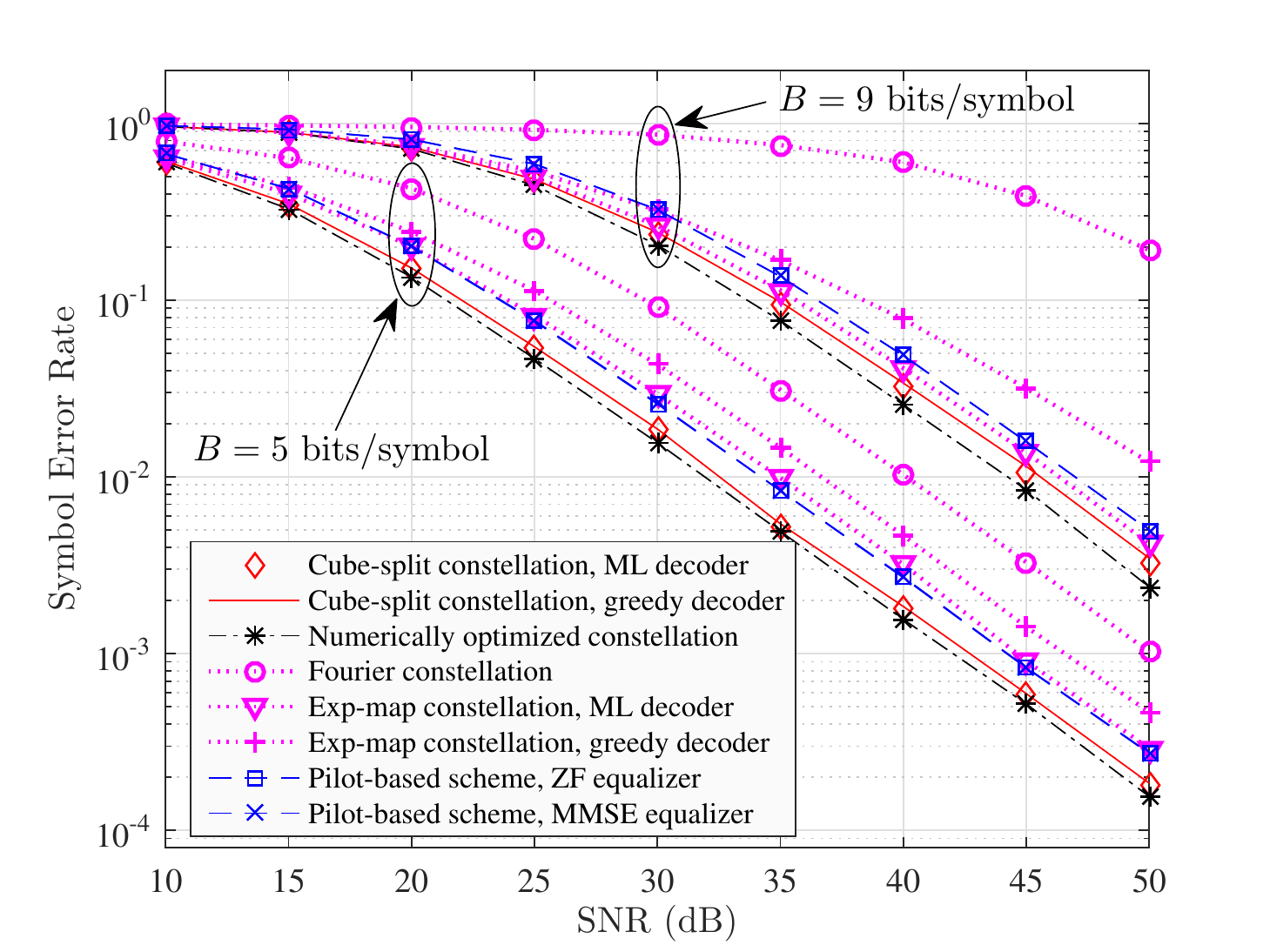}}
	\subfigure[Symbol error rate, $T=4, N=2$]{\includegraphics[width=.48\textwidth]{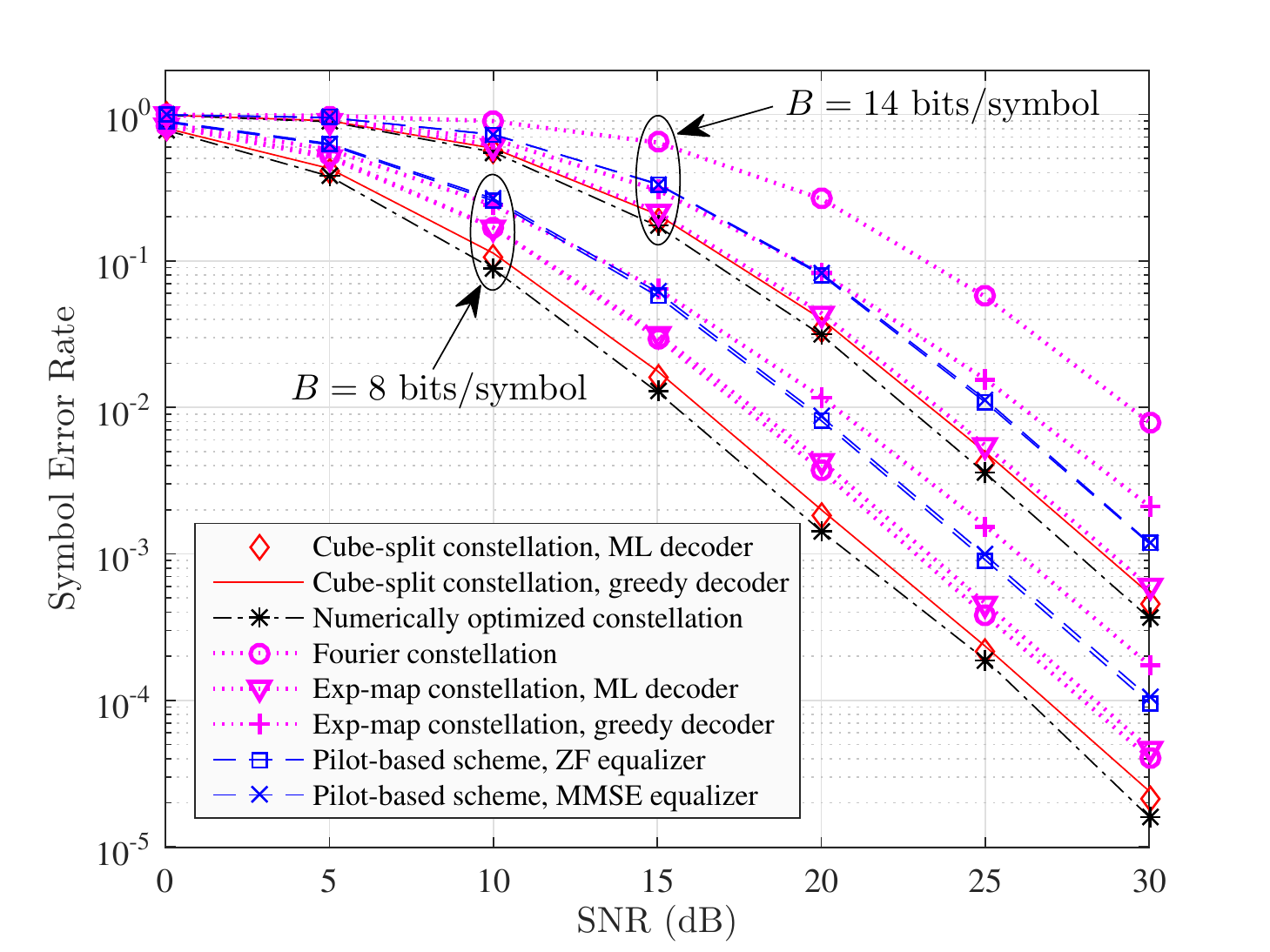}}
	\subfigure[Bit error rate]{\includegraphics[width=.48\textwidth]{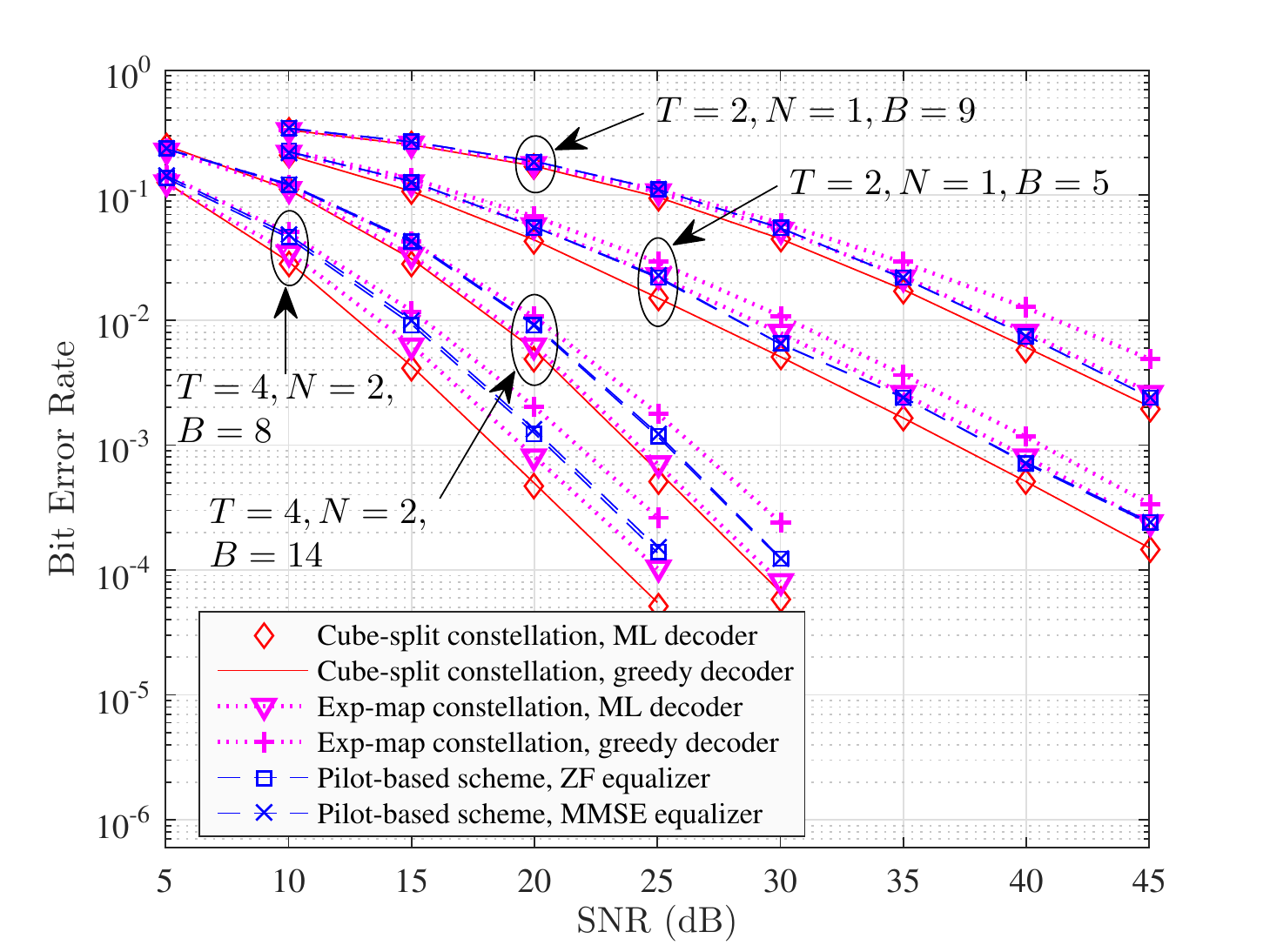}}
	\caption{The error rates of the cube-split constellation in comparison with the numerically optimized constellation, other structured constellations, and the pilot-based scheme with QAM input for $T \in \{2,4\}, N = T/2$, and different transmission rate $B$.}
	\label{fig:SER_BER_T2T4}
\end{figure}

In Fig.~\ref{fig:SER_BER_T8T16}, we show the SER and BER of the $CS(T,1)$ constellation, the exp-map constellation, and the pilot-based scheme with QAM input  for $T \in \{8,16\}$, $N=T/2$, and $B = \log_2(T) + 2(T-1)$. Note that for this large $B$, the numerically optimized constellation and the Fourier frequencies optimization for the Fourier constellation become infeasible. We see that in this regime, the gain of the cube-split constellation w.r.t. the exp-map constellation and the pilot-based scheme is more significant.
\begin{figure}[!h] 
	\centering
	\subfigure[Symbol error rate]{\includegraphics[width=.48\textwidth]{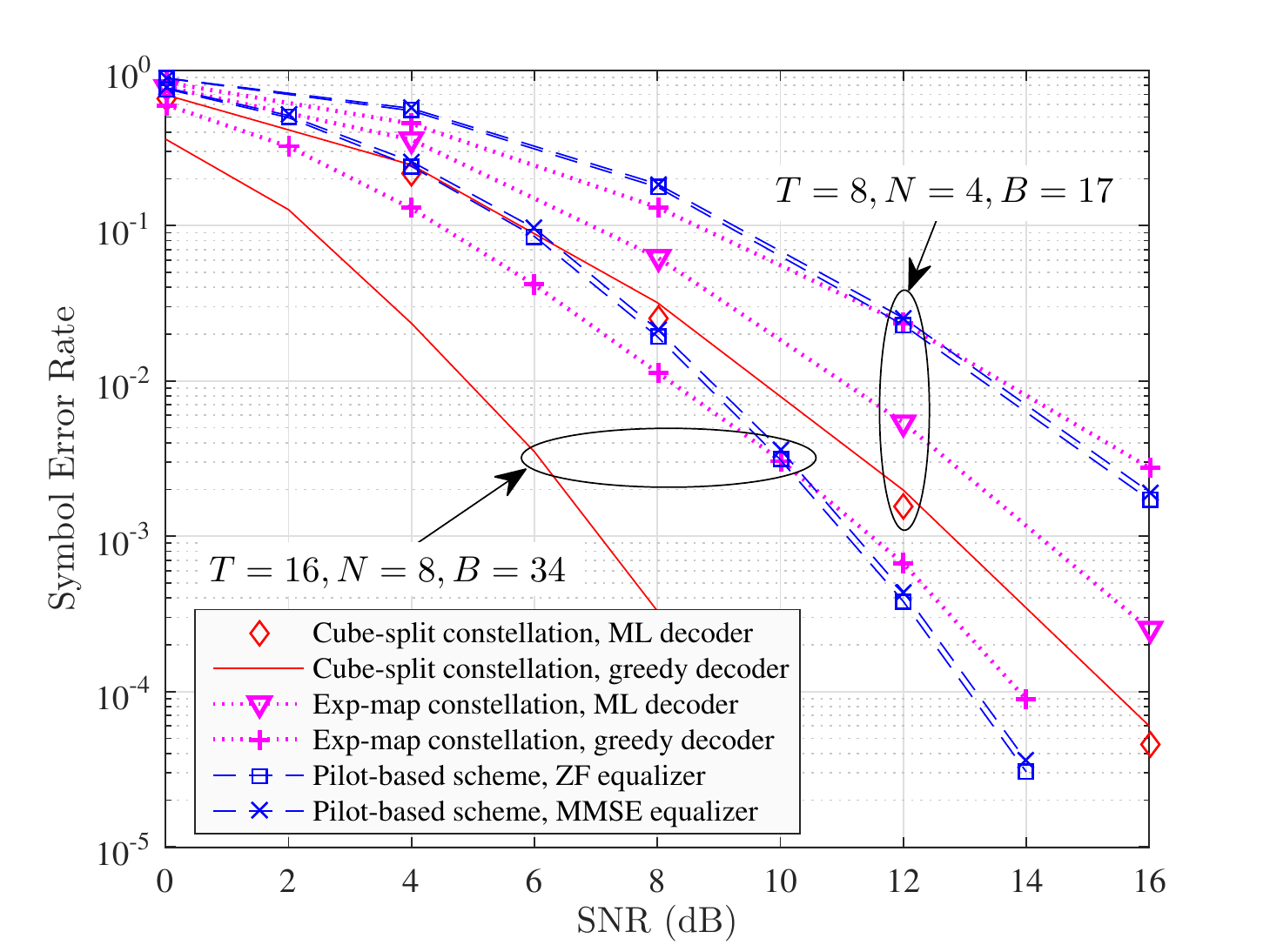}}
	\subfigure[Bit error rate]{\includegraphics[width=.48\textwidth]{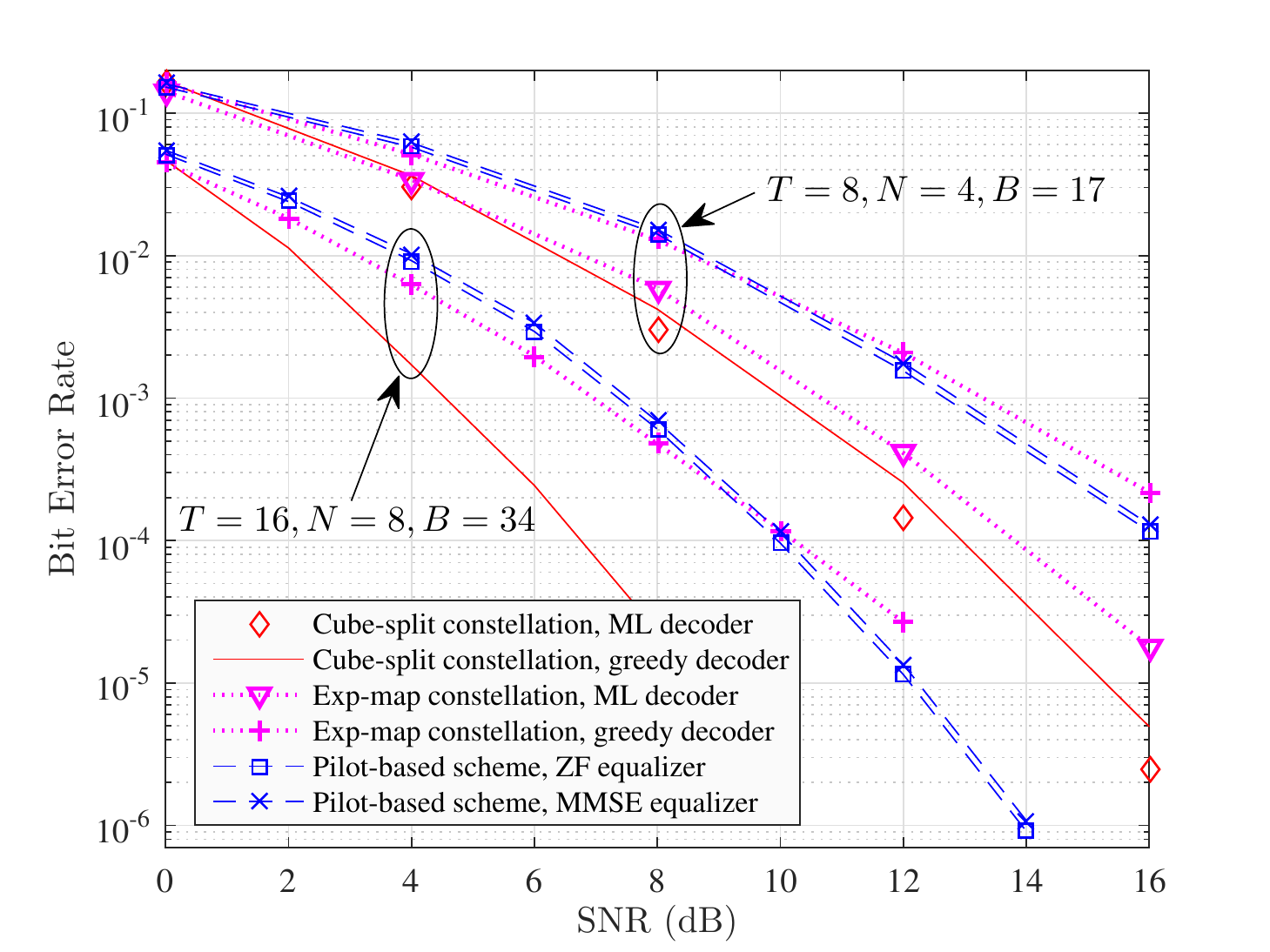}}
\caption{The error rates of the $CS(T,1)$ constellation in comparison with the exp-map constellation and the pilot-based scheme with QAM input for $T \in \{8,16\}, N = T/2$, and $B = \log_2(T) + 2(T-1)$.}
	\label{fig:SER_BER_T8T16}
\end{figure}

\subsection{Performance with Channel Coding}
Next, we integrate a systematic parallel concatenated rate-$1/3$ standard turbo code~\cite{3GPP_TS36_212}. 
The coded bits are mapped into symbols using the Gray-like labeling scheme described in Section~\ref{sec:labeling}. The turbo decoder calculates the LLR of the coded bits as in \eqref{eq:LLRcubesplit} or \eqref{eq:LLRapprox2}, and performs 10 decoding iterations for each packet. {For the pilot-based scheme, the LLR is computed as in~\eqref{eq:LLRpilot}.}

Fig.~\ref{fig:BER_turboCode} presents the BER of the coded cube-split constellation (with exact LLR \eqref{eq:LLRcubesplit} or approximate LLR~\eqref{eq:LLRapprox2}) as compared to the coded pilot-based scheme 
when the turbo encoder is applied in each packet of $640$ bits. {We also consider the MLC-MSD scheme 
in which the same turbo encoder is used in both coding levels and the exact or approximate cell-bit LLR is used in the first decoding stage.
For the $T = 2$ case (Fig.~\ref{fig:BER_turboCode}{(a)}), the BER of the cube-split constellation with approximate LLR is close to the BER with exact LLR. With the considered turbo code, the cube-split constellation outperforms the pilot-based scheme: the power gain is about $2.5$~dB for the same transmission rate of $9$ bits/symbol. 
On the other hand, for the $T \in \{4,8\}$ case (Fig.~\ref{fig:BER_turboCode}{(b)}), with the same (single-level) turbo code, the pilot-based scheme performs better than our cube-split constellation. 
However, with the MLC-MSD scheme, the performance of the cube-split constellation is greatly improved and can be better than that of the pilot-based scheme. 
This is because as the number of cells increases, the reliability of the cell bits becomes more crucial to the overall performance, and by protecting the cell bits with an individual code then using the estimated cell bits as a basis for decoding the coordinate bits, the error rate is reduced.
}
\begin{figure}[!h] 
	\centering
	\subfigure[$T = 2, N = 1, \eta = 5$]{\includegraphics[width=.49\textwidth]{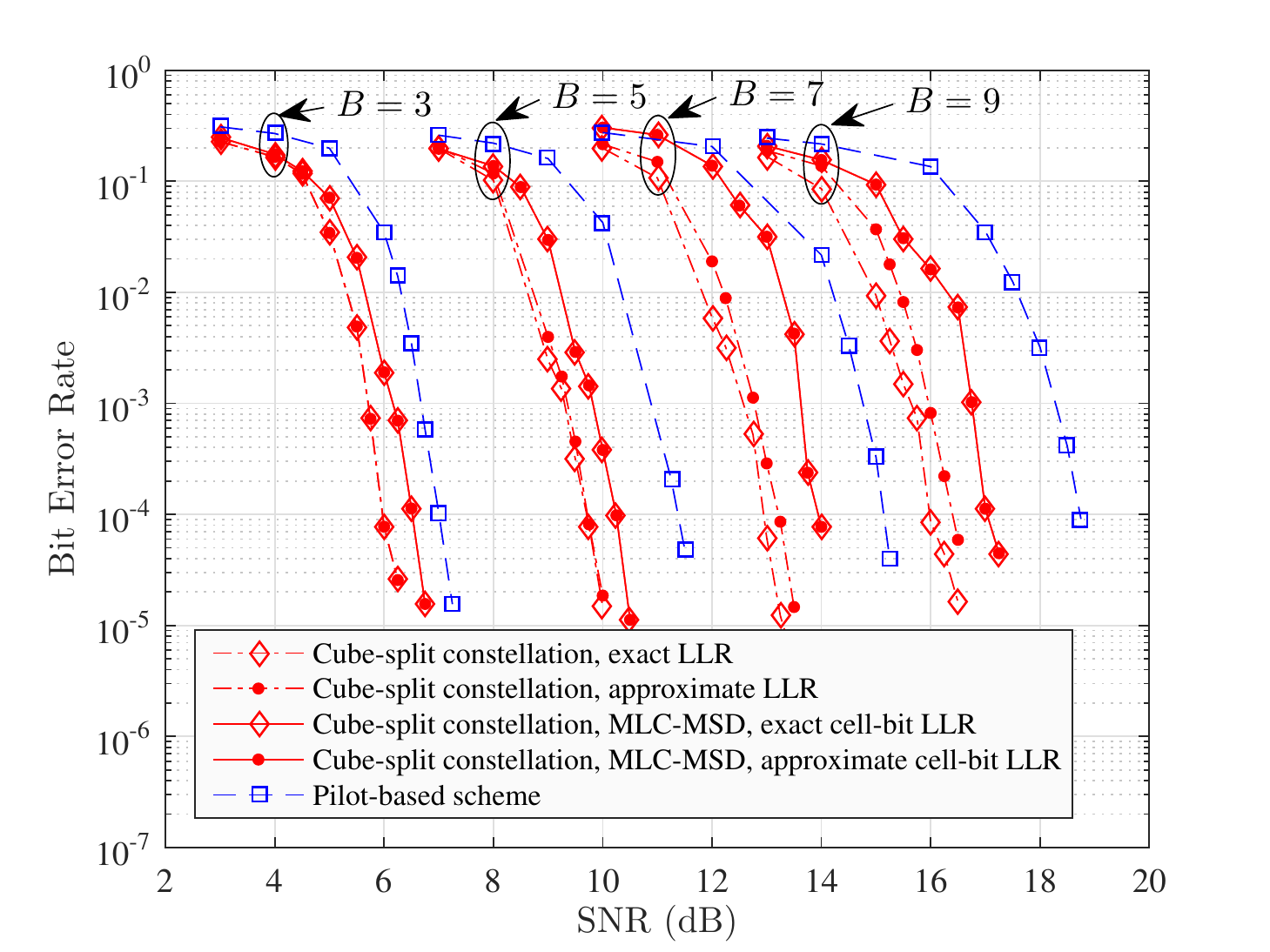}}
	\subfigure[$T \in \{4,8\}, N = T/2$]{\includegraphics[width=.49\textwidth]{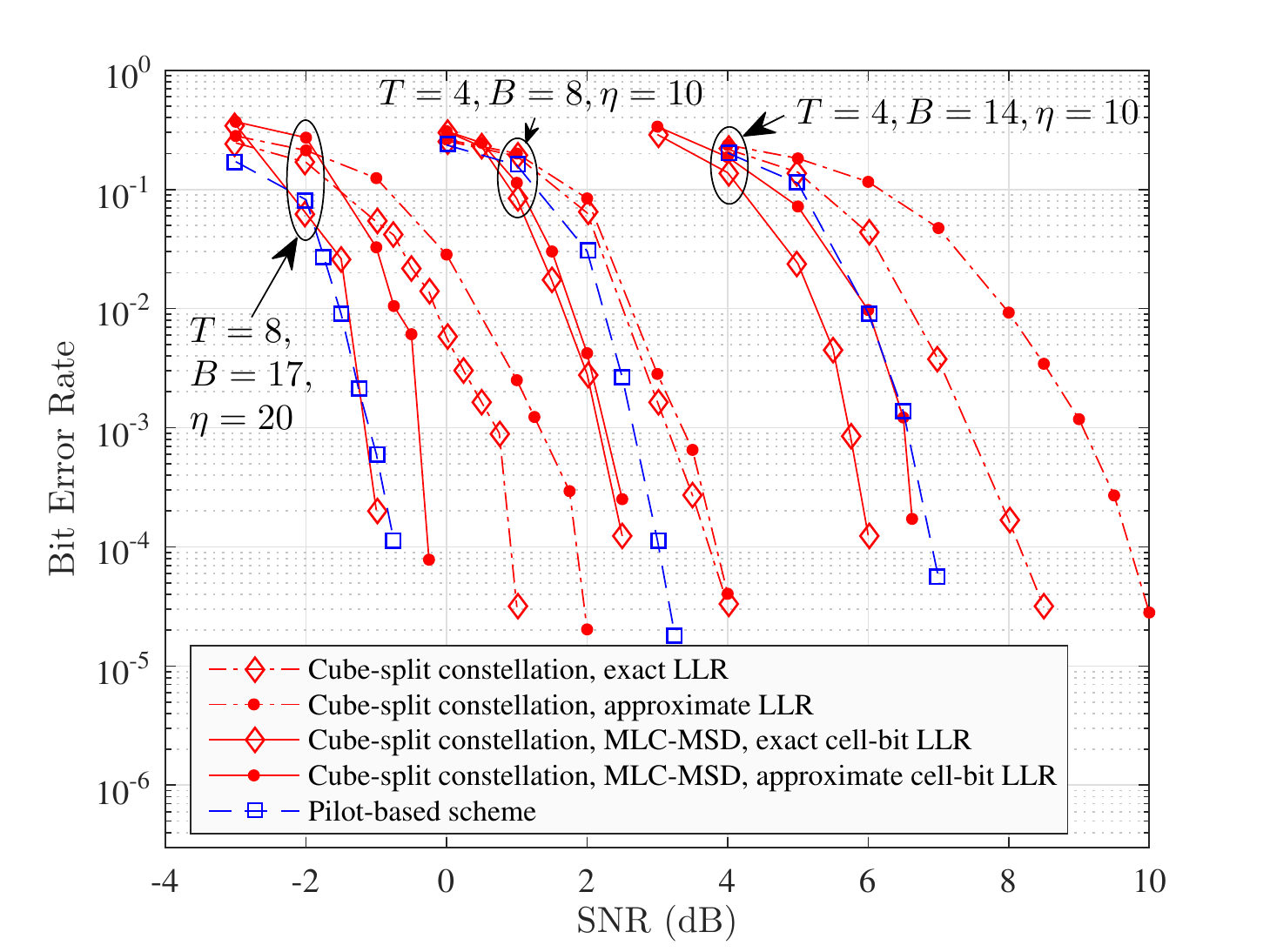}}
	\caption{The bit error rate of the cube-split constellation in comparison with the pilot-based scheme with turbo codes for $T\in \{2,4,8\}, N=T/2,$ and different transmission rate $B$.}
	\label{fig:BER_turboCode}
\end{figure}

\section{Conclusion} \label{sec:conclusion} 
We proposed a novel Grassmannian constellation for non-coherent SIMO communications. The structure of this constellation allows for on-the-fly symbol generation, a simple yet effective binary labeling, low-complexity symbol decoder and bit-wise LLR computation, {and an efficient association with a multilevel coding-multistage decoding scheme}. Analytical and numerical results show that this constellation is close to optimality in terms of packing properties, has larger minimum distance than
other structured constellations in the literature. For small coherence time/large constellation size, it outperforms the coherent pilot-based approach in terms of error rates with/without channel codes and achievable data rate under Rayleigh block fading channel.

\appendix
\subsection{ The Extension to the MIMO Case} \label{app:extension}
In general, if the transmitter has $M \leq \frac{T}{2}$ antennas, we may consider constellation symbols belonging to the Grassmannian $G(\CC^T,M)$ represented by $T\times M$ truncated unitary matrices. To extend the cube-split design to the case $M>1$, we would follow two essential steps: partitioning the Grassmannian into cells and defining a mapping from an Euclidean space onto a cell.  
To partition $G(\CC^T,M)$, we consider a set of reference subspaces $\{\Em_1,\dots,\Em_{V}\}$ defining an initial constellation in $G(\CC^T,M)$ and its associated Voronoi cells
\begin{equation}
\begin{split}
&\cell_{i} \\
&\defeq \{\Xm \in G(\CC^T,M) : d(\Em_{i},\Xm) \le d(\Em_{j},\Xm), \forall j \in [V] \setminus \{i\} \} \\
&= \{\Xm \in G(\CC^T,M) : \|\Em_{i}^\T \Xm\|_F \ge \|\Em_{j}^\T \Xm\|_F, \forall j \in [V] \setminus \{i\} \},
\end{split}
\end{equation}
where $d(\Qm_1,\Qm_2) \defeq \sqrt{M-\trace\{\Qm_1^\H \Qm_2 \Qm_2^\H \Qm_1\}}$ is the chordal distance between the subspaces spanned by the columns of $T\times M$ truncated unitary matrices $\Qm_1$ and $\Qm_2$. 
The problem of choosing the initial constellation $\Em_1,\dots,\Em_{V}$ and designing a mapping from an Euclidean space onto each cell $\cell_{i}$ preserving a property similar to Property~\ref{prop:uniform} is not evident and left as perspective for future work.
In particular, it seems difficult to describe each Voronoi region $\cell_{i}$ as it is the case for the regions of Grassmannian of lines where the condition $\|\Em_{i}^\T \Xm\|_F \ge \|\Em_{j}^\T \Xm\|_F$ expressed on the canonical basis simply translates into a coordinate-wise condition $\big|\frac{x_j}{x_i}\big|<1$.

\subsection{Mathematical Preliminaries} \label{sec:mathPre}
\begin{definition}[Univariate complex Cauchy distribution] \label{def:Cauchy}
	Let $\mu \in \CC$ and $\gamma > 0$. The probability distribution with density
	$
	f(x) = \frac{1}{\pi \gamma}\Big[1+\Big(\frac{|x-\mu|}{\gamma}\Big)^2\ \Big]^{-2}, ~ x\in \CC,
	$
	is called the univariate complex Cauchy distribution with location $\mu$ and scale $\gamma$, denoted by ${\rm Cauchy}(\mu,\gamma)$.
\end{definition}
A multivariate version of the complex Cauchy distribution is given in~\cite[Eq.(2)]{Auderset2005Angular_Gaussian}.
\begin{lemma} \label{lem:Cauchy-Gaussian}
	${\rm Cauchy}(\mu,\gamma)$ is the distribution of the ratio $\frac{\rv{u}_1}{\rv{u}_2}$ where 
	$[\rv{u}_1 \ \rv{u}_2]^\T \sim \Cc\Nc(\mathbf{0},\Sigmam)$ with $\Sigmam = c_0\Big[\begin{smallmatrix}
	\gamma + \|\mu\|^2 & \mu \\ \mu^* & 1
	\end{smallmatrix}\Big]$ for some constant $c_0$. Specifically, if $\rv{u}_1 \sim \Cc\Nc(0,\sigma_1^2)$, $\rv{u}_2\sim \Cc\Nc(0,\sigma_2^2)$, and $\frac{\E[\rv{u}_1\rv{u}_2^*]}{\sigma_1\sigma_2} \eqdef \beta$ then $\mu = \beta \frac{\sigma_1}{\sigma_2}$ and $\gamma = (1-|\beta|^2)\frac{\sigma_1^2}{\sigma_2^2}$.
\end{lemma}
This relation between the Cauchy and Gaussian distributions was stated in, e.g.,~\cite{Auderset2005Angular_Gaussian,Baxley2010complexGaussianRatio}.

\begin{lemma} \label{lem:Cauchy-Grassmann} 
	The span of $\rvVec{x} = [\rv{x}_1 \ \rv{x}_2]^\T$ is uniformly distributed in $\cell_{1} \defeq \big\{\xv\in G(\CC^2,1): |x_{1}|>|x_2|\big\}$ if and only if the quotient $\frac{\rv{x}_2}{\rv{x}_1}$ follows a ${\rm Cauchy}(0,1)$ distribution truncated on $\{t\in\mathbb{C}: |t|<1\}$. 
\end{lemma}
\begin{proof}
	From Lemma~\ref{lem:Cauchy-Gaussian}, if $\frac{\rv{x}_2}{\rv{x}_1} \sim {\rm Cauchy}(0,1)$ then $\frac{\rv{x}_2}{\rv{x}_1}$ is identically distributed to $\frac{\rv{y}_2}{\rv{y}_1}$ with $\rvVec{y} = [\rv{y}_1\ \rv{y}_2]^\T \sim \Cc\Nc(\mathbf{0},\Id_2)$. Since Grassmannian symbols are defined up to a scaling, $\Span{\rvVec{x}}$ has the same distribution as $\Span{\rvVec{y}}$. On the other hand, according to~\cite{Auderset2005Angular_Gaussian}, $\Span{\rvVec{y}}$ is uniformly distributed in $G(\CC^2,1)$. Therefore, $\Span{\rvVec{x}}$ is uniformly distributed in $G(\CC^2,1)$. Furthermore, $\frac{\rv{x}_2}{\rv{x}_1} \in \{t\in\mathbb{C}: |t|<1\}$ means $\Span{\rvVec{x}} \in \cell_1$. The converse follows from the bijectivity of the mapping $\frac{\rv{x}_2}{\rv{x}_1} \mapsto \Span{\rvVec{x}}$.
\end{proof}

\begin{definition}[${\rm F}(2,2)$ distribution~{\cite[Chap.~27]{Johnson1995continuous}}] \label{def:Fisher}
	The probability distribution with pdf $f(x) = \frac{1}{(1+x)^2}$ and CDF $F(x) = \frac{x}{1+x}$, $x\ge 0$, is called the ${\rm F}$-distribution with $(2,2)$ DoF denoted by ${\rm F}(2,2)$. ${\rm F}(2,2)$ is the distribution of the ratio of two independent chi-square random variables with $2$ DoF.
\end{definition}

\begin{lemma} \label{lem:Cauchy-Fisher} 
	If $\rv{x} = \rv{r} e^{\jmath \uptheta}$ where $\rv{r}$ and $\uptheta$ are independent, $\rv{r}^2 \sim {\rm F}(2,2)$, and $\uptheta$ is uniformly distributed in $[0,2\pi]$, then $\rv{x} \sim {\rm Cauchy}(0,1)$.
\end{lemma}
\begin{proof}
	Since $\rv{r}^2$ and $\uptheta$ are independent, $f_{\rv{r}^2, \uptheta}(r^2,\theta) = f_{\rv{r}^2}(r^2)f_{\uptheta}(\theta) = \frac{1}{(1+r^2)^2} \frac{1}{2\pi}$. 
	A change of variable from 
	$\rv{x}$ to $(\rv{r}^2, \uptheta)$ 
	yields $\frac{1}{2}f_{\rv{x}}(r e^{\jmath \theta}) = f_{\rv{r}^2, \uptheta}(r^2,\theta)$, so $f_{\rv{x}}(x) = \frac{1}{\pi(1+|x|^2)^2}$, which is ${\rm Cauchy}(0,1)$ pdf.
\end{proof}

\begin{lemma} \label{lem:Cauchy-Unif} 
	Let  $[\rv{a}_1\ \rv{a}_2]^T$ be uniformly distributed on $(0,1)^2$. Then $\rv{w} \defeq \Nc^{-1}(\rv{a}_1)+\im\Nc^{-1}(\rv{a}_2)$ follows a $\Cc\Nc(0,2)$ distribution and $\rv{t} \defeq \sqrt{\frac{1-\exp(-\frac{|\rv{w}|^2}{2})}{1+\exp(-\frac{|\rv{w}|^2}{2})}} \frac{\rv{w}}{|\rv{w}|}$ follows a ${\rm Cauchy}(0,1)$ distribution truncated on $\{t\in\mathbb{C}: \ |t|<1\}$.
\end{lemma}
\begin{proof}
	The Gaussianity of $\rv{w}$ is a standard result. Then, $|\rv{w}|^2$ is independent from $\frac{\rv{w}}{|\rv{w}|}$ (thus $|\rv{t}|^2$ is independent from $\frac{\rv{t}}{|\rv{t}|}$) and is chi-square distributed with $2$ DoF, and $1-\exp(-\frac{|\rv{w}|^2}{2})$ follows a uniform distribution on $(0,1)$. Then, $\frac{1-\exp(-\frac{|\rv{w}|^2}{2})}{1+\exp(-\frac{|\rv{w}|^2}{2})}$ has CDF $F(x)=\frac{2x}{1+x}, x \in (0,1)$. Using Definition~\ref{def:Fisher}, we see that  $\frac{1-\exp(-\frac{|\rv{w}|^2}{2})}{1+\exp(-\frac{|\rv{w}|^2}{2})}$ follows a ${\rm F}(2,2)$ distribution truncated on $(0,1)$.
	Finally, using Lemma~\ref{lem:Cauchy-Fisher}, we conclude that $\rv{t}$ follows a ${\rm Cauchy}(0,1)$ distribution truncated on $\{t\in\mathbb{C}: \ |t|<1\}$.
\end{proof}

\subsection{Proof of Lemma~\ref{lem:mindist_CS(T,1)}} \label{proof:mindist_CS(T,1)}
If $B_1 = \dots = B_{2(T-1)} = 1$, we have $A_j = \big\{\frac{1}{4},\frac{3}{4}\big\}, \forall j\in [2(T-1)]$. Let $m \defeq \Nc^{-1}(\frac{3}{4}) = -\Nc^{-1}(\frac{1}{4})$, then $w_j \in \{ \pm m \pm \im m\}$ and $|w_j| = m\sqrt{2}$. Then 
$
t_j
= \sqrt{c} q_j, j\in [T-1],
$
where $c\defeq \frac{1-e^{-m^2}}{1+e^{-m^2}}$ and $q_j \in \big\{\pm\frac{1}{\sqrt{2}} \pm \im \frac{1}{\sqrt{2}} \big\}$ is a $4$-QAM symbol with unit power. Substituting $t_j$ in \eqref{eq:encodeMapping}, a constellation symbol can be written simply as
%
$
\xv(\qv) = \big(c^{-1}+T-1\big)^{-1/2}\big[q_1 \, \dots\, q_{i-1} \ c^{-1/2} \ q_{i}\, \dots\, q_{T-1}\big]^\T.
$
Consider another symbol $\bar{\xv}(\bar{\qv}) \ne \xv(\qv)$. 
\subsubsection{If $\xv$ and $\bar{\xv}$ are in the same cell} The correlation between $\xv$ and $\bar{\xv}$ is $\xv^\H\bar{\xv} = (c^{-1}+T-1)^{-1}\big(c^{-1} + \sum_{i=j}^{T-1}q_j^*\bar{q}_j\big)$.
Notice that $q_j^*\bar{q}_j \in \{\pm 1, \pm \im\}$, we denote by $n_{a}$ the number of terms $q_j^*\bar{q}_j$ having value $a \in \{\pm 1, \pm \im\}$. We have that $n_1 + n_{-1} + n_\im + n_{-\im} = T-1$ and $n_1 < T-1$. We would like to find $\{n_a\}$ that maximize $|\xv^\H\bar{\xv}|^2 = (c^{-1}+T-1)^{-2} \big[(c^{-1} + n_1 - n_{-1})^2 + (n_\im - n_{-\im})^2 \big]$. The optimal $\{n_a\}$ must satisfy $\{n_1 = 0$ or $n_{-1} = 0\}$ and $\{n_\im = 0$ or $n_{-\im} = 0\}$ since otherwise, there always exists other $\{n_a\}$ that increases $|\xv^\H\bar{\xv}|^2$. 
Specifically, the optimal $\{n_a\}$ falls into one of two cases: $\{n_\im = 0 ~\text{or}~ n_{-\im} = 0; n_{-1} = 0\}$ or $\{n_\im = 0~\text{or}~n_{-\im} = 0; n_{1} = 0; n_{-1}\ge c^{-1}\}$. By inspecting these cases, we find that the maximal value of $|\xv^\H\bar{\xv}|^2$ is $(c^{-1}+T-1)^{-2} [(c^{-1}+T-2)^2 + 1]$ achieved with $n_1 = T-2,(n_\im;n_{-\im}) \in \{(0;1),(1;0)\}$.

\subsubsection{If $\xv$ is in cell $S_i$ and $\bar{\xv}$ is in cell $S_{\bar{i}}$ with $\bar{i} \ne i$} We denote $\vect{r} = \big[q_1 \, \dots\, q_{i-1} \ c^{-1/2} \ q_{i}\, \dots\, q_{T-1}\big]^\T$ and $\bar{\vect{r}} = \big[\bar{q}_1 \, \dots\, \bar{q}_{\bar{i}-1} \ c^{-1/2} \ \bar{q}_{\bar{i}}\, \dots\, \bar{q}_{T-1}\big]^\T$. Then $\xv^\H\bar{\xv} = (c^{-1}+T-1)^{-1} \big[c^{-1/2}(r_{\bar{i}} + \bar{r}_{i}) + \sum_{j\in [T]\setminus\{i,\bar{i}\}}r_j^*\bar{r}_j\big]$. Observe that $r_{\bar{i}} + \bar{r}_{i} \in \{0,\pm \sqrt{2},\pm \im\sqrt{2},\pm \sqrt{2}\pm \im\sqrt{2}\}$. By looking at each value of $r_{\bar{i}} + \bar{r}_{i}$ and inspecting $\{n_{a}\}$ as done in the previous case, we find that the maximal value of $|\xv^\H\bar{\xv}|^2$ is: $(c^{-1}+T-1)^{-2}(T-2)^2$ if $r_{\bar{i}} + \bar{r}_{i} = 0$; $(c^{-1}+T-1)^{-2}(\sqrt{2c^{-1}} + T-2)^2$ if $r_{\bar{i}} + \bar{r}_{i} \in \{\pm \sqrt{2},\pm \im\sqrt{2} \}$; and $(c^{-1}+T-1)^{-2}(\sqrt{2c^{-1}} + T-2)^2 + 2c^{-1}$ if $r_{\bar{i}} + \bar{r}_{i} \in \{\pm \sqrt{2} \pm \im\sqrt{2} \}$. The maximal value of $|\xv^\H\bar{\xv}|^2$ among these is $(c^{-1}+T-1)^{-2}(\sqrt{2c^{-1}} + T-2)^2 + 2c^{-1}$.

Comparing the above two cases, we conclude that the overall maximal value of $|\xv^\H\bar{\xv}|^2$ is $(c^{-1}+T-1)^{-2} [(c^{-1}+T-2)^2 + 1] = \big|1-\frac{1+\im}{c^{-1} + T-1}\big|^2$ attained with $n_1 = T-2$ and $(n_\im;n_{-\im}) \in \{(0;1),(1;0)\}$, which translates to \eqref{eq:mindist_coor_1} and \eqref{eq:mindist_coor_2} with $B_0 = 1$.

\subsection{Proof of Proposition~\ref{prop:Pe,b0=1}} \label{app:PeBo1}
With $CS(T,1)$, the transmitted signal is 
$
\rvVec{x} = \left(c^{-1}+T-1\right)^{-1/2}\left[\rv{q}_1  \dots \rv{q}_{i-1} \ c^{-1/2} \ \rv{q}_{i} \dots \rv{q}_{T-1}\right]^\T,
$
where $c = \frac{1-e^{-m^2}}{1+e^{-m^2}}$, $m = \Nc^{-1}(\frac{3}{4})$, and $\rv{q}_j \in \big\{\pm \frac{1}{\sqrt{2}} \pm \im\frac{1}{\sqrt{2}}\big\}$, $j\in [T-1]$ (see Appendix~\ref{proof:mindist_CS(T,1)}).
The received symbols are $\rv{y}_{i} = \sqrt{\rho_0} \rv{h} + \rv{z}_{i}$, and $\rv{y}_j = \sqrt{c\rho_0} \rv{q}_l \rv{h} +\rv{z}_j$ for $l = j-\mathbbm{1}\{j\le i \}$ if $j\ne i$, 
where $\rv{h}\sim\Cc\Nc(0,1)$, $\rv{z}_j\sim\Cc\Nc(0,1), j \in [T]$, and $\rho_0 = \frac{\rho T}{1+(T-1)c}$. 

\subsubsection{Cell error probability} 
A cell error occurs if $\exists j \ne i$ s.t. $|\rv{y}_j| > |\rv{y}_{i}|$.\footnote{Without the additive noise $\rvVec{z}$, $|\rv{y}_{i}|^2 = \frac{1+e^{-m^2}}{1-e^{-m^2}}|\rv{y}_{j}|^2 > |\rv{y}_{j}|^2$ for all $j\ne i$, and therefore, there is no cell error.} 
Given $\rv{h}$ and $\rv{y}_{i}$,
\begin{align}
&\Pr\{\hat{i} \ne i \cond \rv{h},\rv{y}_{i}\} \notag\\
&= \Pr\{\exists j \ne i: |\rv{y}_j|^2 > |\rv{y}_{i}|^2 \cond \rv{h},\rv{y}_{i}\} \\
&= 1-\textstyle\prod_{j\ne i} \Pr(|\rv{y}_j|^2 \le |\rv{y}_{i}|^2 \cond \rv{h},\rv{y}_{i}) \label{eq:tmp1089}\\
&= 1-(1-Q_1(\sqrt{2c\rho_0}|\rv{h}|,\sqrt{2}|\rv{y}_{i}|))^{T-1} \label{eq:tmp1090},
\end{align}
where $Q_1(.,.)$ is the Marcum Q-function with parameter 1. Here, \eqref{eq:tmp1089} holds because conditioned on $\rv{h}$ and $\rv{y}_{i}$, the events $|\rv{y}_j| \le |\rv{y}_{i}|$ are mutually independent for all $j \ne i$; \eqref{eq:tmp1090} is because given $\rv{h}$, the variables $2|\rv{y}_j|^2$ for $j \ne i$ are independently non-central chi-squared distributed with two DoF and non-centrality parameters $2c\rho_0 |\rv{h}|^2$, denoted by $\chi^2_2(2c\rho_0 |\rv{h}|^2)$. 

Next, by averaging $\Pr\{\hat{i} \ne i |\ \rv{h},\rv{y}_{i}\}$ over $|\rv{y}_{i}|^2$ and $|\rv{h}|^2$, taking into account that $|\rv{h}|^2$ is exponentially distributed with mean $1$, and given $\rv{h}$, $2|\rv{y}_{i}|^2 \sim \chi^2_2(2\rho_0 |\rv{h}|^2)$,  
we get
\begin{equation*}
\begin{split}
&\Pr\{\hat{i} \ne i\} \notag \\
&= 1 - \E_{|\rv{h}|^2} \E_{|\rv{y}_{i}|^2 | \rv{h}} \left[1-(1-Q_1(\sqrt{2c\rho_0}|\rv{h}|,\sqrt{2}|\rv{y}_{i}|))^{T-1}\right] \\
&=1-\int_{0}^{\infty} \int_{0}^{\infty} \left[1-(1-Q_1(\sqrt{2c\rho_0}|\rv{h}|,\sqrt{2}|\rv{y}_{i}|))^{T-1}\right] \notag\\
&\quad \cdot \exp(-|\rv{y}_{i}|^2-(1+\rho_0) |\rv{h}|^2) I_0\left(2\sqrt{\rho_0} |\rv{y}_{i}| |\rv{h}|\right) \dif |\rv{y}_{i}|^2 \dif |\rv{h}|^2,
\end{split}
\end{equation*}
where $I_0(.)$ is the modified Bessel function of the first kind of order $0$. 
From this, a simple change of variables gives~\eqref{eq:SER_cell}.

\subsubsection{Coordinate error probability given correct cell detection}
We assume that the cell index $i$ has been correctly decoded, i.e., $\hat{i} = i$, and, without loss of generality, $i = T$. 
The decoding strategy for the coordinate bits is similar to a 4-QAM demapper on $\rvVec{t} = [\rv{t}_1 \ \dots \ \rv{t}_{T-1}] = \left[\frac{\rv{y}_1}{\rv{y}_T} \ \dots \ \frac{\rv{y}_{T-1}}{\rv{y}_T}\right]$. 
Given $\rv{q}_j = q_j$, we have $\rv{y}_j = \sqrt{c\rho_0}\ q_j \rv{h} + \rv{z}_j \sim \Cc\Nc(0,1+c\rho_0)$ for $j<T$, $\rv{y}_T = \sqrt{\rho_0} \rv{h} + \rv{z}_T \sim \Cc\Nc(0,1+\rho_0)$, and $\E[\rv{y}_j \rv{y}_T^*] = \sqrt{c} \rho_0 {q}_j$. {Thus, according to Lemma~\ref{lem:Cauchy-Gaussian}, conditioned on $\rv{q}_j = q_j$, $\rv{t}_j = \frac{\rv{y}_j}{\rv{y}_T}$ follows the ${\rm Cauchy}\Big(\frac{\sqrt{c}\rho_0 q_i}{(1+\rho_0)^2}, \big(\frac{1+c\rho_0}{1+\rho_0}\big)^2 - \frac{c\rho_0^2 |q_i|^2}{(1+\rho_0)^4}\Big)$ distribution with the pdf
\begin{equation} \label{eq:pdf_t1}
f_{\rv{t}_j | \rv{q}_j}(t|q_j) \defeq \frac{(1+\rho_0)^2(1+(c+1)\rho_0)}{\pi\big(1+(c+1)\rho_0 + \left|(1+\rho_0)t - \sqrt{c}\rho_0 {q}_j\right|^2 \big)^2}. 
\end{equation}
Furthermore, given that $\hat{i} = i = T$, we have $|\rv{y}_T| > |\rv{y}_j|$, $\forall j<T$. Therefore, the distribution of $\rv{t}_j$ is further truncated on the unit circle. The conditional pdf of $\rv{t}_j$ is given by
$
f_{\rv{t}_j|\rv{q}_j, \hat{i} = i = T}(t|q_j) = \frac{{f}_{t_j | \rv{q}_j}(t|q_j)}{\int_{|x|\le 1} {f}_{\rv{t}_j | \rv{q}_j}(x|q_j) \dif x}.\label{eq:pdf_t}
$
}
An error happens at $\rv{t}_j$ if $\Re(\rv{t}_{j}) \Re(\rv{q}_{j}) < 0$ or $\Im(\rv{t}_{j}) \Im(\rv{q}_{j}) < 0$.\footnote{As for the cell error, there is no coordinate error in the absence of additive noise since in this case, $\rv{t}_j = \sqrt{c}\rv{q}_j$.} 
Therefore,
\begin{equation}
\Pr\{\hat{\rv{q}}_j \ne \rv{q}_j \cond \rv{q}_j = q_j, \hat{i} = i = T\} = 1 - \frac{\int_{\Rc_j} {f}_{t_j | \rv{q}_j}(t|q_j) \dif t}{\int_{|t|\le 1} {f}_{\rv{t}_j | \rv{q}_j}(t|q_j) \dif t},
\end{equation}
where $\Rc_j \defeq \{t \in \CC : |t|\le 1, \Re(t) \Re(\rv{q}_{j}) > 0, \Im(t) \Im(\rv{q}_{j}) > 0\}$. Using the polar coordinate, we have that ${f}_{\rv{t}_j|\rv{q}_j}(t|q_j) \dif t = \tilde{f}(r,\theta,{q}_j) r \dif r \dif\theta$, where 
$\tilde{f}(r,\theta,{q}_j) \defeq {f}_{\rv{t}_j | \rv{q}_j}(re^{\im \theta}|q_j)$. 
Then
\begin{multline}
\Pr\{\hat{\rv{q}}_j \ne \rv{q}_j \cond \rv{q}_j = q_j, \hat{i} = i = T\} \\= 1 - \frac{\int_{0}^{1}\int_{\Theta_j} \tilde{f}(r,\theta,{q}_j) r \dif \theta \dif r}{\int_{0}^{1}\int_{0}^{2\pi} \tilde{f}(r,\theta,{q}_j) r \dif \theta \dif r}, \label{eq:tmp141}
\end{multline}
where $\Theta_j$ is $\left[0,\pi/2\right]$ if ${q}_j = \frac{1}{\sqrt{2}} + \im\frac{1}{\sqrt{2}}$, $\left[\pi/2,\pi\right]$ if ${q}_j = -\frac{1}{\sqrt{2}} + \im\frac{1}{\sqrt{2}}$, $\left[\pi,3\pi/2\right]$ if ${q}_j = -\frac{1}{\sqrt{2}} - \im\frac{1}{\sqrt{2}}$, and $\left[3\pi/2,2\pi\right]$ if ${q}_j = \frac{1}{\sqrt{2}} - \im\frac{1}{\sqrt{2}}$. After some manipulations, we obtain that 
\begin{equation}
\int_{0}^{1}\int_{0}^{2\pi} \! \tilde{f}(r,\theta,{q}_j) r \dif\theta \dif r = \frac{1}{2}+\frac{(1-c)\rho_0}{2\sqrt{(2+(1+c)\rho_0)^2-4c\rho_0^2}},
\end{equation}
and 
for all ${q}_j \in \big\{\pm\frac{1}{\sqrt{2}} \pm \im\frac{1}{\sqrt{2}}\big\}$,
\begin{align}
&\int_{0}^{1} \int_{\Theta_j} \tilde{f}(r,\theta,{q}_j) r \dif \theta \dif r \notag \\
&\quad= \frac{1}{8}
+ \frac{\sqrt{2c}\rho_0\arccot\frac{1+\big(c-\sqrt{\frac{c}{2}}\big)\rho_0}{\sqrt{1+(c+1)\rho_0+\frac{c}{2}\rho_0^2}}}{2\pi\sqrt{1+(1+c)\rho_0+\frac{c}{2}\rho_0^2}} \notag \\
&\qquad+ \frac{(1-c)\rho_0\arccot\frac{2+\big(1-2\sqrt{2c}+c\big)\rho_0}{\sqrt{(2+(1+c)\rho_0)^2-4c\rho_0^2}}}{2\pi\sqrt{(2+(1+c)\rho_0)^2-4c\rho_0^2}}.
\end{align}
Substituting this in~\eqref{eq:tmp141}, and {using the fact that $\Pr\{\hat{\rv{q}}_j \ne \rv{q}_j \cond \rv{q}_j, \hat{i} = i = i_1\} = \Pr\{\hat{\rv{q}}_j \ne \rv{q}_j \cond \rv{q}_j, \hat{i} = i = i_2\}$, for all $i_1,i_2 \in [T]$,} we obtain~\eqref{eq:SER_coor_cond}. The union bound for $P_e$ follows readily. 
 
\subsection{Proof of Corollary~\ref{coro:Pe}} \label{app:PeT2Bo1}
The conditional pdf $f_{\rv{t}|\rv{q}_1}(t|q_1)$ of $\rv{t} = \frac{\rv{y}_{2-i}}{\rv{y}_{i}}$ is given in \eqref{eq:pdf_t1}. The conditional cell error is simply $\Pr\{\hat{i} \ne i| \rv{q}_1 = q_1\} = 1-\Pr\{|\rv{t}|<1|\rv{q}_1 = q_1\} = 1-\int_{|t|\le 1} f_{\rv{t}|\rv{q}_1}(t|q_1) \dif t$. %
By calculating the integral using polar coordinate (which results in the same value for any $q_1 \in \big\{\pm \frac{1}{\sqrt{2}}\pm \im\frac{1}{\sqrt{2}}\big\}$) and averaging over $\rv{q}_1$, we obtain~\eqref{eq:PcellB01T2}. Furthermore, when $T=2$, the union bound~\eqref{eq:SER_coor} 
is tight.

\bibliographystyle{IEEEtran}
\bibliography{IEEEabrv,./biblio}

\end{document}